\documentclass[aps,prl,reprint,superscriptaddress]{revtex4-2}
\usepackage[table,xcdraw]{xcolor}

\newcommand{\blue}[1]{\textcolor{blue}{#1}}
\newcommand{\red}[1]{\textcolor{red}{#1}}
\usepackage[utf8]{inputenc}
\usepackage{amsmath}
\usepackage{amsfonts}
\usepackage{amssymb,amsthm}
\usepackage{bm}

\usepackage{textcomp}
\usepackage[subrefformat=parens,labelformat=parens,caption=false]{subfig}
\usepackage[normalem]{ulem}
\usepackage{stackengine,graphicx}

\newtheorem{theorem}{Theorem}
\newtheorem{lemma}{Lemma}

\newtheorem{definition}{Definition}

\newtheorem{corollary}{Corollary}[theorem]
\usepackage{multirow}

\usepackage{array}
\usepackage{bibunits}
\usepackage{chngcntr}

\newtheorem{smtheorem}{Theorem}[section]
\newtheorem{smdefinition}[smtheorem]{Definition}
\newtheorem{smlemma}[smtheorem]{Lemma}
\newtheorem{smcorollary}{Corollary}[smtheorem]

\usepackage[]{hyperref}
\usepackage{bbold}
\usepackage{cancel}
\newcommand{\stkout}[1]{\ifmmode\text{\sout{\ensuremath{#1}}}\else\sout{#1}\fi}

\newcommand{\eq}[1]{\begin{align}#1\end{align}}

\hypersetup{
  colorlinks = true,
  urlcolor = blue,
  linkcolor = blue,
  citecolor = blue
}
\usepackage{braket}
\usepackage[all]{hypcap}
\usepackage{url}
\usepackage{booktabs}

\usepackage{mathtools}
\usepackage[shortlabels]{enumitem}

\usepackage{adjustbox}
\graphicspath {{figures/}}

\begin{abstract}
Universality in physics describes how disparate systems can exhibit identical low-energy behavior.
Here, we reveal a rich landscape of new universal scattering phenomena governed by the interplay between an interaction and a system's density of states.
We investigate one-dimensional scattering with dispersion relations $\epsilon(k)=|k|^m$ and $\epsilon(k)=\operatorname{sgn}(k)|k|^m$ for any real $m>1$.
For emitter scattering and separable potentials, including contact potentials, we classify the zero-energy $S$-matrix limits.
These universal $S$-matrix fixed points depend solely on the dispersion exponent $m$ and a finite set of integers characterizing the low-momentum interaction structure.
This classification extends far beyond the standard quadratic-dispersion case ($m=2$).
We further derive a generalized Levinson's theorem relating the total winding of the $S$-matrix determinant to the number of bound states.
Finally, for contact potentials and analytic monomial dispersions, we analyze the leading low-energy expansion of the scattering phase near these fixed points and extract the corresponding Wigner-delay scaling.
Our findings are directly relevant to synthetic quantum systems, where engineered dispersion relations in atomic arrays and photonic crystals offer a platform to explore these universal behaviors.
\end{abstract}

\begin{document}

\title{Landscape of scattering universality with general dispersion relations}
\author{Yidan Wang}
\affiliation{Department of Physics, Harvard University, Cambridge, Massachusetts 02138, USA}
\author{Xuesen Na}
\affiliation{Department of Mathematics, University of Illinois Urbana-Champaign, Urbana, IL 61801}
\author{Michael J. Gullans}
\affiliation{Joint Center for Quantum Information and Computer Science, University of Maryland and NIST, College Park, Maryland 20742, USA}
\author{Susanne F. Yelin}
\affiliation{Department of Physics, Harvard University, Cambridge, Massachusetts 02138, USA}
\author{Alexey V. Gorshkov}
\affiliation{Joint Center for Quantum Information and Computer Science, University of Maryland and NIST, College Park, Maryland 20742, USA}
\affiliation{Joint Quantum Institute, University of Maryland and NIST, College Park, Maryland 20742, USA}

\maketitle

\begin{bibunit}[apsrev4-2]

A central pillar of modern physics is universality, the principle that systems with vastly different microscopic details can exhibit identical long-wavelength behavior. Identifying and categorizing such universal behaviors is essential across many fields, from the scaling laws of critical phenomena in statistical mechanics \cite{WilsonRMP, Kadanoff1966} to the low-energy dynamics of quantum field theories \cite{WeinbergEFT, PeskinSchroederQFT}.

\begin{figure}
\includegraphics[width=1\linewidth]{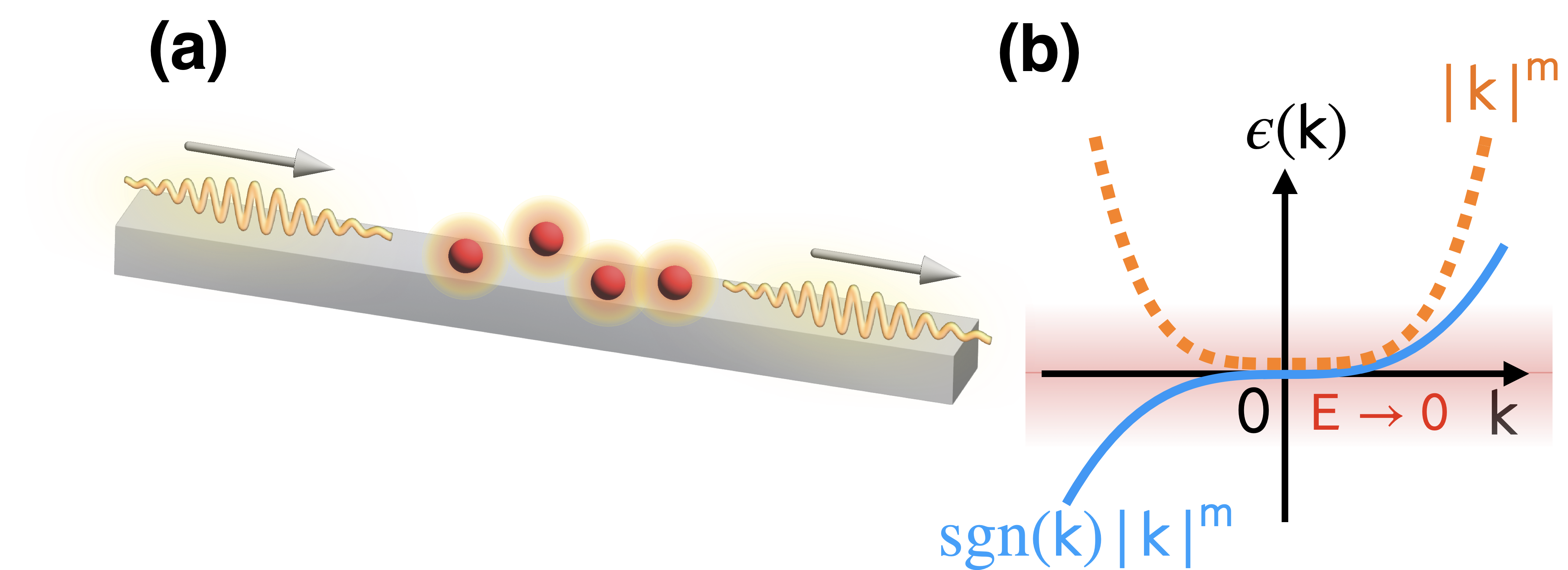}
\caption{\label{fig:disp_scatt}
(a)~Schematic of quantum emitters coupled to a one-dimensional waveguide.
(b)~The general dispersion relations studied: antisymmetric, $\epsilon(k)=\operatorname{sgn}(k)|k|^m$ (solid blue), and symmetric, $\epsilon(k)=|k|^m$ (dashed orange). For $m>1$, the group velocity vanishes at $k=0$, causing the density of states to diverge. The red shaded region indicates the low-energy limit ($E \to 0$) where this universal behavior emerges.
}
\end{figure}

This concept plays an equally crucial role in few-body quantum scattering. At low energies, massive particles are governed by a quadratic dispersion, \(\epsilon(k) \propto k^2\). As the scattering energy approaches zero, the scattering behavior  becomes independent of the intricate details of the interaction potential \cite{LandauLifshitzQM, TaylorScattering, Wigner1948ThresholdCrossSections, Sadeghpour2000CollisionsNearThreshold}. This allows low-energy scattering to be universally described by just a few parameters, such as the scattering length, which dictates most low-energy properties \cite{Bedaque2002, Bethe1949}. This universal behavior, however, depends critically on the spatial dimension.
For generic short-range potentials, scattering approaches perfect reflection at zero energy in one dimension, whereas the $S$-matrix approaches the identity operator in three dimensions~\cite{Senn1988Threshold,TaylorScattering}.
These distinct behaviors are fundamentally linked to the density of states, \(\rho(E)\), which diverges at zero energy in one dimension but vanishes in three.

This motivates a central question: can novel universal scattering phenomena be engineered by controlling the divergence of the density of states, even while holding the spatial dimension fixed? Our recent work \cite{wang2022universal} affirmed this, demonstrating that a rich variety of universal S-matrix limits emerge from general dispersion relations. However, that study focused on a constrained class of interactions, leaving the robustness of these new universality classes an open question. In this Letter, we investigate scattering due to generic emitter interactions [see Fig.~\ref{fig:disp_scatt}(a)] and from separable potentials. We not only confirm the foundational principles of Ref.~\cite{wang2022universal} but also reveal a richer picture where the structure of the interaction is as crucial as the density of states. We find that the competition between kinetic and interaction energies gives rise to a vast landscape of universal behaviors, of which the findings in Ref.~\cite{wang2022universal} represent a special subset.
For quartic dispersion, we find three distinct nontrivial zero-energy limits, in contrast to the single nontrivial limit of perfect reflection for quadratic dispersion [see Fig.~\ref{fig:k2k4}].

\begin{figure}
    \centering
\includegraphics[width=0.95\linewidth]{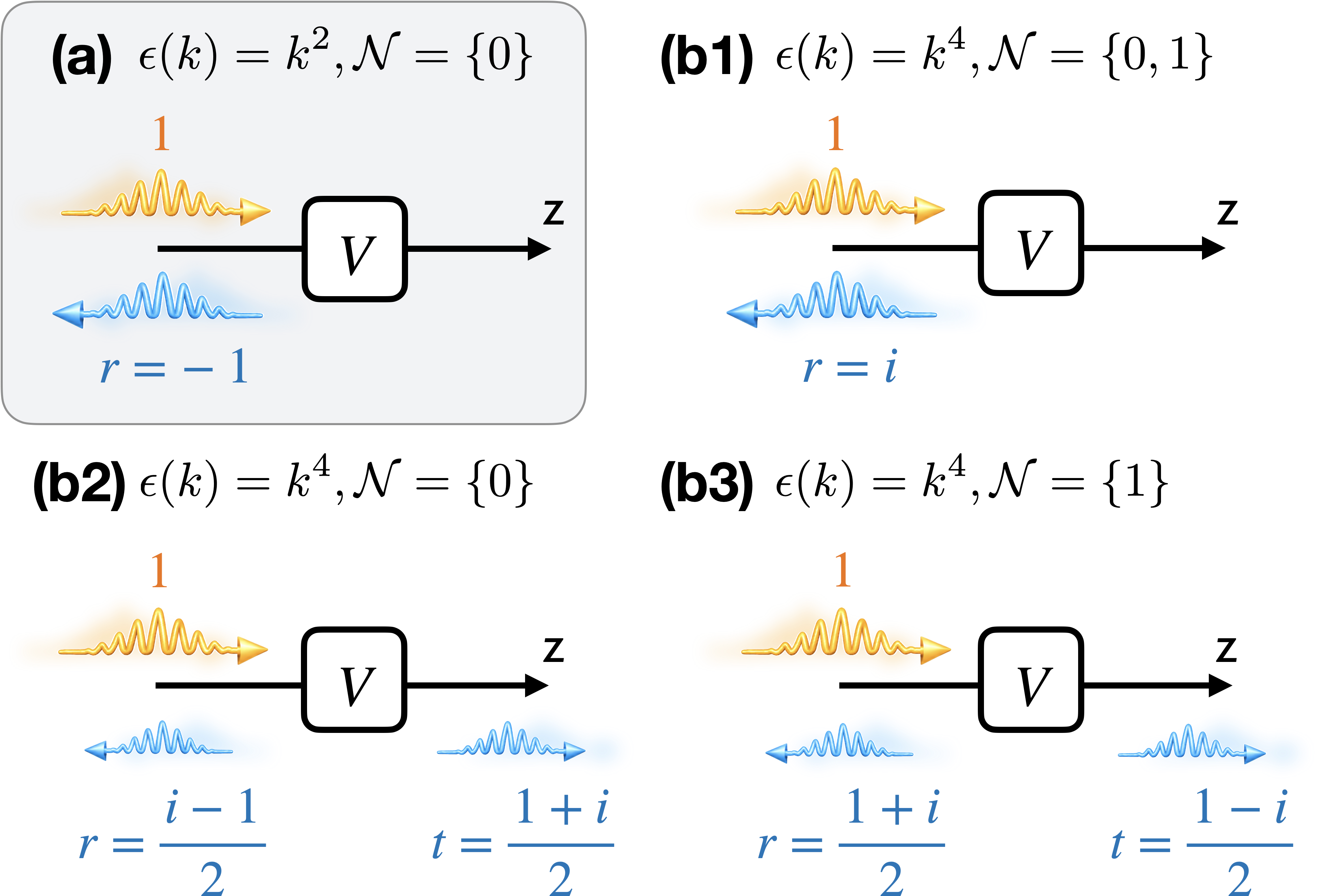}
    \caption{Scattering in the zero-energy limit.
(a) For $\epsilon(k)=k^2$, generic interactions give perfect reflection with amplitude $-1$.
(b1)--(b3) For $\epsilon(k)=k^4$,  nontrivial low-energy scattering falls into three universality classes labeled by the integer set $\mathcal{N}$.
 (b1) Class \(\mathcal N=\{0,1\}\) yields perfect reflection with amplitude $\mathrm{i}$. (b2)--(b3) Classes \(\mathcal N=\{0\}\) and \(\mathcal N=\{1\}\) both yield equal reflection and transmission probabilities of $1/2$, but distinct complex amplitudes.
Yellow and blue waves denote incoming and outgoing waves, respectively, for unit incident amplitude.
}
    \label{fig:k2k4}
\end{figure}

The study of how higher-order dispersions affect low-energy scattering is motivated by growing experimental capabilities to engineer dispersion relations across diverse platforms.
These include photonic crystal waveguides \cite{Burr2013Degenerate,Goban2014,Hood2016,Burr2016Experimental}, photonic metastructures \cite{Cotrufo2023Dispersion,Pearson2026DispersionEngineered}, rhombohedral multilayer graphene \cite{Zhang2010ABCGraphene,Zhang2011ABCExperiment}, ultracold atoms in shaken optical lattices \cite{Parker2013,Logan2016}, superconducting-qubit arrays \cite{Karamlou2022TightBinding,Rosen2025FlatBand}, and trapped-ion chains \cite{porras2004effective,debnath2018observation}. 

The Hamiltonian for one-dimensional scattering is
\begin{equation}
H = \int_{-\infty}^{+\infty} dk\, \epsilon(k) C^\dagger(k) C(k) + V. \label{eqHamil}
\end{equation}
In Eq.~\eqref{eqHamil}, $\epsilon(k)$ is the dispersion relation and $V$ describes interaction with localized scatterers or potentials. $C^\dagger(k)$ creates a propagating particle  with momentum $k$.
We restrict to the single-excitation sector, where the scattering equations are independent of the operator statistics and apply equally to classical linear-wave scattering.

Studying this Hamiltonian for general $\epsilon(k)$ and $V$ reveals a new landscape of universal behaviors.
First, we consider the dispersion relation. We investigate two general forms: $\epsilon(k) = |k|^m$ (symmetric) and $\epsilon(k) = \operatorname{sgn}(k)|k|^m$ (antisymmetric), where $m > 1$ is any real exponent [see Fig.~\ref{fig:disp_scatt}(b)]. These serve as local (in $k$-space) approximations to realistic band structures: near a band edge, an analytic band structure is captured by the symmetric form with even $m$, while near a stationary inflection point, an analytic band structure is captured by the antisymmetric form with odd $m$. The zero-energy classification holds for all real $m>1$ and is insensitive to dispersion details away from $E=0$. For both cases, the density of states, $\rho(E) \propto |E|^{1/m-1}$, diverges as $E \to 0$.

Second, we consider the interaction term $V$, focusing on two closely related and analytically solvable classes: emitter interaction and separable-potential interaction. Emitter interaction is common in synthetic quantum systems [see Fig.~\ref{fig:disp_scatt}(a)], describing, for example, photon absorption and re-emission by two-level atoms.
Here, \(C^\dagger(k)\) creates a propagating photon with momentum \(k\).
For a system with \(N\) emitters, let \(b_a^\dagger\) create an excitation of emitter \(a\). The interaction term is
\begin{equation}
\begin{aligned}
V ={}& \sum_{a,a'=1}^N K^R_{aa'} b_a^\dagger b_{a'} \\
&+\int_{-\infty}^{+\infty}dk\left[\sum_{a=1}^Nv_a(k)C(k)b_a^\dagger+\text{h.c.}\right].
\end{aligned}
\label{eqEmitterInteraction}
\end{equation}

In Eq.~\eqref{eqEmitterInteraction}, the complex functions \(v_a(k)\) are the emitter coupling coefficients. We collect them into the vector $|v(k)\rangle=[v_1(k),\ldots,v_N(k)]^T$.    The \(N \times N\) emitter-interaction matrix \(\boldsymbol{K}^R\) (with complex matrix elements $K^R_{aa'}$) describes coherent and incoherent interactions among the emitters induced by their environment.  We assume that all zero-energy emitter states decoupled from scattering have been projected out, so that $\bm K^R$ has no null vector orthogonal to $|v(k)\rangle$ for all $k$ [see Sec.~II of the SM~\cite{Wang2025supplement}]. The single-excitation dynamics generated by this Hamiltonian is identical to classical wave scattering from linear electric dipoles in a one-dimensional medium, so the same model describes both single-photon and classical-wave scattering.

Separable potentials provide effective representations of scattering from local two-body potentials~\cite{haidenbauer1983separable}. 
Their interaction is generally nonlocal in real space and takes the form 
\eq{
V &= \int_{-\infty}^{+\infty} dk\,dk' \, \sum_{a,a'=1}^N v^*_a(k') \, U_{aa'} \, v_{a'}(k) \, C^\dagger(k')C(k).\label{eqSepPoten0}
}
In Eq.~\eqref{eqSepPoten0}, $N$ now denotes the number of form factors $v_a(k)$, and $U_{aa'}$ are the matrix elements of the $N\times N$  coupling matrix $\bm{U}$.
In this work, we restrict to invertible \(\bm U\).

Contact (zero-range) potentials form an important subclass
of separable potentials and are used in effective field
theories to describe low-energy scattering from short-range
local potentials.
For this subclass, we take monomial form factors
$v_a(k)=k^{n_a}$, where the $n_a$ are distinct nonnegative
integers. In real space, these form factors correspond to
derivatives of the delta function.

In general, the emitter-interaction matrix $\bm{K}^R$ (or the form-factor coupling matrix $\bm{U}$) may be non-Hermitian, describing loss from or gain into the 1D channel. Consequently, the resulting $S$-matrix can be nonunitary.

Throughout, we use the same functions $v_a(k)$ to parametrize the emitter couplings and the form factors of the separable potential. When referring to the two models collectively, \(N\) denotes the number of interaction components, namely emitters or form factors. This unified notation highlights the close structural similarity between the $S$-matrix expressions for the two interaction models (see SM~\cite{Wang2025supplement}).

We assume that $v_a(k)$ are analytic at $k=0$.
For ultraviolet convergence, we further require
$v_a(k)=O(|k|^{m/2-1})$ as $|k|\to\infty$,
except for contact interactions, whose monomial form factors
$v_a(k)=k^{n_a}$ are ultraviolet regularized when necessary.

In the analysis below, we focus on noncritical systems.  A system is critical if arbitrarily small complex perturbations of $\bm K^R$ or $-\bm U^{-1}$ can generate bound states arbitrarily close to zero energy (see Sec.~II of the SM~\cite{Wang2025supplement}).

\begin{figure}
    \centering
\includegraphics[width=1\linewidth]{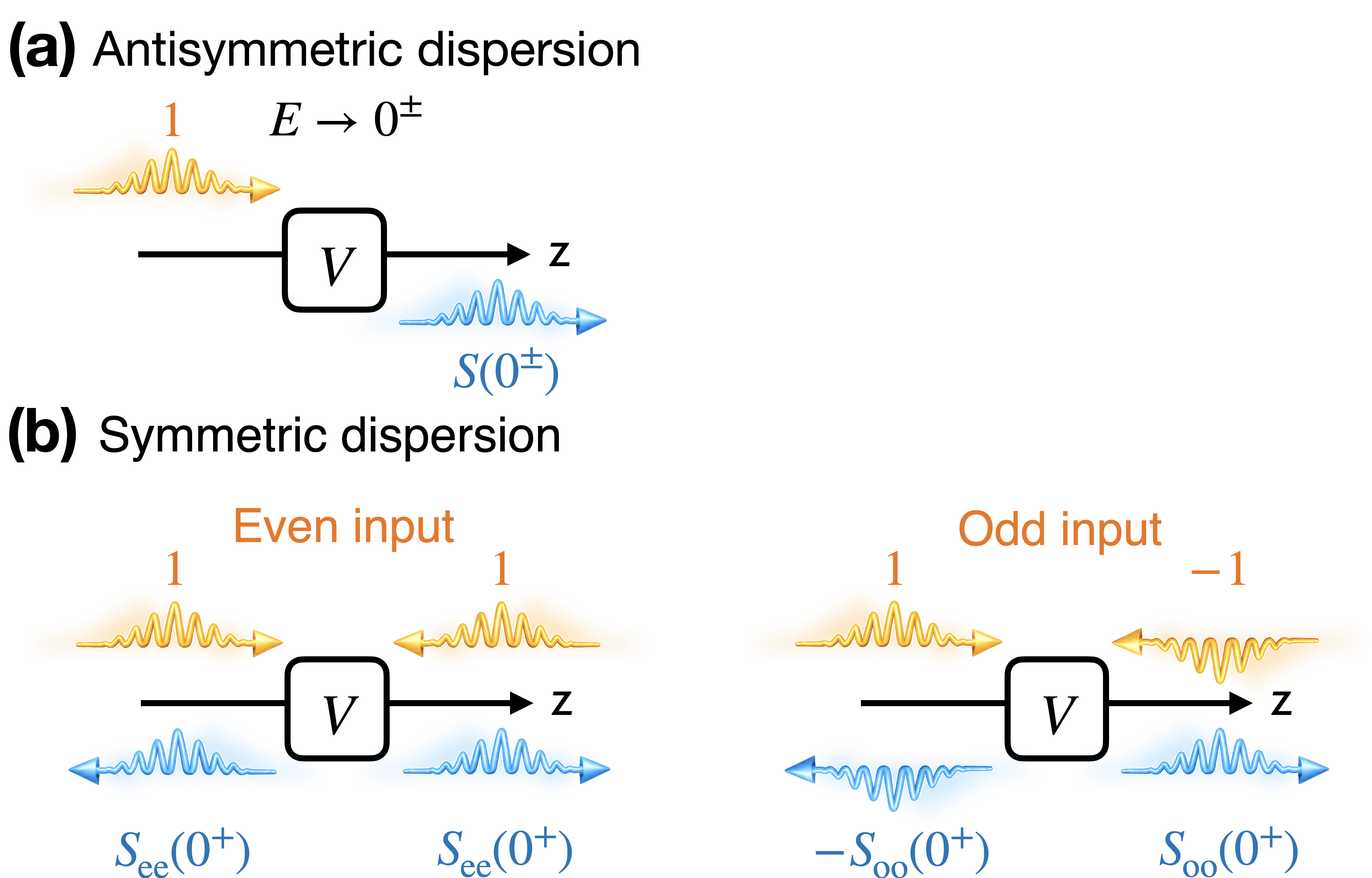}
    \caption{
Schematics for scattering with (a) antisymmetric and (b) symmetric dispersions. Yellow and blue waves denote incoming and outgoing waves, respectively, while \(1\) and \(-1\) denote relative incident amplitudes. (a) For the antisymmetric dispersion, both \(E\to0^\pm\) correspond to incidence from the left, with transmission amplitudes \(S(0^\pm)\). Their universal behaviors are given in Theorem~\ref{theoremUniScatteringAnti}. (b) For the symmetric dispersion, the even and odd incident states have relative amplitudes \((1,1)\) and \((1,-1)\), respectively. At zero energy, the two parity channels decouple, with \(S\)-matrix eigenvalues \(S_{\rm ee}(0^+)\) and \(S_{\rm oo}(0^+)\), whose universal values are given in Theorem~\ref{theoremUniScatteringSym}.
}
        \label{fig:scattering_illu}
\end{figure}

\emph{Physical picture for the antisymmetric dispersion.}--- For antisymmetric dispersion, the S-matrix reduces to a single complex transmission amplitude $S(E)$, with $|S(E)|^2$ giving the transmitted-to-incident flux ratio.  The phase of $S(E)$ is the phase acquired by the transmitted wave relative to free propagation. A representative example is a chiral photonic crystal waveguide with a stationary inflection point at $k=0$. As $E\rightarrow 0^\pm$,
the group velocity drops to zero and the limiting  transmission amplitudes $S(0^\pm)$  are the focus of our study [Fig.~\ref{fig:scattering_illu}(a)].

The behavior of $S(0^\pm)$ falls into one of three cases:
(i) \emph{interior} interactions, where $S(0^\pm)$ are universal, unitary, and complex conjugates, $S(0^+)=S(0^-)^*$,
(ii) \emph{borderline} interactions (arising only when $m$ is an odd integer), where the individual limits $S(0^\pm)$ are nonuniversal and may be nonunitary, yet their ratio $S(0^+)/S(0^-)$ remains unitary and universal, and
(iii) \emph{trivial} interactions, where the scattering center becomes transparent: $S(0^\pm)=1$.
Here, we call a quantity \emph{universal} if it depends solely on $m$ and the local behavior of $|v(k)\rangle$ near $k=0$, remaining insensitive to microscopic details such as $\bm{K}^R$ ($\bm{U}$) or $|v(k)\rangle$ away from $k=0$.
In the following, we make these statements precise, starting from the fundamental case $N=1$ (single emitter or form factor) and generalizing to arbitrary $N$. The symmetric-dispersion case is discussed afterward.

\emph{Single-component case.}--- For $N=1$, we denote the single coupling coefficient (or form factor) by $v(k)\equiv v_1(k)$. The matrices reduce to complex scalars: $\bm{K}^R\to K^R$ for emitter scattering  and $\bm{U}\to U$ for separable-potential scattering. For a decaying emitter, $\operatorname{Re}K^R$ gives its energy shift and $-2\operatorname{Im}K^R$ gives its incoherent decay linewidth.

For antisymmetric dispersion, the transmission amplitude $S(E)$ is related to the Jost function $J(\omega)$ via
\begin{equation}
S(E) = \frac{J(E-\mathrm{i}0)}{J(E+\mathrm{i}0)}, \label{eqdetSSep0}
\end{equation}
where $J(\omega) = K(\omega) + K^R - \omega$ for emitter scattering and $J(\omega) = K(\omega) - U^{-1}$ for separable-potential scattering. Here, $K(\omega)$ is the projected resolvent function
\begin{equation}
K(\omega) = \int_{-\infty}^{+\infty} dk\,\frac{|v(k)|^2}{\omega - \epsilon(k)}. \label{eqKdefiN1}
\end{equation}
The integral in Eq.~\eqref{eqKdefiN1} is analytic in the complex $\omega$-plane away from the continuum spectrum $\mathbb R$.
For emitter scattering, the real and imaginary parts of $K(E-\mathrm{i}0)=K(E+\mathrm{i}0)^*$ represent the energy-dependent shift and half the decay linewidth induced by the guided photon channel, respectively. Zeros of $J(\omega)$ outside the continuum spectrum correspond to bound-state energies.

Suppose the emitter coupling or form factor admits the low-momentum expansion $v(k)=ck^n+o(k^n)$, where $c\neq0$ and $n\ge0$ is an integer. We refer to $n$ as the \emph{low-energy coupling exponent}.
When \(n=0\), \(v(0)\) is a nonzero constant, a case explored in Ref.~\cite{wang2022universal}. For \(n \ge 1\), the exponent indicates how quickly the interaction vanishes at zero momentum, thereby determining the behavior of $K(\omega)$ near $\omega=0$. We define a \emph{borderline coupling exponent} \(\bar{n} = (m-1)/2\), which sets the threshold for infrared divergence in $K(\omega)$.

The physics is governed by the competition between the interaction (via \(n\)) and the kinetic energy (via \(\bar{n}\)).
When \(n<\bar{n}\), we define the interaction to be \emph{interior}. In this case, \(K(\omega)\) diverges as \(\omega\to0\), and the zero-energy S-matrix limits are
\begin{equation}
S(0^\pm)
=
\exp\left(\pm \pi\mathrm{i} \frac{2n+1}{m}\right).
\label{eqS0pmN1}
\end{equation}
The limits in Eq.~\eqref{eqS0pmN1} depend only on \(m\) and \(n\), and are therefore universal. When \(n=\bar{n}\), we define the interaction to be \emph{borderline}. In this case, \(K(\omega)\) has finite boundary values as \(\omega\to0\) from the upper- and lower-half complex planes, but no common limit. The zero-energy S-matrix limits satisfy \(S(0^+)=S(0^-)\), but their common value is nonuniversal. When \(n>\bar{n}\), we define the interaction to be \emph{trivial}. In this case, \(K(\omega)\) converges as \(\omega\to0\), and the scatterer becomes transparent at zero energy, with \(S(0^\pm)=1\).

\emph{Multicomponent case.}--- For general $N\geq 1$, the S-matrix retains the form of Eq.~\eqref{eqdetSSep0}, with the Jost function defined in Sec.~II of the SM~\cite{Wang2025supplement}.
Generalizing the fundamental $N=1$ case, the zero-energy behavior of $S(E)$ is controlled by a set of low-energy coupling exponents $\mathcal{N}=\{n_1, \dots, n_D\}$, with $D=|\mathcal N|\leq N$ and $n_i\leq\bar n$, which describe how $|v(k)\rangle$ vanishes in the low-energy limit in the emitter (or form-factor) Hilbert space. Lemma~\ref{lemmaIntegerSet} provides a precise definition of $\mathcal{N}$. Its construction is given in Sec.~I of the SM~\cite{Wang2025supplement}.

\begin{lemma}\label{lemmaIntegerSet}
There exists a unique, strictly increasing sequence of nonnegative integers \(\mathcal{N}=\{n_1, \dots, n_D\}\) with \(n_i \leq \bar{n}\), and a corresponding orthonormal set \(\{\ket{u_1}, \dots, \ket{u_D}\}\), such that

\begin{equation}
\ket{v(k)} = \sum_{i=1}^{D} \ket{u_i}\left[c_i k^{n_i}+O(k^{n_i+1})\right]+o(k^{\bar n}),
\label{eqCouplingVectorExpansion}
\end{equation}

where the complex coefficients \(c_i\) are nonzero.
\end{lemma}

In Eq.~\eqref{eqCouplingVectorExpansion}, every coupling exponent $n_i$ labels an orthogonal direction $|u_i\rangle$ in the emitter (or form-factor) Hilbert space along which coupling to the continuum produces infrared divergence ($n_i<\bar{n}$) or non-convergence ($n_i=\bar{n}$). 
The interaction is trivial if $\mathcal N=\varnothing$. Otherwise, it is interior if $n_D<\bar n$ and borderline if $n_D=\bar n$.

For contact interactions with monomial form factors, $\mathcal{N}$ follows directly from the monomial exponents. If $|v(k)\rangle=[k^{n_1},\ldots,k^{n_N}]^T$, with $0\leq n_1<\cdots<n_N$, then $\mathcal{N}=\{n_1,\ldots,n_D\}$, where $D$ is the number of exponents satisfying $n_i\leq\bar n$.
\footnote{Monomial form factors with $n_i\geq\bar n$ require an ultraviolet regularization, including a symmetric cutoff for the antisymmetric borderline integral (see SM~\cite{Wang2025supplement}).}

 The zero-energy limits of $S(E)$ for an arbitrary number $N$ of interaction components are given by the following theorem.

\begin{theorem}\label{theoremUniScatteringAnti}  Consider a noncritical system with antisymmetric dispersion $\epsilon(k)=\operatorname{sgn}(k)|k|^m$ with real $m>1$.
For interior and trivial interactions ($\bar n\notin\mathcal N$),
the zero-energy S-matrix limits are universal and satisfy
\begin{equation}
S(0^\pm)=
\exp\left[
\pm\frac{\pi\mathrm{i}}{m}\eta(\mathcal N)
\right],
\label{eqSasymIntro}
\end{equation}
where $\eta(\mathcal A)\equiv\sum_{n\in\mathcal A}(2n+1)$, with $\eta(\varnothing)=0$.
In particular, $\mathcal N=\varnothing$ gives the trivial limit $S(0^\pm)=1$.

For a borderline interaction, $\bar n\in\mathcal N$, the individual limits $S(0^\pm)$ are nonuniversal, while their ratio is universal
\begin{equation}
\frac{S(0^+)}{S(0^-)}=
\exp\left[
\frac{2\pi\mathrm{i}}{m}\eta(\mathcal N)
\right].
\label{eqSeeNeqDsgnratio0}
\end{equation}
\end{theorem}

For interior interactions, $S(0^\pm)$ is determined by
$\eta(\mathcal N)\bmod 2m$.
The allowed values for integer $m\geq 2$ are tabulated
in the SM~\cite{Wang2025supplement}.  For borderline interactions, which occur only for odd integer \(m\), the possible values of  $\eta(\mathcal N)\bmod m$ determining the universal ratio in Eq.~\eqref{eqSeeNeqDsgnratio0} are  $\mathbb Z_m\setminus\{2\}$ for $m=3, 5$, and $\mathbb Z_m$ for odd $m\ge 7$.

The proof of Theorem~\ref{theoremUniScatteringAnti} is detailed in the SM \cite{Wang2025supplement}.
In contrast to Theorem~\ref{theoremUniScatteringAnti} for $m>1$, when \(E\propto k\) (\(m=1\)), \(S(E)\) is continuous at \(E=0\) and its value is generally nonuniversal.

We expect generic short-range local potentials of finite width to belong
to the universality class
$\mathcal N=\{0,1,\ldots,\lfloor\bar n\rfloor\}$~\cite{upcoming26}.
This would allow their zero-energy scattering to be reproduced by
contact interactions with form factors
$|v(k)\rangle=(1,k,\ldots,k^{\lfloor\bar n\rfloor})^T$,
with the coupling matrix $\bm U$ fitted to match
low-energy scattering observables.

\emph{Symmetric dispersion.}--- We now turn to the symmetric dispersion $\epsilon(k)=|k|^m$. 
Each energy $E>0$ corresponds to two propagating states, one right-moving and one left-moving, with momenta $+k$ and $-k$, respectively, where $k=E^{1/m}$. 
We define the even- and odd-parity states as $|\psi_{\mathrm e}\rangle=(|k\rangle+|-k\rangle)/\sqrt{2}$ and $|\psi_{\mathrm o}\rangle=(|k\rangle-|-k\rangle)/\sqrt{2}$. In this basis, the $S$-matrix takes the form 
\begin{equation}
\bm{S}(E) =
\begin{pmatrix}
S_{\mathrm{ee}}(E) & S_{\mathrm{eo}}(E) \\
S_{\mathrm{oe}}(E) & S_{\mathrm{oo}}(E)
\end{pmatrix}.
\label{eqParitySMatrix}
\end{equation}
As $E\to0^+$,  the two parity channels decouple, and the $S$-matrix approaches a diagonal unitary matrix with universal eigenphases [see Fig.~\ref{fig:scattering_illu}(b)].
\begin{theorem} \label{theoremUniScatteringSym}
For a noncritical system with symmetric dispersion $\epsilon(k)=|k|^m$ ($m>1$),
\eq{
\bm{S}(0^+) =\left(\begin{matrix}
\exp\left[\frac{2\pi\mathrm{i} }{m}\eta(\mathcal{N}_{\mathrm{e}})\right] &  0\\
 0 &  \exp\left[\frac{2\pi\mathrm{i} }{m}\eta(\mathcal{N}_{\mathrm{o}})\right]
\end{matrix}\right),\label{eqSuniS}
}
where $\mathcal{N}_{\mathrm{e}}$ and $\mathcal{N}_{\mathrm{o}}$ denote the even and odd subsets of $\mathcal{N}$, respectively.
\end{theorem}
If $\mathcal N_p=\varnothing$ for $p\in\{\mathrm e,\mathrm o\}$, the parity channel $p$ is trivial and $S_{pp}(0^+)=1$.
The allowed values of $\eta(\mathcal N_{\mathrm e})$ and $\eta(\mathcal N_{\mathrm o})$ modulo $m$, which determine the universal phases $S_{\mathrm{ee}}(0^+)$ and $S_{\mathrm{oo}}(0^+)$, are tabulated in the SM~\cite{Wang2025supplement}.

Theorem~\ref{theoremUniScatteringSym} also determines the transmission and reflection amplitudes for a particle incident from either side of the 1D channel: $t(0^+)=[S_{\mathrm{ee}}(0^+)+S_{\mathrm{oo}}(0^+)]/2$ and $r(0^+)=[S_{\mathrm{ee}}(0^+)-S_{\mathrm{oo}}(0^+)]/2$.
The squared moduli $|t|^2$ and $|r|^2$ give the transmission
and reflection probabilities~\footnote{For amplifying
interactions, these quantities are more generally interpreted
as output-to-input flux ratios and need not be bounded by unity.},
while their phases encode the corresponding scattering phase shifts.
Both amplitudes are independent of the incidence direction. Figure~\ref{fig:k2k4} illustrates their zero-energy values for quadratic and quartic dispersions.

For borderline interactions ($\bar n\in\mathcal N$), the zero-energy $S$-matrix remains universal for symmetric dispersion, with the borderline exponent contributing a trivial phase factor $\exp[2\pi\mathrm{i}(2\bar n+1)/m]=1$. By contrast, the individual limits $S(0^\pm)$ are nonuniversal for antisymmetric dispersion. This difference arises because the borderline contribution to the projected resolvent diverges logarithmically for symmetric dispersion but approaches distinct finite limits in the two complex half-planes for antisymmetric dispersion (see SM~\cite{Wang2025supplement}).

\emph{Levinson's theorem.}---
We establish a generalized Levinson's theorem for the dispersions considered here, relating the total scattering-phase change to the number of bound states $N_B$. The zero-energy scattering limit gives rise to a universal contribution to the phase change.
For antisymmetric dispersion
$\epsilon(k)=\operatorname{sgn}(k)|k|^m$,
define the generally complex scattering phase $\delta(E)$ by
$S(E)=\exp[2\mathrm{i}\delta(E)]$. 
We track $\delta(E)$ continuously as $E$ increases from
$-\infty$ to $0^-$ and from $0^+$ to $+\infty$,
defining $\Delta\delta$ as the sum of the phase changes
along these two segments.

We consider noncritical systems without bound states in the continuum and with finite, nonzero $S(E)$ at all nonzero energies.
Under the regularity and ultraviolet-convergence conditions stated in Sec.~VII of the SM~\cite{Wang2025supplement}, we obtain
\begin{subequations}
\begin{align}
 \Delta \delta    & = \pi  \Delta N + \pi\left[D-\frac{\eta(\mathcal N)}{m}\right] ,  \label{eqdelta}\\
 \Delta N&\equiv
\begin{cases}
 N - N_B  & \text{emitter scattering}, \\
 - N_B    & \text{separable potential scattering}
\end{cases}.
\label{eqLevinson}
\end{align}
\label{eqLevinsont}
\end{subequations}
where $\Delta N$ denotes the change in the number of bound states due to the interaction $V$ and
$N_B$ counts bound states outside the continuum with algebraic multiplicity. The zero-energy contribution (the second term in Eq.~\eqref{eqdelta}) depends solely on $m$ and $\mathcal N$. It vanishes for $\mathcal N=\varnothing$ or $\{\bar n\}$.

For the symmetric dispersion $\epsilon(k)=|k|^m$, define $\det[\bm S(E)]=\exp[2\mathrm{i}\delta(E)]$ and track $\delta(E)$ continuously from $E=0^+$ to $+\infty$.
Under the same noncriticality, spectral, and regularity assumptions, with $S(E)$ replaced by $\det[\bm S(E)]$, Eq.~\eqref{eqLevinsont} also holds. Taking $\delta(+\infty)=0$ as the phase reference, it
defines the zero-energy phase
$\delta(0^+)=-\Delta\delta$. If $N_B$ changes by one across a critical point while $m$, $\mathcal N$, and $N$ remain fixed, the zero-energy phase $\delta(0^+)$ changes by $\pi$. This bound-state transition occurs precisely where the noncriticality condition of Theorem~\ref{theoremUniScatteringSym} fails. 

\emph{Outlook.}---
In the End Matter, we analyze the low-energy scattering-phase
expansion for analytic monomial dispersions, deriving the leading
term for contact interactions and the resulting crossover
energy scale and Wigner-delay scaling~\cite{Wigner1955,TaylorScattering}.

Photonic crystal waveguides offer a platform to test our universal scattering predictions through transmission and reflection spectroscopy. Their engineered dispersions could also provide tunable non-Markovian environments for quantum emitters, where the exponent $m$ and low-momentum emitter coupling can influence long-time dynamics and bound-state structure. This could enable control of population trapping~\cite{john_quang_pbg_1994,garraway_decay_reservoir_1997}, collective decay~\cite{john_quang_superradiance_1995,vats_john_nonmarkov_1998}, few-photon scattering and excitation of emitter--photon bound states~\cite{trivedi_two_photon_2021,gu_two_photon_2024}.

Finally, higher-dimensional band structures can exhibit divergent
densities of states near isolated higher-order critical points or
extended manifolds of stationary points
\cite{Yuan2019HighOrderVHS,Wu2024ThreeDimensionalVHS}.
In such settings, anisotropy, geometry, and symmetry emerge as
additional control parameters alongside the dispersion exponent,
promising an even richer landscape of universal scattering phenomena.

\begin{acknowledgments}
\emph{Acknowledgements.}--- We thank Fan Du, Kuan Lu,  Xiao Wang, Kevin Xu, Hannes Pichler, Junyi Zhang, Anatoly Svidzinsky, Seth Whitsitt, Liujun Zou, Oriol Rubies Bigorda, and Jonah S.~Peter for discussions. Y.W. and S.Y. acknowledge the NSF through CUA PFC (PHY-2317134), PHY-2207972, the AFOSR (FA9550-24-1-0311) and the Q-SEnSE QLCI (OMA-2016244). A.V.G.~was supported in part by ARL (W911NF-24-2-0107), DARPA SAVaNT ADVENT, the DoE ASCR Quantum Testbed Pathfinder program (awards No.~DE-SC0019040 and No.~DE-SC0024220), NSF QLCI (award No.~OSI-2553574), NSF STAQ program, AFOSR MURI, ONR MURI, and NQVL:QSTD:Design:FTL. A.V.G.~also acknowledges support from the U.S.~Department of Energy, Office of Science, National Quantum Information Science Research Centers, Quantum Systems Accelerator (QSA) and from the U.S.~Department of Energy, Office of Science, Accelerated Research in Quantum Computing, Fundamental Algorithmic Research toward Quantum Utility (FAR-Qu).
\end{acknowledgments}

\nocite{Smith1960LifetimeMatrix,Chen2021ComplexWignerTimeDelay}
\bibliographystyle{apsrev4-2}
\begingroup\edef\SMbibcall{\endgroup\noexpand\putbib[library_Oct22,\jobname Notes]}\SMbibcall

\begingroup
\appendix

\section*{End Matter: low-energy expansion}
In this End Matter, we analyze the low-energy expansion of the scattering phase $\delta(E)$ around the universal zero-energy $S$-matrix fixed points. The form of this expansion depends on the low-momentum structure of both the dispersion relation (characterized here by its type and exponent $m$) and $|v(k)\rangle$. The coefficients are generally nonuniversal, depending on microscopic details such as $\bm K^R$, $\bm U$, and the behavior of $|v(k)\rangle$ over the full momentum range.

We restrict to analytic monomial dispersions $\epsilon(k)=k^m$ with integer $m>1$, for which the expansion is in integer powers of $k=|E|^{1/m}$.
For even $m$ (symmetric dispersion), all expansion terms $k^{jm}$ with integer $j\geq1$ are absent. For odd $m$ (antisymmetric dispersion), the expansions must be performed separately for $E>0$ and $E<0$. Depending on the low-momentum structure of $|v(k)\rangle$, corresponding coefficients may be required to coincide or allowed to differ. No analogous exclusion rule forbids terms of order $k^{jm}$ for odd $m$.

For both even and odd $m$, the low-momentum structure of $|v(k)\rangle$ can further restrict which powers appear in the expansion.  
In the following, we specialize to contact interactions to determine the lowest allowed order of the phase expansion and the resulting crossover energy scale and Wigner-delay scaling. 
These results are useful for the future development of effective field theories. The proofs are given in the SM~\cite{Wang2025supplement}.





\subsection{Leading low-energy expansion for contact interactions}

 For contact interactions, we take monomial form factors
$|v(k)\rangle=[k^{n_1},\ldots,k^{n_N}]^T$,
with $0\le n_1<\cdots<n_N$, ultraviolet regularized
when necessary as described in the SM~\cite{Wang2025supplement}.
We restrict to interior interactions with $D>0$ and
$\bar n\notin\mathcal N$.
While $\mathcal N=\{n_1,\ldots,n_D\}$ determines the
zero-energy fixed point, the lowest possible order
$\ell_{\rm cut}$ of the scattering-phase correction is

\begin{equation}
\ell_{\mathrm{cut}} =
\begin{cases}
    m-1-2n_D, & D=N,\\
    \min\!\bigl(m-1-2n_D,\,2n_{D+1}-m+1\bigr),
        & D<N.
\end{cases}
\label{eqMainLcut}
\end{equation}
For $D<N$, the order in Eq.~\eqref{eqMainLcut} also depends on $n_{D+1}>\bar n$,
which is not included in $\mathcal N$.

For a noncritical system with symmetric monomial dispersion
$\epsilon(k)=k^m$ (even $m$), the leading-order expansion takes the following form
\eq{
\delta(E) - \delta(0^+) = -\alpha k^{\ell_{\mathrm{cut}}}+o(k^{\ell_{\mathrm{cut}}}). \label{eqexpandEvenm}
}
When \(\ell_{\rm cut}=1\), \(\alpha\) has dimensions of length and  defines  a generalized 1D scattering length.
For  \(\ell_{\rm cut}\neq 1\), $\alpha$ defines a generalized
scattering parameter with dimensions $k^{-\ell_{\rm cut}}$.

For $\epsilon(k)=k^m$ with odd $m$, we restrict the low-energy
expansion to interior interactions ($\bar n\notin\mathcal N$),
since the zero-energy phases $\delta(0^\pm)$ are nonuniversal
for borderline interactions.
The low-energy expansion must be evaluated separately when   $E\rightarrow 0^\pm$:
\eq{
\delta(E) - \delta(0^\pm) = -\alpha^\pm  k^{\ell_{\mathrm{cut}}}+o(k^{\ell_{\mathrm{cut}}}). \label{eqexpandOddm}
}
The coefficients satisfy $\alpha^+=\alpha^-$ except when both
$0<D<N$ and $n_D+n_{D+1}=m-1$ hold.
In this resonant case, $\alpha^+$ and $\alpha^-$ are generically distinct.

The coefficients \(\alpha\) (or \(\alpha^\pm\)) depend on \(\bm U\) and, when \(D<N\), on the ultraviolet regularization of the monomial form factors \(k^{n_i}\) with \(i>D\). They are generically nonzero but may vanish under additional fine-tuning, in which case the leading expansion term occurs at higher order.

\subsection{Crossover energy scale and Wigner delay}
For generic interactions, the expansions retain the forms
of Eqs.~\eqref{eqexpandEvenm} and \eqref{eqexpandOddm},
but the lowest possible order $\ell_{\rm cut}$ remains to be determined.
For even $m$, the low-energy expansion defines a characteristic crossover scale $E^*\sim|\alpha|^{-m/\ell_{\rm cut}}$, below which the universal zero-energy behavior is approached. For odd $m$, the corresponding scales are $E_\pm^*\sim|\alpha^\pm|^{-m/\ell_{\rm cut}}$ on the two sides of zero energy. If $\alpha$ or either of $\alpha^\pm$ vanishes, the next nonzero term in the corresponding low-energy expansion determines the crossover scale.
As a critical point is approached, $\alpha$ or $\alpha^\pm$ can diverge, causing the corresponding crossover scale to vanish.

For antisymmetric dispersion, the Wigner delay
$\tau_W(E)=2\hbar\partial_E\operatorname{Re}[\delta(E)]$
quantifies the scattering-induced time shift of a sufficiently narrow-band transmitted wave packet relative to free propagation.
For symmetric dispersion,
$\tau_W(E)=\hbar\partial_E\operatorname{Re}[\delta(E)]$
is the channel-averaged Wigner--Smith delay for unitary scattering~\cite{Smith1960LifetimeMatrix}. For nonunitary scattering, these expressions give the real part of the generalized delay~\cite{Chen2021ComplexWignerTimeDelay}.
 When $\operatorname{Re}\alpha\neq0$ or $\operatorname{Re}\alpha^\pm\neq0$, Eqs.~\eqref{eqexpandEvenm} and \eqref{eqexpandOddm} give
$\tau_W(E)\propto |E|^{\ell_{\rm cut}/m-1}$
as $E\to0$.
Since $\ell_{\rm cut}<m$ for nontrivial interactions, the Wigner delay generically diverges
near zero energy. For unitary scattering, the sign of $\tau_W(E)$ distinguishes a delay from an advance.

The finite spectral width $\Delta E$ of an incident wave packet limits how closely the low-energy Wigner-delay scaling can be resolved. 
Dissipation, disorder, and deviations from the ideal dispersion may impose additional system-dependent limitations on the observation of the universal low-energy scaling.

\endgroup  
\end{bibunit}

\clearpage
\onecolumngrid
\setcounter{page}{1}
\renewcommand{\thepage}{S\arabic{page}}

\setcounter{equation}{0}
\setcounter{figure}{0}
\setcounter{table}{0}
\setcounter{section}{0}
\counterwithout{equation}{section}
\setcounter{subsection}{0}
\setcounter{secnumdepth}{3}
\setcounter{tocdepth}{3}
\renewcommand{\theequation}{S\arabic{equation}}
\renewcommand{\thefigure}{S\arabic{figure}}
\renewcommand{\thetable}{S\arabic{table}}
\renewcommand{\thesection}{\Roman{section}}
\renewcommand{\thesubsection}{\thesection.\Alph{subsection}}

\providecommand{\theHequation}{}
\providecommand{\theHfigure}{}
\providecommand{\theHtable}{}
\providecommand{\theHsection}{}
\renewcommand{\theHequation}{supplement.equation.\arabic{equation}}
\renewcommand{\theHfigure}{supplement.figure.\arabic{figure}}
\renewcommand{\theHtable}{supplement.table.\arabic{table}}
\renewcommand{\theHsection}{supplement.section.\arabic{section}}

\let\theorem\smtheorem
\let\endtheorem\endsmtheorem
\let\definition\smdefinition
\let\enddefinition\endsmdefinition
\let\lemma\smlemma
\let\endlemma\endsmlemma
\let\corollary\smcorollary
\let\endcorollary\endsmcorollary

\renewcommand{\blue}[1]{#1}
\renewcommand{\red}[1]{#1}
\newcommand{\orange}[1]{#1}
\newcommand{\brown}[1]{#1}

\makeatletter
\let\SM@originaladdcontentsline\addcontentsline
\def\SM@toctype{toc}
\renewcommand{\addcontentsline}[3]{%
  \def\SM@thisfile{#1}%
  \ifx\SM@thisfile\SM@toctype
     \SM@originaladdcontentsline{stoc}{#2}{#3}%
  \else
     \SM@originaladdcontentsline{#1}{#2}{#3}%
  \fi
}
\renewcommand{\tableofcontents}{%
   \begin{center}{\large\bfseries Contents}\end{center}%
   \@starttoc{stoc}%
}
\makeatother

\begin{center}
{\Large\bfseries Supplemental Material for: ``Landscape of scattering universality with general dispersion relations''\par}
\vspace{1.5ex}
Yidan Wang\textsuperscript{1}, Xuesen Na\textsuperscript{2}, Michael J. Gullans\textsuperscript{3}, Susanne F. Yelin\textsuperscript{1}, Alexey V. Gorshkov\textsuperscript{3,4}\par
\vspace{1ex}
{\small \textsuperscript{1}Department of Physics, Harvard University, Cambridge, Massachusetts 02138, USA\\[0.3ex]
\textsuperscript{2}Department of Mathematics, University of Illinois Urbana-Champaign, Urbana, IL 61801\\[0.3ex]
\textsuperscript{3}Joint Center for Quantum Information and Computer Science, University of Maryland and NIST, College Park, Maryland 20742, USA\\[0.3ex]
\textsuperscript{4}Joint Quantum Institute, University of Maryland and NIST, College Park, Maryland 20742, USA\par}
\end{center}

\begin{center}
\begin{minipage}{0.94\textwidth}
This Supplemental Material provides the proofs and supporting derivations for the zero-energy scattering results and the generalized Levinson's theorem presented in the main text for one-dimensional emitter scattering and separable potential scattering with general dispersions. We also derive the low-energy expansion of the scattering phase for analytic monomial dispersions.

We first construct the integer set $\mathcal{N}$ appearing in Lemma~1 (Sec.~I). We then introduce the $J$-matrix $\bm{J}(\omega)$ and Jost-function formalism and define criticality perturbatively. We show that this definition is equivalent to a condition on the smallest singular value of $\bm{J}(\omega)$ (Sec.~II).

We then prove Theorem~1 for antisymmetric dispersions, beginning with the $N=1$ case and extending to arbitrary $N$ for both interior and borderline interactions (Secs.~III--IV). For symmetric dispersions, we first establish the $S$-matrix formalism and the $N=1$ result, and then prove Theorem~2 for arbitrary $N$ by analyzing the parity-resolved infrared structure of $\bm{J}(\omega)$ (Secs.~V--VI). We subsequently prove the generalized Levinson's theorem for both dispersion types using the analytic properties of the Jost function (Sec.~VII).

Finally, we study the low-energy expansion for analytic monomial dispersions $\epsilon(k)=k^m$. We first establish the general structure of the scattering-phase expansion and then derive the leading low-energy correction for contact interactions (Secs.~VIII--IX). We conclude with a complete classification of the universal zero-energy $S$-matrix phases for integer $m>1$ and a summary of the notation used throughout (Secs.~X--XI).
\end{minipage}
\end{center}

\begin{bibunit}[apsrev4-2]
\renewcommand{\bibnumfmt}[1]{[S#1]}
\renewcommand{\citenumfont}[1]{S#1}

 \tableofcontents

\section{Proof of Lemma 1: Determination of the Integer Set \(\mathcal{N}\)}
\label{AppenIntegerSet}

We prove Lemma~1 by constructing the unique exponent set \(\mathcal{N}\). 
If $\lVert\ket{v(k)}\rVert=o(|k|^{\bar n})$ as $k\to0$, then $\mathcal N=\varnothing$. Otherwise, consider the lowest-order term in the expansion of \(\ket{v(k)}\),
\begin{equation}
    \ket{v(k)}
    = c_1 k^{n_1}\ket{u_1}+o(k^{n_1}),
\end{equation}
where \(c_1\neq 0\) and \(\ket{u_1}\) is normalized. The preceding assumption implies $n_1\leq\bar n$, so $n_1\in\mathcal N$.

To identify the remaining elements of \(\mathcal{N}\), remove the component along \(\ket{u_1}\) and define
\begin{equation}
    \ket{v_1^\perp(k)}
    =\ket{v(k)}
    -\braket{u_1|v(k)}\ket{u_1}.
\end{equation}
If \(\|\ket{v_1^\perp(k)}\|=o(|k|^{\bar n})\), the construction terminates and
\(\mathcal{N}=\{n_1\}\). Otherwise, let
\begin{equation}
    \ket{v_1^\perp(k)}
    = c_2 k^{n_2}\ket{u_2}+o(k^{n_2}),
\end{equation}
where \(c_2\neq 0\) and \(\ket{u_2}\) is normalized and orthogonal to
\(\ket{u_1}\). The projection removes the order $n_1$, so $n_1<n_2\leq\bar n$ and $n_2\in\mathcal N$.

Repeating this projection and lowest-order expansion yields mutually orthogonal normalized vectors
\(\ket{u_1},\ldots,\ket{u_D}\) and corresponding integers
\(n_1,\ldots,n_D\). The procedure terminates when the component of
\(\ket{v(k)}\) orthogonal to the span of
\(\{\ket{u_1},\ldots,\ket{u_D}\}\) is $o(|k|^{\bar n})$. Thus,
\[
    \mathcal{N}=\{n_1,\ldots,n_D\},
\]
where $D$ counts the independent coupling directions appearing at orders not exceeding $\bar n$. Since each step selects the unique lowest-order nonzero component in the orthogonal complement, the exponent set $\mathcal N$ and the coupling directions $\ket{u_i}$ are uniquely determined up to arbitrary phases of the normalized vectors.

\section{$J$-matrix formalism and critical points\label{sec_critical_defi}}

To lay the groundwork for the proofs of Theorems~1 and 2, we first introduce a unified mathematical framework describing single-excitation scattering for both emitter interactions and separable potentials. In Sec.~\ref{seJostfunction}, we define the Jost function, which serves as the central analytical object in our derivations. Subsequently, in Sec.~\ref{SecCritical}, we provide a rigorous definition of critical points for the interactions under consideration.

\subsection{$J$-matrix and the Jost function \label{seJostfunction}}

The primary mathematical tool in our proof is a generalized version of Eq.~(4) from the main text. This relation expresses the determinant of the $S$-matrix as the ratio of the boundary values of the Jost function evaluated across the continuous spectrum:
\begin{subequations}
\begin{equation}
\det[\bm{S}(E)] = \frac{J(E-\mathrm{i}0)}{J(E+\mathrm{i}0)}, \qquad J(\omega) \equiv \det[\bm{J}(\omega)]. 
\label{SratioJ}
\end{equation}
Here, the Jost function is defined as the determinant of an $N\times N$ matrix $\bm{J}(\omega)$. For an arbitrary dispersion relation, this $J$-matrix is given by
\begin{equation}
\bm{J}(\omega) =
\begin{cases}
\bm{K}(\omega)+\bm{W}-\omega\mathbb{1}_N, & \text{emitter scattering}, \\
\bm{K}(\omega)+\bm{W}, & \text{separable potential},
\end{cases} 
\label{eqMKrelations} 
\end{equation}
where $\mathbb{1}_N$ is the identity matrix, and $\bm{W}$ is a frequency-independent interaction matrix defined as
\begin{equation}
\bm{W} =
\begin{cases}
 \bm{K}^R,  & \text{emitter scattering}, \\
 -\bm{U}^{-1}, & \text{separable potential}.
\end{cases} 
\label{eqWrelations0}
\end{equation}
The inverse $\bm U^{-1}$ is well defined because $\bm U$ is assumed to be invertible, as stated in the main text.
The matrix $\bm{K}(\omega)$ is the projected resolvent, defined by the integral
\begin{equation}
\bm{K}(\omega) \equiv \int_{-\infty}^{+\infty} dk \frac{|v(k)\rangle \langle v(k)|}{\omega-\epsilon(k)}.
\label{eqKdefi}
\end{equation}
\label{eqSdetJratio}
\end{subequations}
Physically, the projected resolvent $\bm K(\omega)$ describes propagation through the continuum projected onto the finite emitter or form-factor space. For emitter scattering, it gives the continuum-induced self-energy and effective coupling between emitters, while for separable-potential scattering it is the free resolvent projected onto the form-factor subspace. Furthermore, its elements are analytic throughout the complex $\omega$-plane, excluding a potential branch cut along the continuum spectrum, which spans $(-\infty, +\infty)$ for antisymmetric dispersions and $(0, +\infty)$ for symmetric dispersions.

By encapsulating the model-specific parameters $\bm{K}^R$ and $-\bm{U}^{-1}$ into the frequency-independent interaction matrix $\bm{W}$ [Eqs.~\eqref{eqMKrelations} and \eqref{eqWrelations0}], both types of interactions can be analyzed within a single, unified mathematical framework. As stated in the main text, we require $\bm{K}^R$ to have no null vectors that are simultaneously orthogonal to $|v(k)\rangle$ for all $k$. Any emitter scattering Hamiltonian can be reduced to satisfy this condition without loss of generality. Suppose $\bm{K}^R$ has $N_{\rm null}$ such null vectors; these represent zero-energy eigenstates that are entirely decoupled from scattering. Projecting these $N_{\rm null}$ states out yields an equivalent system of $N-N_{\rm null}$ emitters with an identical scattering matrix. Because this reduction decreases both $N$ and $N_B$ by $N_{\rm null}$, Levinson's theorem remains invariant.
Furthermore, because $\bm{U}^{-1}$ is invertible, this condition on $\bm{K}^R$ ensures that $\bm{W}$ likewise possesses no null vectors orthogonal to all $|v(k)\rangle$. This requirement on $\bm{W}$ simplifies our subsequent definition of critical points, which involves the emergence of bound states from zero energy.

For antisymmetric dispersions, the $S$-matrix reduces to a scalar, $S(E)=\bm{S}(E)=\det[\bm{S}(E)]$. Consequently, Eq.~\eqref{SratioJ} yields the $S$-matrix as the ratio of the Jost function's boundary values across the real axis. For symmetric dispersions, however, the $S$-matrix is a $2\times 2$ matrix. In this case, the determinant relation in Eq.~\eqref{SratioJ} is insufficient to fully specify the scattering, and the complete $S$-matrix expression is instead deferred to Sec.~\ref{secSymm}.

The universal zero-energy $S$-matrix limits established in Theorems~1 and~2 of the main text break down at critical points. In the following subsection, we define criticality in terms of the sensitivity of the system to arbitrarily small perturbations. We then translate this physical concept into a robust, intrinsic condition on the $J$-matrix, $\bm{J}(\omega)$.

\subsection{Physical Definition of a Critical Point\label{SecCritical}}

In ordinary Hermitian scattering, criticality is associated with a bound state reaching a spectral edge. For the antisymmetric dispersion considered here, $\epsilon(k)=\operatorname{sgn}(k)|k|^m$, the continuum has no spectral edge. We therefore generalize this notion by defining criticality relative to the low-energy fixed point at $E=0$.

For the interactions considered in this work, we call a system critical if arbitrarily small perturbations of the interaction matrix $\bm{W}$ can create a bound state arbitrarily close to \(E=0\). 
\begin{definition} \label{definitionCritical}

A Hamiltonian is defined as critical if there exists a sequence of perturbations \(\delta \bm{W}_\nu\) with \(\|\delta \bm{W}_\nu\|\to0\) that generates bound states with energies \(\omega_\nu\to0\), where each \(\omega_\nu\) lies outside the continuum spectrum.

\end{definition}
Here, \(\delta \bm{W}_\nu\) is an \(N \times N\) complex-valued matrix representing an arbitrary perturbation to the interaction matrix \(\bm{W} \). The norm \(\lVert \delta \bm{W}_\nu \rVert\) is the operator norm. The perturbed Hamiltonian has the same dispersion relation \(\epsilon(k)\) and the same coupling coefficients, or form factors, \(\ket{v(k)}\), as the original system.
These perturbations \(\delta \bm{W}_\nu\) are allowed to be general non-Hermitian perturbations, including both dissipative and amplifying components. Hence the generated bound states may have complex energies. A bound state with energy above the real axis has an amplitude that grows in time, while one with energy below the real axis has an amplitude that decays in time. See Ref.~\cite{wang2018single} for a discussion of such bound states in non-Hermitian systems.

For a Hermitian Hamiltonian with an antisymmetric dispersion, any bound state, if present, must be a fine-tuned bound state in the continuum. Thus, perturbations satisfying the criterion in Definition~\ref{definitionCritical} are generically non-Hermitian, and the near-zero bound states they generate generally have complex energies with nonzero imaginary parts.

\subsection{Critical Points and Singular Values of the J-matrix}
\label{AppendixCriticalDefi}

 Definition~\ref{definitionCritical} is  perturbative: it classifies a system as critical if an arbitrarily small matrix perturbation renders \(\bm{J}(\omega)\) singular near \(\omega=0\). While this formulation captures the clear physical intuition of criticality, it is mathematically cumbersome to apply directly. We therefore recast it as an intrinsic condition on the unperturbed matrix \(\bm{J}(\omega)\). Letting \(\sigma_{\min}(\bm{A})\) denote the smallest singular value of a matrix \(\bm{A}\), we establish the following criterion.

\begin{lemma} \label{definitionCritical0}
A system is critical if and only if there exists a sequence $\omega_\nu\to0$, with each $\omega_\nu$ outside the continuum spectrum, such that
\begin{equation}
    \sigma_{\min}[\bm{J}(\omega_\nu)] \to 0 .
    \label{eqCriticalSmin}
\end{equation}
Equivalently, there exists such a sequence for which either $\bm J(\omega_\nu)$ is singular for every $\nu$, or $\bm J(\omega_\nu)$ is invertible for every $\nu$ and
\begin{equation}
    \lVert \bm{J}(\omega_\nu)^{-1}\rVert\to\infty . 
    \label{eqCriticalInverseNorm}
\end{equation}
\end{lemma}

The inverse-norm formulation, Eq.~\eqref{eqCriticalInverseNorm}, follows from Eq.~\eqref{eqCriticalSmin} and the identity
\begin{equation}
    \lVert \bm J(\omega)^{-1}\rVert
    =
    \frac{1}{\sigma_{\min}[\bm J(\omega)]},
    \label{eqInverseNormSmin}
\end{equation}
which holds whenever \(\bm J(\omega)\) is invertible.

 In subsequent sections, Lemma~\ref{definitionCritical0} is converted into explicit criteria involving the finite matrix elements  of \(\bm{J}(\omega)\). These criteria are stated for antisymmetric dispersions in Lemmas~\ref{lemma_inter_crit} and \ref{lemma_criti_anti_borderline}, and for symmetric dispersions in Lemma~\ref{lemmaCriticalSym}.

 \begin{proof}[Proof of Lemma~\ref{definitionCritical0}]
By definition, a bound state at an energy $\omega$ outside the continuous spectrum exists if and only if the $J$-matrix is singular, i.e., $\det[\bm J(\omega)]=0$. To analyze perturbations, we write $\bm J(\omega;\bm W)$ to explicitly denote its dependence on the interaction matrix $\bm W$. Because $\bm W$ enters the $J$-matrix strictly additively [see Eq.~\eqref{eqMKrelations}], any shift obeys
\begin{equation}
\bm J(\omega;\bm W+\delta\bm W)=
\bm J(\omega;\bm W)+\delta\bm W.
\label{eqCriticalJShift}
\end{equation}

Assume first that Eq.~\eqref{eqCriticalSmin} holds. There must then exist a sequence $\omega_\nu\to0$, residing outside the continuum, for which the corresponding singular values vanish:
\begin{equation}
    \sigma_\nu \equiv \sigma_{\min}\!\left[\bm J(\omega_\nu;\bm W)\right] \to 0 \qquad\text{as }\nu\to\infty.
\end{equation}

Let $|x_\nu\rangle$ and $|y_\nu\rangle$ be the normalized right and left singular vectors corresponding to $\sigma_\nu$, such that $\bm J(\omega_\nu;\bm W)|x_\nu\rangle=\sigma_\nu|y_\nu\rangle$. We then define a sequence of rank-one perturbations
\begin{equation}
\delta\bm W_\nu
\equiv
-\sigma_\nu|y_\nu\rangle\langle x_\nu|.
\label{eqCriticalRankOnePerturbation}
\end{equation}
By construction, the operator norm of this perturbation is $\|\delta\bm W_\nu\|=\sigma_\nu$, which vanishes as $\nu \to \infty$. Furthermore, applying Eqs.~\eqref{eqCriticalJShift} and \eqref{eqCriticalRankOnePerturbation} yields
\begin{equation}
\bm J(\omega_\nu;\bm W+\delta\bm W_\nu)|x_\nu\rangle = (\sigma_\nu - \sigma_\nu)|y_\nu\rangle = 0.
\label{eqCriticalPerturbedSingular}
\end{equation}
Thus, the perturbed system supports a bound state at energy $\omega_\nu$. Because $\omega_\nu\to0$ under an arbitrarily small perturbation ($\|\delta\bm W_\nu\|\to 0$), the unperturbed system is critical according to Definition~\ref{definitionCritical}.

Conversely, assume the system is critical according to Definition~\ref{definitionCritical}. There must exist a sequence of perturbations $\delta\bm W_\nu$ with $\|\delta\bm W_\nu\|\to0$ generating corresponding bound-state energies $\omega_\nu\to0$ outside the continuum. At these energies, the perturbed matrix $\bm J(\omega_\nu;\bm W+\delta\bm W_\nu)$ is exactly singular. By Eq.~\eqref{eqCriticalJShift}, $\delta\bm W_\nu$ represents the specific additive shift required to render the unperturbed matrix $\bm J(\omega_\nu;\bm W)$ singular. We now invoke the standard linear algebraic identity that the operator-norm distance from any matrix to the nearest singular matrix is precisely its smallest singular value. Therefore, $\sigma_{\min}[\bm J(\omega_\nu;\bm W)]$ is   bounded by the size of this specific singularity-inducing perturbation
\begin{equation}
\sigma_{\min}[\bm J(\omega_\nu;\bm W)]
\leq
\|\delta\bm W_\nu\|
\to 0.
\label{eqCriticalConverseBound}
\end{equation}
This establishes Eq.~\eqref{eqCriticalSmin}, completing the proof.
\end{proof}

\section{Antisymmetric Dispersion: The $N=1$ case\label{SecAntiN1}}
 We now begin the proof of Theorem~1 for antisymmetric dispersions, which spans both the current section and Sec.~\ref{secProofOdd}. Here, we establish the mathematical foundation by solving the simple $N=1$ case and defining the infrared model functions $L_{m,l}(\omega)$ that dictate the divergent low-energy physics. This setup naturally paves the way for the full multi-component ($N>1$) proof detailed in Sec.~\ref{secProofOdd}.  For this simplified case, The proof of Theorem~1 for this case is achieved by proving the following lemma.

\begin{lemma}[Single-component $S$-matrix limits]
\label{lemmaN1}
For $N=1$ with interaction exponent $n$, the zero-energy limit of the $S$-matrix exhibits three distinct regimes:
\begin{enumerate}
    \item[(i)] If $n>\bar{n}$ and the system is noncritical, the scattering is trivial: $S(0^\pm)=1$.
    \item[(ii)] If $n=\bar{n}$ and the system is noncritical, the limit is continuous but nonuniversal: $S(0^+)=S(0^-)$.
    \item[(iii)] If $n<\bar{n}$, the system is unconditionally noncritical, and the $S$-matrix attains the universal limit
    \begin{equation}
    S(0^\pm)=\exp\!\left( \pm \mathrm{i}\pi \frac{2n+1}{m}\right).
    \label{eqSantisymmN1}
    \end{equation}
\end{enumerate}
\end{lemma}

\begin{proof}
For a single-component interaction ($N=1$), the projected resolvent matrix $\bm{K}(\omega)$ reduces to the complex scalar function $K(\omega)$ defined in Eq.~(5) of the main text. To prove the three regimes of the scattering limit, we evaluate the asymptotic behavior of $K(\omega)$ and the corresponding scalar Jost function $J(\omega)$. The scalar S-matrix $S(E)$ is subsequently evaluated from Eq.~\eqref{eqSdetJratio}. According to Lemma~\ref{definitionCritical0}, the system is critical if and only if $J(0^+)=0$ or $J(0^-)=0$.

First, we consider statement (i), where $n>\bar{n}$ (i.e., $2n > m-1$). Here, the integral defining $K(\omega)$ is infrared-convergent, meaning $K(\omega)$ approaches a finite, unique limit as $\omega \to 0$, independent of the complex half-plane. Assuming the system is noncritical ($J(0) \neq 0$), the ratio $J(E-\mathrm{i}0)/J(E+\mathrm{i}0) \to 1$, yielding the trivial scattering limit $S(0^\pm)=1$.

Next, we address statement (ii). When $m$ is odd, the interaction can be borderline ($n=\bar{n}$, or $2n=m-1$). The projected resolvent $K(\omega)$ approaches finite but distinct limits as $\omega \to \mathrm{i}0^\pm$. The Jost function subsequently takes the form
\begin{subequations}
\label{eqJborderlineN1}
 \begin{align}
J(\omega)
&=-\operatorname{sgn}[\operatorname{Im}(\omega)] \frac{\mathrm{i}\pi}{m}
|c|^2+j_{\rm reg}+o(1),
\label{eqJborderlineAsymptotic}
\\
j_{\rm reg}
&\equiv
W 
+
\operatorname{PV}\!\int_{-\infty}^{+\infty}dk\,
\frac{|v(k)|^2}{-\operatorname{sgn}(k)|k|^m},
\label{eqyB}
\end{align}
\end{subequations}
where $\operatorname{PV}$ denotes the principal value of the integral. Here, $j_{\rm reg}$ denotes the finite, half-plane-independent regular contribution to the scalar Jost function, which depends on global properties of $v(k)$ and the interaction matrix $\bm{W}=W$. For a noncritical system, $J(0^\pm)\neq0$, ensuring that the ratio $J(E-\mathrm{i}0)/J(E+\mathrm{i}0)$ is well defined. Substituting Eq.~\eqref{eqJborderlineN1} into this ratio and evaluating the limits $|E|\to 0^+$ yields the nonuniversal but equal limits 
\eq{
S(0^\pm)= \frac{j_{\rm reg}+\frac{\pi \mathrm{i}}{m}|c|^2}{j_{\rm reg}-\frac{\pi \mathrm{i}}{m}|c|^2},
}
which depend on $W$ and the global integral over $k$ defined in Eq.~\eqref{eqyB}. 

Finally, we establish statement (iii). When $n<\bar{n}$, the projected resolvent $K(\omega)$ exhibits an infrared divergence as $\omega\rightarrow 0$. To characterize this divergence, we define the infrared model functions
\begin{equation}
L_{m,l}(\omega) \equiv \int_{-\infty}^{+\infty} dk\, \frac{k^{l}}{\omega - \operatorname{sgn}(k)|k|^m}, \label{eqSLmndefinition0}
\end{equation}
for integers $0 \le l < m-1$. Parameterizing the complex frequency as $\omega = |\omega|e^{\mathrm{i}\theta}$ with $\theta \in (0,2\pi)$, and setting the complex momentum $q_\omega = |\omega|^{1/m}e^{\mathrm{i}\theta/m}$, we have
\begin{subequations}
\label{eqLml_all}
\begin{equation}
L_{m,l}(\omega) = -\frac{2\pi \mathrm{i}}{m} \kappa_{m,l}\big(\!\operatorname{sgn}[\operatorname{Im}(\omega)]\big) q_\omega^{-m+l+1}, \label{eqLml}
\end{equation}
where the coefficients \(\kappa_{m,l}(\pm)\) depend on the half-plane:
\begin{align}
\kappa_{m,l}(+) &=-\frac{1}{(-1)^l\mu^{l+1}-1}, \\
\kappa_{m,l}(-)&=-\frac{1}{\mu^{l+1}(\mu^{l+1}-(-1)^l)}, 
\end{align} 
\label{eqkappa}
\end{subequations}
\noindent with $\mu \equiv \exp(\mathrm{i}\pi/m)$. The ratio of these half-plane coefficients is exactly $\kappa_{m,l}(-)/\kappa_{m,l}(+)=(-1)^l\mu^{-(l+1)}$.

The leading-order divergence of the Jost function is thus captured by the corresponding model function:
\begin{equation}
J(\omega) = |c|^2 L_{m,2n}(\omega) +o\!\left(L_{m,2n}(\omega)\right).
\label{eqLinvK1}
\end{equation}
Because $|J(\omega)| \to \infty$ as $\omega \to 0$, the condition for criticality can never be met; the system is therefore unconditionally noncritical. Substituting Eq.~\eqref{eqLinvK1} into the ratio $J(E-\mathrm{i}0)/J(E+\mathrm{i}0)$ and evaluating the limits $|E|\to 0^+$ yields the universal phase $S(0^\pm)=\exp\!\left( \pm \mathrm{i}\pi \frac{2n+1}{m}\right)$, completing the proof.
\end{proof}

To facilitate our analysis of borderline interactions in the full proof, which can occur only for odd integer $m$, we extend the infrared model functions [Eq.~\eqref{eqSLmndefinition0}] to $l=m-1$. Because the integrand behaves as $\sim 1/k$ at large momenta, the integral exhibits logarithmic ultraviolet divergences. We formally define the borderline model function via a symmetric integration limit:
\begin{equation}
L_{m,m-1}(\omega) \equiv \lim_{\Lambda\rightarrow \infty }\int_{-\Lambda}^{+\Lambda} dk\, \frac{k^{m-1}}{\omega - \operatorname{sgn}(k)|k|^m}, \qquad \text{odd } m.
\label{eqSLBorderline}
\end{equation}
For odd $m$, this symmetric cutoff guarantees the exact cancellation of the logarithmic divergences along the positive and negative real axes. Consequently, the integral evaluates to
\begin{equation}
L_{m,m-1}(\omega)
=-\operatorname{sgn}[\operatorname{Im}(\omega)] \frac{\mathrm{i}\pi}{m}, \qquad \text{odd } m, 
\label{eqLmmlimitsBorderline}
\end{equation}
which matches the algebraic form of Eq.~\eqref{eqLml} when taking $l=m-1$.

 \section{Antisymmetric Dispersion: Proof of Theorem 1}
\label{secProofOdd}

This section details the proof of Theorem~1 from the main text. The central step is to determine the zero-energy asymptotic behavior of both the Jost function \(J(\omega)=\det[\bm J(\omega)]\) and the inverse matrix \(\bm{J}(\omega)^{-1}\). The former dictates the zero-energy limit of the \(S\)-matrix via Eq.~\eqref{SratioJ},  while the latter determines whether the system is critical according to Lemma~\ref{definitionCritical0}.

For any finite set of integers \(\mathcal A\subseteq\{0,1,\ldots,\bar{n}\}\), we define the product of infrared model functions
\begin{equation}
L_{\mathcal A}(\omega)
\equiv
\prod_{n\in\mathcal A}L_{m,2n}(\omega),
\qquad
L_{\varnothing}(\omega)\equiv1.
\label{eqLsetAnti}
\end{equation}
For borderline interactions, we further define the restricted exponent set which excludes the borderline exponent $\bar{n}$ from $\mathcal{N}$
\begin{equation}
\mathcal N^{<}\equiv\mathcal N\setminus\{\bar{n}\}.
\label{eqNlessBorderline}
\end{equation}

The primary objective of this section is to establish the following lemma.

\begin{lemma}\label{lemma1}
For a noncritical system, the zero-energy asymptotic behavior of the Jost function depends on the interaction type. For interior interactions, it behaves as
\begin{equation}
J(\omega)\sim A\,L_{\mathcal N}(\omega)
\label{eqJmultiSM}
\end{equation}
as \(\omega\to0\), where \(A\neq0\) is a constant. For borderline interactions, the limit is half-plane dependent:
\begin{equation}
J(\omega)
\sim
\begin{cases}
A^+L_{\mathcal N^{<}}(\omega), & \operatorname{Im}\omega>0,\\
A^-L_{\mathcal N^{<}}(\omega), & \operatorname{Im}\omega<0,
\end{cases}
\label{eqJomegacasesSM}
\end{equation}
where \(A^+\) and \(A^-\) are nonvanishing constants.
\end{lemma}

For an interior interaction, substituting the asymptotic form from Eq.~\eqref{eqJmultiSM} into the ratio $S(E)=J(E-\mathrm{i}0)/J(E+\mathrm{i}0)$ yields the zero-energy S-matrix limits
\begin{subequations}
\begin{align}
S(0^+)
&=
\lim_{|E|\to0}
\prod_{n_i \in \mathcal{N}}
\frac{L_{m,2n_i}(|E|e^{\mathrm{i}2\pi^-})}
     {L_{m,2n_i}(|E|e^{\mathrm{i}0^+})}
=
\exp\!\left(\frac{\pi \mathrm{i} \eta(\mathcal{N})}{m}\right),
\label{eqSzeroPlusInteriorAnti}
\\[1ex]
S(0^-)
&=
\lim_{|E|\to0}
\prod_{n_i \in \mathcal{N}}
\frac{L_{m,2n_i}(|E|e^{\mathrm{i}\pi^+})}
     {L_{m,2n_i}(|E|e^{\mathrm{i}\pi^-})}
=
\exp\!\left(-\frac{\pi \mathrm{i} \eta(\mathcal{N})}{m}\right),
\label{eqSzeroMinusInteriorAnti}
\end{align}
\end{subequations}
where 
\eq{
\eta(\mathcal{A}) \equiv \sum_{n_i\in \mathcal{A}}(2n_i+1). \label{eqnAdefi}
}

Similarly, for a borderline interaction, substituting the asymptotic form from Eq.~\eqref{eqJomegacasesSM} into the ratio $J(E-\mathrm{i}0)/J(E+\mathrm{i}0)$ and evaluating the limit $|E|\to 0^+$ yields 
\begin{equation}
S(0^\pm)
= \frac{A^-}{A^+}
\exp\!\left(\pm\frac{\pi \mathrm{i} \eta(\mathcal{N}^<)}{m}\right)=
-\frac{A^-}{A^+}
\exp\!\left(\pm\frac{\pi \mathrm{i} \eta(\mathcal{N})}{m}\right),
\label{eqSzeroBorderlineAnti}
\end{equation}
where we have used $\eta(\mathcal{N})=\eta(\mathcal{N}^<)+m$, which gives $\exp[\pm \mathrm{i}\pi \eta(\mathcal{N}^<)/m] = -\exp[\pm \mathrm{i}\pi \eta(\mathcal{N})/m]$.
Although the prefactor \(A^-/A^+\) is nonuniversal, it cancels in the ratio, leaving the universal phase discontinuity
\begin{equation}
\frac{S(0^+)}{S(0^-)}
=
\exp\!\left(\frac{2\pi \mathrm{i} \eta(\mathcal{N})}{m}\right).
\label{eqSratioAnti}
\end{equation}
This completes the proof of Theorem~1.

The remainder of this section is devoted to the proof of Lemma~\ref{lemma1}. The discussion is organized by interaction type. Interior interactions are treated in Sec.~\ref{SecAppendetZAnti}, and borderline interactions in Sec.~\ref{SecDNBorderline}. Within each subsection, we systematically address the $D=N$ and $D<N$ cases.

\subsection{Interior Interactions\label{SecAppendetZAnti}}

We first prove Lemma~\ref{lemma1} for interior interactions ($\bar{n}\notin \mathcal{N}$). We begin by establishing the general block structure and asymptotic properties of the $J$-matrix, after which we treat the $D=N$ and $D<N$ cases individually.

 Lemma~1 of the main text identifies $D$ orthonormal vectors, $\{\ket{u_1},\ldots,\ket{u_D}\}$, corresponding to the exponent sequence $n_1 < \dots < n_D$. 
Defining $V_i(k)\equiv \langle u_i|v(k)\rangle$, we have
\begin{equation}
     V_i(k)= c_i k^{n_i}+o(k^{n_i}), \qquad c_i\neq0, \qquad n_i< \bar{n}\equiv\frac{m-1}{2}.
    \label{eqAntiViLeading}
\end{equation}

We refer to the span of these vectors as the infrared-divergent subspace. To form a complete orthonormal basis, we introduce the orthogonal complement vectors $\ket{u_{D+1}}, \ldots, \ket{u_N}$ and identify their span as the infrared-finite subspace.
We define the 
 block decomposition of the $J$-matrix 
\begin{equation}
\bm J(\omega)
=
\begin{pmatrix}
\bm J^{(D)}(\omega) & \bm R(\omega)\\
\bm C(\omega) & \bm J^{\perp}(\omega)
\end{pmatrix}.
\label{eqInteriorBlockDecompositionAnti}
\end{equation}
In this representation, the $D\times D$ leading principal block $\bm J^{(D)}(\omega)$ is the restriction of the $J$-matrix to the infrared-divergent subspace, while the $(N-D)\times(N-D)$ infrared-finite block $\bm J^{\perp}(\omega)$ is its restriction to the infrared-finite subspace. The off-diagonal matrices $\bm R(\omega)$ and $\bm C(\omega)$ describe the couplings between these two subspaces. For a Hermitian system, $\bm{C}(\omega)^\dagger=\bm{R}(\omega^*)$. If $D=N$, the infrared-finite sector is absent, $\bm J(\omega)=\bm J^{(D)}(\omega)$.

The four blocks in Eq.~\eqref{eqInteriorBlockDecompositionAnti} exhibit different low-energy asymptotic behaviors. Specifically, the block $\bm J^{(D)}(\omega)$ is  divergent, with its matrix elements scaling as
\begin{equation}
J_{ij}(\omega)
\sim
c_i c_j^*L_{m,n_i+n_j}(\omega)  =
\Theta\left(
q_\omega^{1-m+n_i+n_j}
\right),
\qquad
1\leq i,j\leq D.
\label{eqInteriorJEntriesAnti}
\end{equation}

In comparison, $\bm{J}^{\perp}(\omega)$  has a finite, common limit from the two sides of the real line:
\begin{equation}
\bm J^{\perp}(0)\equiv\lim_{\omega\to0}\bm J^{\perp}(\omega).
\label{eqInteriorJPerpLimitAnti}
\end{equation}
Using indices $i\leq D$ and $r>D$, the elements of the off-diagonal blocks obey the asymptotic scaling
\begin{equation}
R_{ir}(\omega)
=
o\!\left(q_\omega^{n_i-\bar{n}}\right),
\qquad
C_{r i}(\omega)
=
o\!\left(q_\omega^{n_i-\bar{n}}\right).
\label{eqInteriorCrossBlockScalingAnti}
\end{equation}

We first consider the case $D=N$, for which $\bm{J}(\omega)=\bm{J}^{(D)}(\omega)$. The resulting low-energy analysis for the determinant and inverse of $\bm{J}^{(D)}(\omega)$ will be useful in analyzing the corresponding quantities for the full $\bm{J}$-matrix when $D<N$.

\subsubsection{$D=N$\label{InteriorDeqNAnti}}

We begin by calculating the asymptotic behavior of $\det[\bm J^{(D)}(\omega)]$. The Leibniz formula gives
\begin{equation}
\det[\bm J^{(D)}(\omega)]
=
\sum_{\pi\in\mathfrak S_D}
\operatorname{sgn}(\pi)
\prod_{i=1}^{D}J_{i,\pi(i)}(\omega),
\label{eqInteriorLeibnizAnti}
\end{equation}
where \(\mathfrak S_D\) is the set of permutations of
\(\{1,\ldots,D\}\). Using Eq.~\eqref{eqInteriorJEntriesAnti} and the
explicit form of \(L_{m,l}(\omega)\) in Eq.~\eqref{eqLml}, every product in
Eq.~\eqref{eqInteriorLeibnizAnti} has the asymptotic behavior in the two complex half-planes:
\begin{equation}
\prod_{i=1}^{D}J_{i,\pi(i)}(\omega)
\sim
\left(\prod_{i=1}^{D}|c_i|^2\right)
L_{\mathcal N}(\omega)
\prod_{i=1}^{D}
Z_{i,\pi(i)}(\pm), \quad \pm \operatorname{Im}[\omega]>0
\label{eqInteriorPermutationFactorAnti}
\end{equation}
where the coefficients

\begin{equation}
Z_{ij}(\pm)
\equiv
\frac{
\kappa_{m,n_i+n_j}(\pm)
}{
\kappa_{m,2n_i}(\pm)
}
\label{eqZmatrixAnti}
\end{equation}
define constant matrices $\bm{Z}(\pm)$.

Substituting Eq.~\eqref{eqInteriorPermutationFactorAnti} back into Eq.~\eqref{eqInteriorLeibnizAnti} and summing over all permutations gives
\begin{equation}
\det[\bm J^{(D)}(\omega)]
\sim
\left(\prod_{i=1}^{D}|c_i|^2\right)
\det[\bm Z(\pm)]
L_{\mathcal N}(\omega).
\label{eqInteriorDetBeforeZAnti}
\end{equation}
provided that the coefficient $\det[\bm Z(\pm)]$ is nonzero. The following lemma establishes this nonvanishing property and, furthermore, shows that the determinants evaluated in the two half-planes are identical.

\begin{lemma}\label{eqlemmadetZ}
Given a set of distinct integers
\(\mathcal N=\{n_1,\ldots,n_D\}\), with
\(n_i\leq \bar{n}\equiv(m-1)/2\), the matrices \(\bm Z(\pm)\) defined in
Eq.~\eqref{eqZmatrixAnti} satisfy
\begin{equation}
Z(\mathcal N)
\equiv
\det[\bm Z(+)]
=
\det[\bm Z(-)]
=
\prod_{1\leq j<i\leq D}
\frac{
\sin^2\!\left[
\dfrac{\pi(m+1)(n_i-n_j)}{2m}
\right]
}{
\sin^2\!\left[
\dfrac{\pi}{2m}
\bigl((m+1)(n_i+n_j)+1\bigr)
\right]
}.
\label{eqZNoddSM}
\end{equation}
In particular, \(Z(\mathcal N)\) is finite and   positive.
\end{lemma}

\begin{proof}
Define
$x_i
\equiv
(-1)^{n_i}
e^{\mathrm{i}\pi(n_i+1/2)/m}$,
the matrix elements in Eq.~\eqref{eqZmatrixAnti} can then be written as
\begin{equation}
Z_{ij}(+)
=
\frac{1-x_i^2}{1-x_i x_j},
\qquad
Z_{ij}(-)
=
\frac{x_i}{x_j}Z_{ij}(+).
\label{eqZplusMinusRelationAnti}
\end{equation}
Equivalently,
\begin{equation}
\bm Z(-)
=
\operatorname{diag}(x_1,\ldots,x_D)
\bm Z(+)
\operatorname{diag}(x^{-1}_1,\ldots,x^{-1}_D).
\label{eqZSimilarityAnti}
\end{equation}
Here, $\operatorname{diag}(a_1,\ldots,a_D)$ denotes the $D\times D$ diagonal matrix with diagonal entries $a_1,\ldots,a_D$.
Thus \(\bm Z(-)\) and \(\bm Z(+)\) are similar matrices and
\begin{equation}
\det[\bm Z(-)]=\det[\bm Z(+)].
\label{eqZDetEqualAnti}
\end{equation}
To evaluate their common value, define the \(D\times D\) matrix
\(\bm M\) by
 $M_{ij}
\equiv
\frac{1}{1-x_i x_j}$.
Eq.~\eqref{eqZplusMinusRelationAnti} gives
\begin{equation}
\bm Z(+)
=
\operatorname{diag}(1-x_1^2,\ldots,1-x_D^2)\bm M,
\end{equation}
and hence
\begin{equation}
\det[\bm Z(+)]
=
\left(\prod_{i=1}^{D}(1-x_i^2)\right)
\det[\bm M].
\label{eqZCauchyReductionAnti}
\end{equation}
The matrix \(\bm M\) has Cauchy form. The Cauchy determinant identity
gives
\begin{equation}
\det[\bm M]
=
\frac{
\displaystyle
\prod_{1\leq j<i\leq D}(x_i-x_j)^2
}{
\displaystyle
\left(\prod_{i=1}^{D}(1-x_i^2)\right)
\prod_{1\leq j<i\leq D}(1-x_i x_j)^2
}.
\label{eqCauchyDetAnti}
\end{equation}
Substituting Eq.~\eqref{eqCauchyDetAnti} into
Eq.~\eqref{eqZCauchyReductionAnti}, we obtain
\begin{equation}
\det[\bm Z(+)]
=
\prod_{1\leq j<i\leq D}
\frac{(x_i-x_j)^2}{(1-x_i x_j)^2}.
\label{eqdetZplus}
\end{equation}
Using definition of $x_i$, we have
\begin{equation}
\frac{x_i-x_j}{1-x_i x_j}
=
-
\frac{
\sin\!\left[
\dfrac{\pi(m+1)(n_i-n_j)}{2m}
\right]
}{
\sin\!\left[
\dfrac{\pi}{2m}
\bigl((m+1)(n_i+n_j)+1\bigr)
\right]
}.
\label{eqxRatioSineAnti}
\end{equation}
Equation~\eqref{eqZNoddSM} follows from squaring this expression and
taking the product over \(1\leq j<i\leq D\). Since the integers \(n_i\)
are distinct and satisfy \(n_i\leq \bar{n}\), none of the sine factors
vanishes and $Z(\mathcal N)$ is positive.
\end{proof}

Lemma~\ref{eqlemmadetZ} is stated and proved for \(n_i\leq \bar{n}\). It therefore also applies when one exponent reaches the
borderline value \(n_D=\bar{n}\), a case that will be used in
Sec.~\ref{SecDNBorderline}.

Combining Eq.~\eqref{eqInteriorDetBeforeZAnti} with
Lemma~\ref{eqlemmadetZ}, we find
\begin{equation}
\det[\bm J^{(D)}(\omega)]
\sim
\left(
Z(\mathcal N)
\prod_{i=1}^{D}|c_i|^2
\right)
L_{\mathcal N}(\omega).
\label{eqdet_J_D_asymptotic}
\end{equation}

Equation~\eqref{eqdet_J_D_asymptotic} establishes that
\begin{equation}
J(\omega)\sim A\,L_{\mathcal N}(\omega),
\qquad
A=
Z(\mathcal N)\prod_{i=1}^{D}|c_i|^2\neq0,
\label{eqInteriorDNJostPreCriticality}
\end{equation}
for every interior interaction with $D=N$. This result was obtained without assuming that the system is noncritical, whereas noncriticality is a required assumption in Lemma~\ref{lemma1}. We now show that these two statements are consistent by proving that an interior system with $D=N$ is inherently noncritical.

 Equation~\eqref{eqdet_J_D_asymptotic} implies that $\bm J(\omega)=\bm J^{(D)}(\omega)$ is invertible for all sufficiently small, nonzero $\omega$ in either complex half-plane. By Lemma~\ref{definitionCritical0}, it therefore suffices to show that $\|\bm J(\omega)^{-1}\|$ remains bounded as $\omega\to0$. In fact, we will establish the stronger result that this norm vanishes entirely. This asymptotic analysis of $\bm{J}^{(D)}(\omega)^{-1}$ will also prove instrumental in our subsequent treatment of the $D<N$ case.

\begin{lemma}\label{lemmaInteriorDivergentInverse}
For an interior interaction,
\begin{equation}
\bigl[\bm J^{(D)}(\omega)^{-1}\bigr]_{ij}
=
O\!\left(
q_\omega^{m-1-n_i-n_j}
\right),
\qquad
1\leq i,j\leq D.
\label{eqInteriorDivergentInverseScaling}
\end{equation}
Consequently,
\begin{equation}
\left\|
\bm J^{(D)}(\omega)^{-1}
\right\|
=
O\!\left(
q_\omega^{m-1-2n_D}
\right)
\longrightarrow0.
\label{eqInteriorDivergentInverseNorm}
\end{equation}
\end{lemma}

\begin{proof}
The inverse of $\bm J^{(D)}(\omega)$ can be computed using the adjugate formula,
\begin{equation}
\bigl[\bm J^{(D)}(\omega)^{-1}\bigr]_{ij}
=
\frac{
(-1)^{i+j}
\det[\bm J^{(D)}_{\cancel{ji}}(\omega)] 
}{
\det[\bm J^{(D)}(\omega)]
}.
\label{eqInteriorAdjugateAnti}
\end{equation}
where \(\bm J^{(D)}_{\cancel{ji}}(\omega)\)  denotes the submatrix obtained by deleting row $j$ and column $i$ from $\bm J^{(D)}(\omega)$.

Using Eq.~\eqref{eqdet_J_D_asymptotic} and the functional form of $L_{m,2n_i}(\omega)$, we find that the denominator of Eq.~\eqref{eqInteriorAdjugateAnti} satisfies
\begin{equation}
\det[\bm J^{(D)}(\omega)]
=
\Theta\!\left(
q_\omega^{\eta(\mathcal{N})-mD
}
\right)
\label{eqInteriorDivergentDetPower}
\end{equation}
where $\eta(\mathcal{N})$ is given in Eq.~\eqref{eqnAdefi}.  
Furthermore, since every product in the Leibniz expansion of $\det[\bm J^{(D)}_{\cancel{ji}}(\omega)]$ contains exactly one element from each row except $j$ and one element from each column except $i$, the minor satisfies
\eq{
\det[\bm J^{(D)}_{\cancel{ji}}(\omega)]
=
O\!\left(
q_\omega^{
 \eta(\mathcal{N})-mD+m-1-n_i-n_j
}
\right).
\label{eqInteriorDivergentMinorPower}
}
Combining Eqs.~\eqref{eqInteriorDivergentDetPower} and
\eqref{eqInteriorDivergentMinorPower} yields
Eq.~\eqref{eqInteriorDivergentInverseScaling}.  
Equation~\eqref{eqInteriorDivergentInverseNorm} then follows directly.
\end{proof}

For $D=N$, Lemma~\ref{lemmaInteriorDivergentInverse} yields $\|\bm J(\omega)^{-1}\|\to0$, which, according to Lemma~\ref{definitionCritical0}, guarantees that the system is noncritical. Combined with Eq.~\eqref{eqInteriorDNJostPreCriticality}, this proves the interior part of Lemma~\ref{lemma1} for $D=N$.

Although Lemma~\ref{lemmaInteriorDivergentInverse} is stated for the
\(D\times D\) leading principal block, the same result applies to any principal
block spanned by a subset of the divergent directions
\(\{\ket{u_1},\ldots,\ket{u_D}\}\).

\subsubsection{$D<N$\label{SecDlessNinterior}}

When \(D<N\), all four blocks in
Eq.~\eqref{eqInteriorBlockDecompositionAnti} are present. We  
determine the asymptotic behavior of the Jost function and show that 
 its leading coefficient vanishes when the system is critical. 

Equation~\eqref{eqdet_J_D_asymptotic} shows that the divergent
 block \(\bm J^{(D)}(\omega)\) is invertible for all
sufficiently small nonzero \(\omega\) in either complex half-plane.
Therefore, we can use the Schur-complement determinant identity to factorize the determinant of $\bm{J}(\omega)$ as 
\begin{subequations}
\begin{equation}
J(\omega)
=
\det[\bm J^{(D)}(\omega)]
\det[\bm J_{\mathrm{red}}(\omega)],
\label{eqInteriorSchurFactorizationAnti}
\end{equation} where $\bm{J}_{\mathrm{red}}(\omega)$ is the Schur complement
\begin{equation}
\bm J_{\mathrm{red}}(\omega)
\equiv
\bm J^{\perp}(\omega)
-
\bm C(\omega)
[\bm J^{(D)}(\omega)]^{-1}
\bm R(\omega).
\label{eqInteriorReducedJAnti}
\end{equation}
\label{eqJproductAnti}
\end{subequations}

Determining the leading asymptotic behavior of $J(\omega)$ thus relies on the zero-energy limit of $\bm{J}_{\mathrm{red}}(\omega)$.  For \(r,s>D\), the \((r,s)\) entry of the second term in Eq.~\eqref{eqInteriorReducedJAnti} is given by
\begin{equation}
\left[
\bm C(\omega)
\bm J^{(D)}(\omega)^{-1}
\bm R(\omega)
\right]_{rs}
=
\sum_{i,j=1}^{D}
C_{r i}(\omega)
\bigl[\bm J^{(D)}(\omega)^{-1}\bigr]_{ij}
R_{js}(\omega).
\label{eqInteriorSchurCorrectionComponent}
\end{equation}
Using Eq.~\eqref{eqInteriorCrossBlockScalingAnti} and Lemma~\ref{lemmaInteriorDivergentInverse}, every term in this finite sum satisfies
\begin{align}
C_{r i}(\omega)
\bigl[\bm J^{(D)}(\omega)^{-1}\bigr]_{ij}
R_{js}(\omega)
&=
o\!\left(q_\omega^{n_i-\bar{n}}\right)
O\!\left(q_\omega^{m-1-n_i-n_j}\right)
o\!\left(q_\omega^{n_j-\bar{n}}\right)
\nonumber\\
&=
o\!\left(q_\omega^{m-1-2\bar{n}}\right)
=
o(1),
\label{eqInteriorSchurCorrectionPower}
\end{align}
Consequently, this matrix product vanishes entrywise in the zero-energy limit:
\begin{equation}
\bm C(\omega)
\bm J^{(D)}(\omega)^{-1}
\bm R(\omega)
\longrightarrow 0.
\label{eqInteriorSchurCorrectionTendsZero}
\end{equation}
Combined with the limit \( \bm J^{\perp}(\omega)\to\bm J^{\perp}(0) \), this yields
\begin{equation}
\bm J_{\mathrm{red}}(\omega)
=
\bm J^{\perp}(0)+o(1),
\label{eqInteriorReducedJLimitAnti}
\end{equation}
and therefore,
\begin{equation}
\det[\bm J_{\mathrm{red}}(\omega)]
=
\det[\bm J^{\perp}(0)]+o(1).
\label{eqInteriorReducedDetLimitAnti}
\end{equation}
Applying Eqs.~\eqref{eqdet_J_D_asymptotic} and \eqref{eqInteriorReducedDetLimitAnti} to the Schur factorization in Eq.~\eqref{eqInteriorSchurFactorizationAnti} yields
\begin{equation}
J(\omega)
=
\left(
Z(\mathcal N)\prod_{i=1}^{D}|c_i|^2
\right)
\left\{
\det[\bm J^{\perp}(0)]+o(1)
\right\}
L_{\mathcal N}(\omega).
\label{eqInteriorJostAsymptoticAnti}
\end{equation}
Consequently, if $\det[\bm J^{\perp}(0)]\neq0$, we find the leading asymptotic behavior
\begin{equation}
J(\omega)
\sim
A\,L_{\mathcal N}(\omega),
\qquad
A
=
\left(
Z(\mathcal N)\prod_{i=1}^{D}|c_i|^2
\right)
\det[\bm J^{\perp}(0)]
\neq0.
\label{eqInteriorFinalJostAnti}
\end{equation}
To conclude the proof of Lemma~\ref{lemma1}, we  are left to demonstrate that $\det[\bm J^{\perp}(0)]\neq0$ is   the requirement for the system to be noncritical.

\begin{lemma}\label{lemma_inter_crit}
For an interior interaction, the system is critical if and only if
\begin{equation}
D<N
\qquad\text{and}\qquad
\det[\bm J^{\perp}(0)]=0.
\label{eqInteriorCriticalityAnti}
\end{equation}
\end{lemma}

\begin{proof}
The \(D=N\) case was shown in
Sec.~\ref{InteriorDeqNAnti} to be noncritical. We therefore consider
\(D<N\).

Provided that $\bm J_{\mathrm{red}}(\omega)$ is invertible, the inverse of the full matrix $\bm J(\omega)$ is given by the block matrix identity
\begin{equation}
\bm J(\omega)^{-1}
=
\begin{pmatrix}
(\bm J^{(D)})^{-1}
+
(\bm J^{(D)})^{-1}
\bm R
\bm J_{\mathrm{red}}^{-1}
\bm C
(\bm J^{(D)})^{-1}
&
-
(\bm J^{(D)})^{-1}
\bm R
\bm J_{\mathrm{red}}^{-1}
\\[0.5em]
-
\bm J_{\mathrm{red}}^{-1}
\bm C
(\bm J^{(D)})^{-1}
&
\bm J_{\mathrm{red}}^{-1}
\end{pmatrix},
\label{eqInteriorBlockInverseAnti}
\end{equation}
For readability, the explicit $\omega$-dependence has been omitted from the matrix blocks on the right-hand side.

We now analyze the asymptotic behavior of the individual terms appearing in Eq.~\eqref{eqInteriorBlockInverseAnti}.  
For indices \(i\leq D\) and \(r>D\), Eq.~\eqref{eqInteriorCrossBlockScalingAnti} and Lemma~\ref{lemmaInteriorDivergentInverse} yield
\begin{align}
\left[
\bm J^{(D)}(\omega)^{-1}\bm R(\omega)
\right]_{ir}
&=
\sum_{j=1}^{D}
\bigl[\bm J^{(D)}(\omega)^{-1}\bigr]_{ij}
R_{jr}(\omega)
\nonumber\\
&=
o\!\left(q_\omega^{m-1-n_i-\bar{n}}\right)
=
o\!\left(q_\omega^{\bar{n}-n_i}\right)
\longrightarrow0,
\label{eqInteriorInverseRTendsZero}
\\
\left[
\bm C(\omega)\bm J^{(D)}(\omega)^{-1}
\right]_{r i}
&=
\sum_{j=1}^{D}
C_{r j}(\omega)
\bigl[\bm J^{(D)}(\omega)^{-1}\bigr]_{ji}
\nonumber\\
&=
o\!\left(q_\omega^{m-1-n_i-\bar{n}}\right)
=
o\!\left(q_\omega^{\bar{n}-n_i}\right)
\longrightarrow0.
\label{eqInteriorCInverseTendsZero}
\end{align}

First, consider the case where \(\det[\bm J^{\perp}(0)]\neq0\). According to Eq.~\eqref{eqInteriorReducedJLimitAnti}, \(\bm J_{\mathrm{red}}(\omega)\) is invertible for all sufficiently small \(\omega\), and its inverse satisfies
\begin{equation}
\bm J_{\mathrm{red}}(\omega)^{-1}
\longrightarrow
[\bm J^{\perp}(0)]^{-1}.
\label{eqInteriorReducedInverseLimitAnti}
\end{equation}
Substituting Eqs.~\eqref{eqInteriorDivergentInverseScaling}, \eqref{eqInteriorInverseRTendsZero}, and \eqref{eqInteriorCInverseTendsZero} into the block identity of Eq.~\eqref{eqInteriorBlockInverseAnti}, we obtain the limit
\begin{equation}
\bm J(\omega)^{-1}
\longrightarrow
\begin{pmatrix}
\bm 0 & \bm 0\\
\bm 0 & [\bm J^{\perp}(0)]^{-1}
\end{pmatrix}.
\label{eqInteriorFullInverseLimitAnti}
\end{equation}
Equation~\eqref{eqInteriorFullInverseLimitAnti} demonstrates that $\|\bm J(\omega)^{-1}\|$ remains bounded as $\omega \to 0$. By Lemma~\ref{definitionCritical0}, the system is therefore noncritical.

Conversely, suppose that $\det[\bm J^{\perp}(0)]=0$. To prove that the system is critical, Lemma~\ref{definitionCritical0} states that it suffices to find a sequence $\omega_\nu\to0$ (lying outside the continuous spectrum) along which either $\bm J(\omega_\nu)$ is singular or $\|\bm J(\omega_\nu)^{-1}\|\to\infty$. We consider two distinct scenarios for $\bm J_{\mathrm{red}}(\omega)$.

If $\bm J_{\mathrm{red}}(\omega)$ is singular at points arbitrarily close to zero, we can construct a sequence $\omega_\nu \to 0$ from these points. The Schur factorization in Eq.~\eqref{eqInteriorSchurFactorizationAnti} then yields $\det[\bm J(\omega_\nu)]=0$, implying that the full matrix $\bm J(\omega_\nu)$ is singular along this sequence. This satisfies the first condition of Lemma~\ref{definitionCritical0}, establishing criticality.

Alternatively, suppose $\bm J_{\mathrm{red}}(\omega)$ is invertible for all nonzero $\omega$ in some deleted neighborhood of the origin (outside the continuous spectrum). Because $\bm J^{(D)}(\omega)$ is also invertible near $\omega=0$, the Schur determinant identity, $\det[\bm J(\omega)]=\det[\bm J^{(D)}(\omega)]\det[\bm J_{\rm red}(\omega)]$, ensures that $\bm J(\omega)$ is invertible in this neighborhood. 
According to Eq.~\eqref{eqInteriorBlockInverseAnti}, $\bm {J}_{\mathrm{red}}(\omega)^{-1}$ is the lower-right block of $\bm J(\omega)^{-1}$. Consequently, the operator norm of the full inverse is bounded below by that of its submatrix. Using Eq.~\eqref{eqInteriorReducedJLimitAnti} and the continuity of the smallest singular value, we have
\begin{equation}
\sigma_{\min}[\bm J_{\mathrm{red}}(\omega)]
\longrightarrow
\sigma_{\min}[\bm J^{\perp}(0)]
=0.
\label{eqantiSymDlessJredSingularValue}
\end{equation}
Therefore, combining Eq.~\eqref{eqInverseNormSmin} with Eq.~\eqref{eqantiSymDlessJredSingularValue}, we find that along any valid sequence $\omega_\nu\to0$,
\begin{equation}
\|\bm J(\omega_\nu)^{-1}\|
\geq
\|\bm J_{\mathrm{red}}(\omega_\nu)^{-1}\|
=
\frac{1}{\sigma_{\min}[\bm J_{\mathrm{red}}(\omega_\nu)]}
\longrightarrow\infty.
\label{eqantiSymDlessJredInvDivergence}
\end{equation}
Equation~\eqref{eqantiSymDlessJredInvDivergence} satisfies the second condition for criticality in Lemma~\ref{definitionCritical0}. Thus, the system is critical in both scenarios, completing the proof of Lemma~\ref{lemma_inter_crit}.

\end{proof}

Combining the criticality condition established in Lemma~\ref{lemma_inter_crit} with Eq.~\eqref{eqInteriorFinalJostAnti} proves Lemma~\ref{lemma1} for interior interactions with $D<N$. Together with our earlier result for the $D=N$ case, this completes the proof of Lemma~\ref{lemma1} for all interior interactions.

\subsection{Borderline Interactions\label{SecDNBorderline}}

We now prove the borderline case of Lemma~\ref{lemma1}, for which $n_D=\bar{n}=\frac{m-1}{2}$. In this regime, the block $\bm{J}^{(D)}(\omega)$ is no longer entirely divergent. Specifically, the matrix element $J_{DD}(\omega)$ approaches distinct finite limits as $\omega\to 0$ from the upper and lower complex half-planes. We refer to this entry as the borderline matrix element. Its asymptotic behavior is given by
\begin{subequations}
\label{eqJDDborderline}
\begin{align}
J_{DD}(\omega)
&=
|c_D|^2L_{m,m-1}(\omega)+j_{\rm reg}+o(1),
\label{eqJDDborderlineAsymptotic}
\\
j_{\rm reg}
&\equiv
W_{DD}
+
\operatorname{PV}\!\int_{-\infty}^{+\infty}dk\,
\frac{|\braket{u_D|v(k)}|^2}{-\operatorname{sgn}(k)|k|^m},
\label{eqyBborderline}
\end{align}
\end{subequations}
where $L_{m,m-1}(\omega)
=-\operatorname{sgn}(\operatorname{Im}[\omega])\frac{\mathrm{i}\pi}{m}$.

To isolate the strictly divergent terms, we introduce a modified block partition of the $\bm J$-matrix, splitting it after the first \(D-1\) basis vectors:
\begin{equation}
\bm J(\omega)
=
\begin{pmatrix}
\bm J^{(D-1)}(\omega) & \bm R^{(D-1)}(\omega)\\
\bm C^{(D-1)}(\omega) & \bm J^{(D-1)\perp}(\omega)
\end{pmatrix}.
\label{eqBorderlineFiniteBlockDecomp}
\end{equation}
Here, \(\bm J^{(D-1)}(\omega)\) is a $(D-1)\times (D-1)$ matrix spanning the indices \(i,j\leq D-1\), while \(\bm J^{(D-1)\perp}(\omega)\) is an $(N-D+1)\times (N-D+1)$ matrix spanning the indices \(r,s\geq D\).
By construction, all matrix elements in the upper-left block \(\bm J^{(D-1)}(\omega)\) are infrared-divergent, scaling as
\begin{equation}
J_{ij}(\omega)
\sim
c_i c_j^*L_{m,n_i+n_j}(\omega)
=
\Theta\!\left(q_\omega^{-m+1+n_i+n_j}\right),
\qquad i,j\leq D-1.
\label{eqBorderlineDivergentEntriesAnti}
\end{equation}
The borderline element $J_{DD}(\omega)$ occupies the top-left corner of the complementary block $ \bm{J}^{(D-1)\perp}(\omega)$, which remains finite in the limit $\omega\to 0$. For \(D=N\), this block reduces to the scalar \( \bm{J}^{(D-1)\perp}(\omega)=J_{DD}(\omega)\). For \(D<N\), all other matrix elements in $ \bm{J}^{(D-1)\perp}(\omega)$ approach identical finite limits from both complex half-planes. 

We now summarize the asymptotic properties of the determinant and inverse of the divergent block $\bm{J}^{(D-1)}(\omega)$. For \(D>1\), Lemma~\ref{lemmaInteriorDivergentInverse} yields
\begin{equation}
\bigl[\bm J^{(D-1)}(\omega)^{-1}\bigr]_{ij}
=
O\!\left(q_\omega^{m-1-n_i-n_j}\right),
\qquad i,j<D.
\label{eqJlessBorderlineInverseScaling}
\end{equation}

Finally, the off-diagonal blocks connect the infrared-divergent and finite subspaces. For indices \(i<D\) and \(r\geq D\), their matrix elements scale as
\begin{equation}
R^{(D-1)}_{ir}(\omega)
=
\begin{cases}
O\!\left(q_\omega^{n_i-\bar{n}}\right), & r=D,\\
o\!\left(q_\omega^{n_i-\bar{n}}\right), & r>D,
\end{cases}
\qquad
C^{(D-1)}_{r i}(\omega)
=
\begin{cases}
O\!\left(q_\omega^{n_i-\bar{n}}\right), & r=D,\\
o\!\left(q_\omega^{n_i-\bar{n}}\right), & r>D.
\end{cases}
\label{eqBorderlineCrossBlockScalingAnti}
\end{equation}

As in the case of interior interactions, our primary tools for this proof are the Jost function $J(\omega)$ and the inverse of the $\bm J$-matrix. For $D>1$, Lemma~\ref{lemmaInteriorDivergentInverse} guarantees that   $\bm J^{(D-1)}(\omega)$ is invertible near zero energy, allowing us to write the Jost function via the Schur factorization
\begin{equation}
J(\omega)
=
\det[\bm J^{(D-1)}(\omega)]
\det[\bm J_{\mathrm{red}}(\omega)],
\label{eqJblockborderlineSchur}
\end{equation}
where the Schur complement is given by
\begin{equation}
\bm J_{\mathrm{red}}(\omega)
\equiv
\bm J^{(D-1)\perp}(\omega)
-
\bm C^{(D-1)}(\omega)
\bm J^{(D-1)}(\omega)^{-1}
\bm R^{(D-1)}(\omega).
\label{eqFiniteSchurBorderline}
\end{equation}
When \(D=1\), the divergent block is absent. In this case, we adopt the conventions
\begin{equation}
\det[\bm J^{(0)}]=1,
\qquad
\bm J_{\mathrm{red}}(\omega)=\bm J(\omega),
\label{eqDOneBorderlineConventions}
\end{equation}
ensuring that Eq.~\eqref{eqJblockborderlineSchur} remains valid for all \(D\geq 1\).

Combining Eqs.~\eqref{eqJlessBorderlineInverseScaling} and \eqref{eqBorderlineCrossBlockScalingAnti}, we determine the asymptotic scaling of the second term in Eq.~\eqref{eqFiniteSchurBorderline}:
\begin{equation}
\left[
\bm C^{(D-1)}(\omega)
\bm J^{(D-1)}(\omega)^{-1}
\bm R^{(D-1)}(\omega)
\right]_{rs}
=
\begin{cases}
O(1), & r=s=D,\\
o(1), & (r,s)\neq(D,D).
\end{cases}
\label{eqBorderlineSchurCorrectionScaling}
\end{equation}
For $r,s\geq D$ with $(r,s)\neq(D,D)$, define the common finite limit $J_{rs}(0)\equiv\lim_{\omega\to0}J_{rs}(\omega)$, independent of the complex half-plane. Substituting Eq.~\eqref{eqBorderlineSchurCorrectionScaling} into Eq.~\eqref{eqFiniteSchurBorderline} produces
\begin{equation}
[\bm J_{\mathrm{red}}(\omega)]_{rs}
= J_{rs}(\omega)+o(1)
\longrightarrow  
J_{rs}(0).
\label{eqBorderlineReducedOffDD}
\end{equation}
Therefore, it only remains to determine the limiting behavior of the borderline element $[\bm J_{\mathrm{red}}(\omega)]_{DD}$. 

\begin{lemma}\label{lemmaJredDD}
The top-left entry of the reduced matrix $\bm J_{\mathrm{red}}(\omega)$ can be computed  by the ratio
\begin{equation}
[\bm J_{\mathrm{red}}(\omega)]_{DD}
=
\frac{\det[\bm J^{(D)}(\omega)]}
{\det[\bm J^{(D-1)}(\omega)]}.
\label{eqReducedJBorderlineDD}
\end{equation}
Evaluating Eq.~\eqref{eqReducedJBorderlineDD} in the limit $\omega\to 0$ yields two distinct finite limits, $J_{\rm red}^\pm$, depending on the complex half-plane from which the limit is taken:
\begin{equation}
[\bm J_{\mathrm{red}}(\omega)]_{DD} \longrightarrow J_{\rm red}^\pm, \qquad
J_{\rm red}^\pm =
j_{\rm reg}\mp\frac{\mathrm{i}\pi}{m}\frac{Z(\mathcal N)}{Z(\mathcal N^<)}|c_D|^2, \qquad \pm\operatorname{Im}[\omega]>0, 
\label{eqjDpmBorderline} 
\end{equation}
where $j_{\rm reg}$ is defined in Eq.~\eqref{eqyBborderline}, and the ratio 
\begin{equation*}
  \frac{Z(\mathcal N)}{Z(\mathcal N^<)}
=
\prod_{n\in\mathcal N^<}
\tan^2\!\left[
\frac{\pi(m+1)}{2m}(\bar{n}-n)
\right]
\end{equation*}
is  positive.
\end{lemma}

 \begin{proof}
 We first prove Eq.~\eqref{eqReducedJBorderlineDD}. Applying the determinant formula for partitioned matrices, the ratio of the determinants on the right-hand side of Eq.~\eqref{eqReducedJBorderlineDD} evaluates to the Schur complement:
\begin{equation}
\frac{\det[\bm J^{(D)}(\omega)]}{\det[\bm J^{(D-1)}(\omega)]}
=
J_{DD}(\omega)-\sum_{i,j=1}^{D-1}J_{Di}(\omega)\bigl[\bm{J}^{(D-1)}(\omega)^{-1}\bigr]_{ij} J_{jD}(\omega). 
\label{eqJDinvDD}
\end{equation}
Concurrently, from the definition in Eq.~\eqref{eqFiniteSchurBorderline}, the $(D,D)$ element of $\bm J_{\mathrm{red}}(\omega)$ is explicitly given by
\begin{equation}
[\bm J_{\mathrm{red}}(\omega)]_{DD}
=
J_{DD}(\omega)-\sum_{i,j=1}^{D-1}[\bm{C}^{(D-1)}(\omega)]_{Di} \bigl[\bm{J}^{(D-1)}(\omega)^{-1}\bigr]_{ij}[\bm{R}^{(D-1)}(\omega)]_{jD}. 
\label{eqJredinvDD}
\end{equation}
Identifying that the submatrix elements are exactly $[\bm{C}^{(D-1)}(\omega)]_{Di}=J_{Di}(\omega)$ and $[\bm{R}^{(D-1)}(\omega)]_{jD}=J_{jD}(\omega)$, we see that the right-hand sides of Eqs.~\eqref{eqJDinvDD} and \eqref{eqJredinvDD} are strictly identical. Hence, Eq.~\eqref{eqReducedJBorderlineDD} is established.

We now proceed to evaluate $[\bm J_{\mathrm{red}}(\omega)]_{DD}$ using Eq.~\eqref{eqReducedJBorderlineDD}. Following the derivation of $\det[\bm J^{(D)}(\omega)]$ for interior interactions [cf.\ Eq.~\eqref{eqInteriorDetBeforeZAnti}], the denominator of Eq.~\eqref{eqReducedJBorderlineDD} is given asymptotically by
\begin{equation}
\det[\bm J^{(D-1)}(\omega)]=
[A_{D-1} +o(1)]L_{\mathcal N^<}(\omega), \quad A_{D-1} 
\equiv
Z(\mathcal N^<)
\prod_{i=1}^{D-1}|c_i|^2. \label{eqJDm1boderline}
\end{equation}

To compute the numerator $\det[\bm J^{(D)}(\omega)]$, we define a $D\times D$ auxiliary matrix \( \bm{J}^{\rm aux}(\omega)\) constructed  from the infrared model functions:
\begin{equation}
J^{\rm aux}_{ij}(\omega)
=
c_i c_j^*L_{m,n_i+n_j}(\omega),
\qquad
1\leq i,j\leq D. \label{eqJprimeBorderline}
\end{equation}
Its determinant can be evaluated using a procedure analogous to Eq.~\eqref{eqInteriorDetBeforeZAnti}. Noting that  Lemma~\ref{eqlemmadetZ} applies for all $n_i\leq \bar{n}$, we find
\begin{equation}
\det[\bm J^{\rm aux}(\omega)]=
[A_{D} +o(1)]L_{\mathcal N^<}(\omega) L_{m,m-1}(\omega), \quad A_{D} 
\equiv
Z(\mathcal N)
\prod_{i=1}^{D}|c_i|^2. \label{eqJaux1boderline} 
\end{equation}
The block $\bm{J}^{(D)}(\omega)$ asymptotically approaches \( \bm{J}^{\rm aux}(\omega)\) everywhere except for its finite, borderline diagonal element $J_{DD}(\omega)$. Specifically, we can express this relationship as
\begin{equation}
    J_{ij}^{(D)}(\omega) = 
    \begin{cases}
        J_{ij}^{\mathrm{aux}}(\omega)[1 + o(1)], & (i,j) \neq (D,D), \\
        J_{DD}^{\mathrm{aux}}(\omega) + j_{\mathrm{reg}} + o(1), & (i,j) = (D,D).
\end{cases}\label{eqJDJaux}
\end{equation}

 Expanding the determinant of $\bm{J}^{(D)}(\omega)$ from Eq.~\eqref{eqJDJaux} by linearity in its $D$-th column yields
\begin{equation}
\det[\bm{J}^{(D)}(\omega)]
=
\det[\bm{J}^{\rm aux}(\omega)]
+
j_{\rm reg} \,\det[\bm J^{(D-1)}(\omega)] +o(L_{\mathcal{N}^<}(\omega)).
\label{eqdetHsignsecC_revised}
\end{equation}

Inserting Eqs.~\eqref{eqJDm1boderline} and \eqref{eqJaux1boderline} into Eq.~\eqref{eqdetHsignsecC_revised}, and using the limit $L_{m,m-1}(\mathrm{i}0^\pm)=\mp \frac{\mathrm{i}\pi}{m}$, we obtain
\begin{equation}
\det[\bm{J}^{(D)}(\omega)]
=
A_{D-1}
\left[
J_{\rm red}^\pm+o(1)
\right]
L_{\mathcal N^<}(\omega), \qquad \pm\operatorname{Im}[\omega]>0,
\label{eqJDNBorderlineFactorized}
\end{equation}
where $J_{\rm red}^\pm$ is defined in Eq.~\eqref{eqjDpmBorderline}. Finally, substituting the asymptotic expressions for $\det[\bm{J}^{(D-1)}(\omega)]$ [Eq.~\eqref{eqJDm1boderline}] and $\det[\bm{J}^{(D)}(\omega)]$ [Eq.~\eqref{eqJDNBorderlineFactorized}] into Eq.~\eqref{eqReducedJBorderlineDD}, the common factor of $A_{D-1} L_{\mathcal N^<}(\omega)$ cancels, yielding Eq.~\eqref{eqjDpmBorderline}.
\end{proof}

Combining Eq.~\eqref{eqBorderlineReducedOffDD} with Lemma~\ref{lemmaJredDD}, we define the finite half-plane limits of $\bm{J}_{\rm red}(\omega)$ as $\bm J_{\mathrm{red}}^\pm$, such that
\begin{equation}
\bm{J}_{\rm red}(\omega)\longrightarrow \bm J_{\mathrm{red}}^\pm,  \qquad \pm\operatorname{Im}[\omega]>0.  \label{Jredlimitborderline}
\end{equation}
These two limiting matrices $\bm J_{\mathrm{red}}^\pm$ differ only in their borderline \((D,D)\) entry.
When \(D=N\), \(\bm{J}_{\mathrm{red}}^\pm=J_{\mathrm{red}}^\pm\) is a scalar. 
For \(D<N\), the matrix elements with \(D\le r,s\le N\) are given by
\begin{equation}
[\bm J_{\mathrm{red}}^\pm]_{rs}
=
\begin{cases}
J_{\rm red}^\pm, & r=s=D,\\
 J_{rs}(0), & (r,s)\neq(D,D).
\end{cases}
\label{eqJredLimitEntries}
\end{equation}

Substituting Eqs.~\eqref{eqJDm1boderline} and \eqref{Jredlimitborderline} into the Schur factorization of Eq.~\eqref{eqJblockborderlineSchur}, we recover Eq.~\eqref{eqJomegacasesSM} in Lemma~\ref{lemma1}, provided that the coefficients $A^\pm \equiv A_{D-1}\det[\bm J_{\mathrm{red}}^\pm]$ are nonzero. To complete the proof for the borderline interaction, it remains only to establish that the vanishing of these determinants is equivalent to the criticality of the system.

\begin{lemma}\label{lemma_criti_anti_borderline}
For a borderline interaction, the system is critical if and only if
\begin{equation}
\det[\bm J_{\mathrm{red}}^+]=0
\qquad\text{or}\qquad
\det[\bm J_{\mathrm{red}}^-]=0.
\label{eqBorderlineDlessNCriticalCriterion0}
\end{equation}
\end{lemma}

\begin{proof}
$\bm{J}^{(D-1)}$ is invertible near zero energy. If $\bm{J}_{\rm red}(\omega)$ is invertible, the inverse of the full matrix can be expressed via the block matrix identity
\begin{equation}
\bm J(\omega)^{-1}
=
\begin{pmatrix}
(\bm J^{(D-1)})^{-1}
+
(\bm J^{(D-1)})^{-1}
\bm R^{(D-1)}
\bm J_{\mathrm{red}}^{-1}
\bm C^{(D-1)}
(\bm J^{(D-1)})^{-1}
&
-
(\bm J^{(D-1)})^{-1}
\bm R^{(D-1)}
\bm J_{\mathrm{red}}^{-1}
\\[0.5em]
-
\bm J_{\mathrm{red}}^{-1}
\bm C^{(D-1)}
(\bm J^{(D-1)})^{-1}
&
\bm J_{\mathrm{red}}^{-1}
\end{pmatrix},
\label{eqBorderlineBlockInverseAnti}
\end{equation}
where the explicit $\omega$-dependence has been omitted on the right-hand side for brevity. Combining Eqs.~\eqref{eqJlessBorderlineInverseScaling} and \eqref{eqBorderlineCrossBlockScalingAnti}, we find that for \(D>1\),
\begin{equation}
\bm J^{(D-1)}(\omega)^{-1}\bm R^{(D-1)}(\omega)
\longrightarrow \bm 0,
\qquad
\bm C^{(D-1)}(\omega)\bm J^{(D-1)}(\omega)^{-1}
\longrightarrow \bm 0.
\label{eqBorderlineCrossInverseAnti}
\end{equation}
Applying Eqs.~\eqref{eqJlessBorderlineInverseScaling}, \eqref{Jredlimitborderline}, and \eqref{eqBorderlineCrossInverseAnti} to Eq.~\eqref{eqBorderlineBlockInverseAnti}, when $\det[\bm J^\pm_{\rm red}]\neq 0$, the full inverse asymptotically approaches the block-diagonal limit
\begin{equation}
\bm J(\omega)^{-1}
\longrightarrow
\begin{pmatrix}
\bm 0 & \bm 0\\
\bm 0 & (\bm J_{\rm red}^\pm)^{-1}
\end{pmatrix},    \qquad \pm\operatorname{Im}[\omega]>0. 
\label{eqBorderlineFullInverseLimitAnti}
\end{equation}
By an argument analogous to that used for interior interactions [immediately following Eq.~\eqref{eqInteriorFullInverseLimitAnti} in the proof of Lemma~\ref{lemma_inter_crit}], we conclude that the system is critical if and only if $\det[\bm J_{\mathrm{red}}^+]=0$ or $\det[\bm J_{\mathrm{red}}^-]=0$.
\end{proof}

This establishes the criticality condition for borderline interactions. Together with the results for interior interactions derived in the previous subsection, this completes the proof of Lemma~\ref{lemma1} in its entirety.

\section{Symmetric Dispersion: S-matrix and the $N=1$ Case \label{secSymm}}

In this section and Sec.~\ref{SecProofSym}, we prove Theorem~2 for symmetric dispersion. The S-matrix is a $2\times2$ matrix acting on the even- and odd-parity channels, which makes the proof more involved than in the antisymmetric case. We give the general S-matrix formulas and introduce the infrared model functions $\breve L_{m,l}(\omega)$. We then prove Theorem~2 for $N=1$ to illustrate the main strategy that will be generalized to arbitrary $N$ in Sec.~\ref{SecProofSym}. Throughout this section and  Sec.~\ref{SecProofSym}, the limit $\omega\to0$ is taken in $\mathbb C\setminus[0,\infty)$.

For symmetric dispersion, scattering between the even- and odd-parity channels is described by the S-matrix in Eq.~(10) of the main text, while Eq.~\eqref{SratioJ} relates the determinant of the S-matrix to the Jost function. For $E>0$ and channel labels $\lambda,\lambda'\in\{\mathrm e,\mathrm o\}$, its matrix elements are
\begin{equation}
S_{\lambda\lambda'}(E)
=
\delta_{\lambda\lambda'}
+2\pi \mathrm{i}\,\rho(E)\,
\langle v_{\lambda'}(k_E)|\bm J(E+\mathrm{i}0)^{-1}|v_{\lambda}(k_E)\rangle.
\label{eqgvT0}
\end{equation}
Here $|v_{\mathrm{e}}(k)\rangle \equiv (|v(k)\rangle + |v(-k)\rangle)/\sqrt{2}$ and $|v_{\mathrm{o}}(k)\rangle \equiv (|v(k)\rangle - |v(-k)\rangle)/\sqrt{2}$  are the even and odd channel couplings, respectively.  
The on-shell momentum and density of states are $k_E=|E|^{1/m}$ and $\rho(E)=m^{-1}|E|^{-1+1/m}$. The $J$-matrix $\bm J(\omega)$ is given in Eq.~\eqref{eqMKrelations}, with $\bm K(\omega)$ defined by Eq.~\eqref{eqKdefi} for $\epsilon(k)=|k|^m$.

We now prove Theorem~2 of the main text for $N=1$. We consider a coupling coefficient (or form factor) $v(k)=ck^n+o(k^n)$, where $n\geq0$ is an integer and $c\neq0$. Our objective is to establish the following zero-energy scattering and criticality results. For $n>\bar{n}$, every noncritical system has the identity S-matrix at zero energy. For $n\leq \bar{n}$, the system is necessarily noncritical. In this case, let $p\in\{\mathrm e,\mathrm o\}$ denote the parity of $n$, and let $\bar p$ denote the opposite parity. We will prove that
\begin{subequations}
\eq{
S_{pp}(0^{+}) &=\exp\left(2\pi \mathrm{i}\, \frac{2n+1}{m}\right) \label{eqSymmN1goal1}
}
\eq{
S_{\bar{p}p}(0^{+}) &= S_{p\bar{p}}(0^{+}) = 0, \quad S_{\bar{p}\bar{p}}(0^{+}) = 1.\label{eqSymmN1goal2}
}
\end{subequations}
In the borderline case $n=\bar{n}$, Eq.~\eqref{eqSymmN1goal1} gives $S_{pp}(0^+)=1$, so the S-matrix is also the identity matrix.

For $N=1$, $\bm J(\omega)=J(\omega)$ and $\bm K(\omega)=K(\omega)$ are scalar functions, while $\bm W=W$ is a complex number. Eq.~\eqref{eqgvT0} becomes
\begin{equation}
S_{\lambda\lambda'}(0^+)
=
\delta_{\lambda\lambda'}
+
\lim_{E\to0^+}
2\pi \mathrm{i}\,\rho(E)\,
\frac{V_{\lambda'}(k_E)^*V_{\lambda}(k_E)}{J(E+\mathrm{i}0)},
\label{eqSN1e}
\end{equation}
where $\lambda,\lambda'\in\{\mathrm e,\mathrm o\}$ and $V_{\mathrm e}(k)=[v(k)+v(-k)]/\sqrt{2}$ and $V_{\mathrm o}(k)=[v(k)-v(-k)]/\sqrt{2}$ are the channel couplings. 

We first consider $n>\bar{n}$. In this case, $K(\omega)$ and $J(\omega)$ have finite limits as $\omega\to0$. Since $J(\omega)$ is scalar, Lemma~\ref{definitionCritical0} implies that for a noncritical system, the boundary value $J(E+\mathrm{i}0)$   approaches a nonzero limit as $E\to0^+$. Moreover, $V_{\lambda}(k_E)=o(k_E^{\bar{n}})$ and $\rho(E)=\Theta(k_E^{1-m})$. Since $2\bar{n}=m-1$, the second term in Eq.~\eqref{eqSN1e} is $o(1)$ for every pair of channels. Hence $S_{\lambda\lambda'}(0^+)=\delta_{\lambda\lambda'}$.

We next consider $n\leq \bar{n}$, for which $K(\omega)$ and $J(\omega)$ diverge as $\omega\to0$. To characterize these divergences, we introduce the infrared model functions $\breve L_{m,l}(\omega)$ for integers $l$ satisfying $0\leq l<m-1$,
\begin{equation}
\breve{L}_{m,l}(\omega)
=
\int_{-\infty}^{+\infty}dk\,
\frac{k^{l}}{\omega-|k|^m}.
\label{eqSLmndefinition1}
\end{equation}
For $\omega=|\omega|e^{\mathrm{i}\theta}$ with $\theta\in(0,2\pi)$ and $q_\omega=|\omega|^{1/m}e^{\mathrm{i}\theta/m}$, Eq.~\eqref{eqSLmndefinition1} gives
\begin{equation}
\breve{L}_{m,l}(\omega)
=
\begin{cases}
0, & \text{odd }l,\\[0.3em]
-\dfrac{2\pi \mathrm{i}}{m}\breve{\kappa}_{m,l}\,
q_\omega^{-m+1+l}, & \text{even }l\neq m-1,
\end{cases}
\label{eqKsum2}
\end{equation}
where
\begin{equation}
\breve{\kappa}_{m,l}
=
\frac{2}{1-\exp\left[2\pi \mathrm{i}(l+1)/m\right]}.
\label{eqkappasym}
\end{equation}

For $l\geq m-1$, the integral in Eq.~\eqref{eqSLmndefinition1} requires ultraviolet regularization. At $l=m-1$, we therefore define the infrared model function directly by
\eq{
\breve{L}_{m,m-1}(\omega)
\equiv
\begin{cases}
0, & \text{odd }m-1,\\[0.3em]
\dfrac{2}{m}\log|\omega|
+\dfrac{2\mathrm{i}}{m}[\arg(\omega)-\pi], & \text{even }m-1.
\end{cases}
\label{eqSymmLbrevedefiborder}
}
The branch in Eq.~\eqref{eqSymmLbrevedefiborder} is fixed by $\arg(\omega)\in(0,2\pi)$. For any ultraviolet regularization that leaves the integrand unchanged near $k=0$, Eq.~\eqref{eqSymmLbrevedefiborder} gives the nonanalytic infrared part of the regularized integral.
The borderline case $n=\bar{n}$ can occur only for odd integer $m$. In this case, $m-1$ is even, and the corresponding model function in Eq.~\eqref{eqSymmLbrevedefiborder} diverges logarithmically as $\omega\to0$. 

Compared with the antisymmetric model functions, the symmetric functions   have three distinguishing properties: $\breve L_{m,l}(\omega)$ vanishes for odd $l$, the borderline case $l=m-1$ is logarithmically divergent, and the coefficient $\breve\kappa_{m,l}$ in Eq.~\eqref{eqkappasym} is independent of the half-plane containing $\omega$.

We now return to the $N=1$ proof and use these infrared model functions to characterize the low-energy behavior of the Jost function.
For $n\leq \bar{n}$,
\begin{equation}
J(\omega)
=
|c|^2\breve{L}_{m,2n}(\omega)
+o\!\left(\breve{L}_{m,2n}(\omega)\right).
\label{eqLinvK}
\end{equation}
Since $|J(\omega)|\to\infty$, Lemma~\ref{definitionCritical0} shows that every system with $n\leq \bar{n}$ is noncritical.

We now determine the zero-energy S-matrix. Start with the interior case $n<\bar{n}$.  We first obtain the S-matrix determinant from Eqs.~\eqref{SratioJ}, \eqref{eqKsum2} and \eqref{eqLinvK}, 

\begin{align}
\det[\bm S(0^+)]
&=
\lim_{E\to0^+}
\frac{\breve L_{m,2n}(E-\mathrm{i}0)}{\breve L_{m,2n}(E+\mathrm{i}0)}
\nonumber\\
&=
\exp\left(2\pi \mathrm{i}\,\frac{-m+2n+1}{m}\right)
=
\exp\left(2\pi \mathrm{i}\,\frac{2n+1}{m}\right).
\label{eqdetSymmN1}
\end{align}

We next evaluate the matrix elements in Eq.~\eqref{eqSymmN1goal2}. The parity-resolved coupling components satisfy $V_p(k)=\sqrt{2}ck^n+o(k^n)$ and $V_{\bar p}(k)=o(k^n)$. Consequently,
\begin{equation}
\rho(E)V_{\lambda'}(k_E)^*V_{\lambda}(k_E)
=
\begin{cases}
\Theta\!\left(k_E^{-m+1+2n}\right),
& \lambda=\lambda'=p,\\[0.3em]
o\!\left(k_E^{-m+1+2n}\right),
& \text{otherwise}.
\end{cases}
\label{eqSymmN1NumeratorScaling}
\end{equation}
For $n<\bar{n}$, Eq.~\eqref{eqLinvK} gives $J(E+\mathrm{i}0)=\Theta(k_E^{-m+1+2n})$. Eq.~\eqref{eqSymmN1NumeratorScaling}   shows that the second term in Eq.~\eqref{eqSN1e} vanishes for every matrix element except $S_{pp}$, hence
  we obtain Eq.~\eqref{eqSymmN1goal2}.
Using Eq.~\eqref{eqSymmN1goal2}, $S_{pp}(0^+)$ equals the determinant of the zero-energy S-matrix. Therefore, Eq.~\eqref{eqSymmN1goal1} follows from Eq.~\eqref{eqdetSymmN1}. 

For $n=\bar{n}$, Eq.~\eqref{eqSymmLbrevedefiborder}  gives $J(E+\mathrm{i}0)=\Theta(|\log E|)$, while the numerator of the second term in Eq.~\eqref{eqSN1e} is $O(1)$ for all $\lambda, \lambda'$. Hence the S-matrix approaches an identity matrix. 

Together with the result for noncritical systems with $n>\bar{n}$, this completes the proof of Theorem~2 for $N=1$.

\section{Symmetric Dispersion: Proof of Theorem 2}
\label{SecProofSym}

In this section, we prove Theorem~2, which establishes the zero-energy limit of the $S$-matrix for an arbitrary number $N$ of emitters or form factors. As in the proof for antisymmetric dispersions in Sec.~\ref{secProofOdd}, the analysis relies on two central objects: the Jost function $J(\omega)=\det[\bm{J}(\omega)]$ and the inverse matrix $\bm{J}(\omega)^{-1}$. The Jost function determines $\det[\bm{S}(E)]$ via Eq.~\eqref{SratioJ}, while the inverse $\bm{J}$-matrix both yields the individual $S$-matrix elements via Eq.~\eqref{eqgvT0} and provides the diagnostic criterion for criticality via Lemma~\ref{definitionCritical0}. Throughout this section, the limit $\omega\to0$ is taken in $\mathbb{C}\setminus[0,\infty)$.

The proof for symmetric dispersions differs from the antisymmetric case in three key respects. First, $\bm{S}(E)$ is now a $2\times 2$ matrix. Although its determinant can still be extracted from the branch-cut discontinuity of the Jost function across the real axis (Eq.~\eqref{SratioJ}), the individual matrix elements must be computed explicitly. Second, the infrared-divergent block of $\bm{J}(\omega)$ naturally separates according to the parity of the coupling exponents $n_i$.
Matrix elements linking directions with same-parity coupling exponents scale exactly as their designated infrared reference orders, whereas opposite-parity elements are strictly subleading relative to this reference.
Third, the borderline matrix element exhibits a logarithmic divergence. This stands in contrast to the antisymmetric dispersion treated in Sec.~\ref{SecDNBorderline}, where the borderline element approaches finite, half-plane-dependent limits.

The proof proceeds by analyzing successively more general block structures of $\bm{J}(\omega)$. Section~\ref{subsecJsym} introduces a parity-ordered basis and establishes the low-energy hierarchy of the resulting matrix blocks. Sections~\ref{SecSymSingleParityDN} and \ref{SecSymMixedParityDN} treat the $D=N$ regime, addressing cases where the exponent set $\mathcal{N}$ contains a single parity and mixed parities, respectively. In these sections, we demonstrate that systems with $D=N$ are unconditionally noncritical. Finally, Sec.~\ref{SecSymDlessN} considers the $D<N$ regime, where an infrared-finite subspace is present. It is the coupling to this finite subspace that introduces the possibility of tuning the system to a critical point.

For any finite subset $\mathcal{A}\subseteq\mathcal{N}$, we define
\begin{equation}
    \breve{L}_{\mathcal{A}}(\omega) \equiv \prod_{n\in\mathcal{A}}\breve{L}_{m,2n}(\omega), \qquad \breve{L}_{\varnothing}(\omega)\equiv 1.
    \label{eqSymLset}
\end{equation}
The breve accent ($\breve{\ }$) distinguishes the symmetric model functions in Eq.~\eqref{eqSymLset} from the corresponding antisymmetric products $L_{\mathcal{A}}(\omega)$.

\subsection{Asymptotic behavior of the J-matrix and the parity-ordered basis}
\label{subsecJsym}

We begin by selecting a basis that cleanly isolates the infrared divergences and parity structures of $\bm{J}(\omega)$, allowing us to determine the low-energy asymptotic order of each matrix block.

\subsubsection{Parity-ordered basis and block structure}

Lemma~1 of the main text yields distinct integers $n_i\in\mathcal{N}$ for $1\leq i\leq D$, along with a set of orthonormal vectors $\{|u_i\rangle\}_{i=1}^{D}$ such that
\begin{equation}
    V_i(k)\equiv \langle u_i|v(k)\rangle = c_i k^{n_i}+o(k^{n_i}), \qquad c_i\neq0, \qquad n_i\leq \bar{n}\equiv\frac{m-1}{2}.
    \label{eqSymViLeading}
\end{equation}

The vectors $|u_1\rangle,\ldots,|u_D\rangle$ span the infrared-divergent subspace. We extend these to an orthonormal basis $\{|u_i\rangle\}_{i=1}^{N}$ spanning the full emitter or form-factor space. The remaining $N-D$ vectors, $\{|u_r\rangle\}_{r=D+1}^{N}$, span the orthogonal complement. We refer to this complement as the infrared-finite subspace because the restriction of $\bm{J}(\omega)$ to this subspace approaches a finite limit as $\omega\to0$. 

We reorder the first $D$ basis vectors according to the parity of $n_i$. The $N_{\mathrm{e}}$ even exponents are placed first and the $N_{\mathrm{o}}$ odd exponents second, with each group sorted in ascending order such that $D=N_{\mathrm{e}}+N_{\mathrm{o}}$. Let $\mathcal{I}_{\mathrm{e}}$ and $\mathcal{I}_{\mathrm{o}}$ denote the index sets for the even and odd infrared-divergent directions, respectively, and let $\mathcal{I}_{\perp}$ denote the indices of the basis vectors spanning the infrared-finite subspace. Throughout, we use indices $i,j\leq D$ for the infrared-divergent subspace and $r,s\in\mathcal{I}_{\perp}$ for the infrared-finite subspace. 

For $p\in\{\mathrm{e},\mathrm{o}\}$, we define the sets of coupling exponents 
\begin{equation}
    \mathcal{N}_p\equiv\{n_i:i\in\mathcal{I}_p\}, \qquad
    \mathcal{N}_p^{<}\equiv\{n\in\mathcal{N}_p:n<\bar{n}\}.
    \label{eqSymNpDef}
\end{equation}
The sets in Eq.~\eqref{eqSymNpDef} separate the two parities, with the superscript $<$ indicating the strict exclusion of the borderline exponent $\bar{n}$. We adopt the standard convention that an empty product equals one and an empty sum equals zero.

In this parity-ordered basis, the $\bm{J}$-matrix takes the block form
\begin{equation}
    \bm{J}(\omega) =
    \begin{pmatrix}
    \bm{J}^{\mathrm{e}}(\omega) & \bm{J}^{\mathrm{eo}}(\omega) & \bm{R}^{\mathrm{e}}(\omega)\\
    \bm{J}^{\mathrm{oe}}(\omega) & \bm{J}^{\mathrm{o}}(\omega) & \bm{R}^{\mathrm{o}}(\omega)\\
    \bm{C}^{\mathrm{e}}(\omega) & \bm{C}^{\mathrm{o}}(\omega) & \bm{J}^{\perp}(\omega)
    \end{pmatrix}
    =
    \begin{pmatrix}
    \bm{J}^{(D)}(\omega) & \bm{R}(\omega)\\
    \bm{C}(\omega) & \bm{J}^{\perp}(\omega)
    \end{pmatrix},
    \label{eqSymThreeBlockJ}
\end{equation}
where
\begin{equation}
    \bm{J}^{(D)}(\omega) \equiv
    \begin{pmatrix}
    \bm{J}^{\mathrm{e}}(\omega) & \bm{J}^{\mathrm{eo}}(\omega)\\
    \bm{J}^{\mathrm{oe}}(\omega) & \bm{J}^{\mathrm{o}}(\omega)
    \end{pmatrix}
    \label{eqSymJDdef}
\end{equation}
is the $D\times D$ leading principal block   acting on the infrared-divergent subspace. The infrared-finite block $\bm{J}^{\perp}(\omega)$ acts on the infrared-finite subspace, while $\bm{R}(\omega)$ and $\bm{C}(\omega)$ couple the two subspaces. When the Hamiltonian is Hermitian, $\bm{R}(\omega)=\bm{C}(\omega^*)^\dagger$.

\subsubsection{Low-energy scaling}
We now determine the asymptotic scaling of each block in Eq.~\eqref{eqSymThreeBlockJ}. As $\omega \to 0$, the infrared divergences of $\bm{J}(\omega)$ are dictated by the projected-resolvent matrix $\bm{K}(\omega)$. 

We first consider the block $\bm{J}^{(D)}(\omega)$ acting on the infrared-divergent subspace. From Eq.~\eqref{eqKdefi}, its projected-resolvent matrix elements are
\begin{equation}
    K_{ij}(\omega) = \int_{-\infty}^{\infty} dk\,\frac{V_i(k)V_j(k)^*}{\omega-|k|^m}.
    \label{eqSymKijIntegral}
\end{equation}
By Eq.~\eqref{eqSymViLeading}, $V_i(k)V_j(k)^*=c_i c_j^*k^{n_i+n_j}+o(k^{n_i+n_j})$ as $k\to0$.
First, consider the same-parity blocks, $\bm{J}^{\mathrm{e}}(\omega)$ and $\bm{J}^{\mathrm{o}}(\omega)$. For $i,j\in\mathcal{I}_p$ with the same parity $p\in\{\mathrm{e},\mathrm{o}\}$, the leading numerator term $c_i c_j^*k^{n_i+n_j}$ in the integrand is an even function of $k$. We obtain
\begin{equation}
    J_{ij}(\omega) = c_i c_j^*\breve{L}_{m,n_i+n_j}(\omega)+o\!\left(\breve{L}_{m,n_i+n_j}(\omega)\right).
    \label{eqSymJsameParity}
\end{equation}
If $n_i+n_j< m-1$, Eq.~\eqref{eqSymJsameParity} exhibits a pure algebraic scaling:
\begin{equation}
    J_{ij}(\omega) = \Theta\!\left(q_\omega^{1-m+n_i+n_j}\right).
    \label{eqSymJsameParityPower}
\end{equation}
If the interaction is borderline (i.e., $\bar{n}\in\mathcal{N}$), let ${\iota}$ denote the unique index corresponding to the borderline exponent, such that $n_{{\iota}}=\bar{n}$. The corresponding borderline diagonal element diverges logarithmically:
\begin{equation}
    J_{{\iota} {\iota}}(\omega) \sim |c_{{\iota}}|^2\breve{L}_{m,m-1}(\omega) = \Theta\!\left(|\log|\omega||\right).
    \label{eqSymJborderLog}
\end{equation}

We next examine the mixed-parity blocks, $\bm{J}^{\mathrm{eo}}(\omega)$ and $\bm{J}^{\mathrm{oe}}(\omega)$. When $i$ and $j$ have opposite parities, the leading numerator term $c_i c_j^*k^{n_i+n_j}$ in Eq.~\eqref{eqSymKijIntegral} is an odd function of $k$, whereas the denominator $\omega-|k|^m$ is strictly even. Because the integration domain is symmetric, the integral of this leading term evaluates to exactly zero. Consequently, these elements are asymptotically suppressed relative to the algebraic reference order set by Eq.~\eqref{eqSymJsameParityPower}:
\begin{equation}
    J_{ij}(\omega) = o\!\left(q_\omega^{1-m+n_i+n_j}\right), \qquad i\in\mathcal{I}_{\mathrm{e}},\ j\in\mathcal{I}_{\mathrm{o}} \quad \text{or vice versa}.
    \label{eqSymJeoScaling}
\end{equation}

Finally, we consider the diagonal block $\bm{J}^{\perp}(\omega)$ and the coupling blocks $\bm{R}(\omega)$ and $\bm{C}(\omega)$. For indices $r,s\in \mathcal{I}_\perp$ in the orthogonal complement, let $V_r(k)\equiv\langle u_r|v(k)\rangle$. The local estimate $V_r(k)V_s(k)^*=O(|k|^{2\lfloor\bar n\rfloor+2})$ holds as $k\to0$. The projected-resolvent integral is
\begin{equation}
    K_{rs}(\omega) = \int_{-\infty}^{\infty} dk\,\frac{V_r(k)V_s(k)^*}{\omega-|k|^m}.
    \label{eqSymKrsIntegral}
\end{equation}
Because this integral remains convergent as $\omega \to 0$, it approaches a well-defined zero-energy limit, justifying its designation as the infrared-finite block:
\begin{equation}
    \bm{J}^{\perp}(0) \equiv \lim_{\omega\to0}\bm{J}^{\perp}(\omega).
    \label{eqSymJperpLimit}
\end{equation}
Here and below, $\bm{J}^{\perp}(\omega)$ denotes the frequency-dependent block, while $\bm{J}^{\perp}(0)$ denotes its constant zero-energy limit. 

For the coupling blocks connecting the infrared-divergent ($i\leq D$) and infrared-finite ($r\in\mathcal{I}_{\perp}$) subspaces, the matrix elements satisfy
\begin{equation}
    R_{ir}(\omega),\ C_{ri}(\omega) =
    \begin{cases}
     o\!\left(q_\omega^{n_i-\bar{n}}\right) & n_i<\bar{n},\\[0.4em]
     O(1), & n_i=\bar{n}.
    \end{cases}
    \label{eqSymRCPiecewise}
\end{equation}
The constant $O(1)$ scaling in the second line occurs exclusively for the borderline index $i={{\iota}}$.

To evaluate the $S$-matrix via Eq.~\eqref{eqgvT0}, we decompose the parity-resolved coupling vector as
\begin{equation}
    V_{\lambda,i}(k) \equiv \langle u_i|v_\lambda(k)\rangle, \qquad \lambda\in\{\mathrm{e},\mathrm{o}\}.
    \label{eqSymVcomponentDef}
\end{equation}
Under this definition, the vector components satisfy the parity-coupling hierarchy:
\begin{equation}
    V_{\mathrm{e},i}(k) =
    \begin{cases}
        \sqrt{2}c_i k^{n_i}+o(k^{n_i}),&i\in\mathcal{I}_{\mathrm{e}},\\
        o(k^{n_i}),&i\in\mathcal{I}_{\mathrm{o}},
    \end{cases} \qquad
    V_{\mathrm{o},i}(k) =
    \begin{cases}
        o(k^{n_i}),&i\in\mathcal{I}_{\mathrm{e}},\\
        \sqrt{2}c_i k^{n_i}+o(k^{n_i}),&i\in\mathcal{I}_{\mathrm{o}}.
    \end{cases}
    \label{eqSymVparityHierarchy}
\end{equation}
We apply the concept of reference order to these parity couplings as well, utilizing $k^{n_i}$ as the reference scale. Only an infrared-divergent coupling $V_{p,i}(k)$ whose parity channel precisely matches the index parity ($i\in\mathcal{I}_p$) attains its reference order. All parity-mismatched couplings are strictly subleading. Furthermore, for every infrared-finite direction $r\in\mathcal{I}_{\perp}$, both parity components of $V_{\lambda,r}(k) \equiv \langle u_r|v_\lambda(k)\rangle$ are subleading relative to the threshold exponent $\bar{n}$
\begin{equation}
    V_{\mathrm{e},r}(k)=o(k^{ \bar{n}}), \qquad V_{\mathrm{o},r}(k)=o(k^{ \bar{n}}).
    \label{eqSymVperpHierarchy}
\end{equation}

With the asymptotic scaling of the $J$-matrix and the coupling components $V_{\lambda,i}(k)$ fully established, we now proceed to the proof of Theorem~2.  
\subsection{Single-parity systems with \texorpdfstring{$D=N$}{D=N}}
\label{SecSymSingleParityDN}
We begin with the $D=N$ regime, where all coupling exponents in $\mathcal{N}$ share the same parity $p\in\{\mathrm{e},\mathrm{o}\}$. Let $\bar{p}$ denote the opposite parity. In this case, the full $\bm{J}$-matrix reduces to a single same-parity block:
\begin{equation}
    \bm{J}(\omega) = \bm{J}^p(\omega).
    \label{eqSymSingleParityJ}
\end{equation}
Our objective is to prove that every single-parity system with $D=N$ is unconditionally noncritical, and that its zero-energy $S$-matrix limits satisfy
\begin{subequations}
\begin{align}
    S_{pp}(0^+) &= \exp\left[2\pi \mathrm{i} \frac{\eta(\mathcal{N}_p)}{m}\right], \label{eqSymSingleParityObjective1}\\
    S_{\bar{p}\bar{p}}(0^+) &= 1, \qquad 
    S_{p\bar{p}}(0^+) = S_{\bar{p}p}(0^+) = 0. \label{eqSymSingleParityObjective2}
\end{align}
\end{subequations}

The proof proceeds in three steps. First, we determine the asymptotic behavior of the Jost function, which fixes the determinant $\det[\bm{S}(0^+)]$ via Eq.~\eqref{SratioJ}. Second, we evaluate the asymptotics of the inverse $\bm{J}$-matrix to establish the trivial $S$-matrix elements in Eq.~\eqref{eqSymSingleParityObjective2} via Eq.~\eqref{eqgvT0}. Finally, combining the determinant with these three trivial entries uniquely determines the remaining nontrivial element, $S_{pp}(0^+)$, in Eq.~\eqref{eqSymSingleParityObjective1}.

\subsubsection{Jost function and the S-matrix determinant}

To execute the first step of this proof, we evaluate the Jost function, $\det[\bm{J}(\omega)]$. 
When the index set $\mathcal{N}$ possesses a uniform parity, this evaluation parallels the procedure established in Lemma~\ref{eqlemmadetZ} for antisymmetric dispersions.  
In the Leibniz expansion, multiple permutation-product terms share the same order of divergence, and their permutation-dependent coefficients sum to a Cauchy-type determinant. Evaluating this determinant is crucial to verify that the leading-order terms do not cancel.
Let $\mathcal{A}=\{n_1,\ldots,n_M\}$ consist of $M$ distinct integers of the same parity, with $n_i<\bar n$. Define the matrix $\breve{\bm Z}(\mathcal A)$ through its elements
\begin{equation}
\breve Z_{ij}(\mathcal A)
\equiv
\frac{\breve\kappa_{m,n_i+n_j}}{\breve\kappa_{m,2n_i}},
\qquad
\breve Z(\mathcal A)
\equiv
\det\big[\breve{\bm Z}(\mathcal A)\big],
\qquad
\breve Z(\varnothing)\equiv 1 .
\label{eqSymZdef}
\end{equation}

\begin{lemma}
\label{lemmaSymDetZ}
For the matrix in Eq.~\eqref{eqSymZdef},
\begin{equation}
\breve Z(\mathcal A)
=
\prod_{1\leq j<i\leq M}
\frac{\sin^2[\pi(n_i-n_j)/m]}
{\sin^2[\pi(n_i+n_j+1)/m]}
>0.
\label{eqSymZdet}
\end{equation}
\end{lemma}

\begin{proof}
Set $x_i=\exp[2\pi \mathrm{i}(n_i+1/2)/m]$. The definition of $\breve\kappa_{m,l}$ in Eq.~\eqref{eqkappasym} gives $\breve Z_{ij}(\mathcal A)=(1-x_i^2)/(1-x_i x_j)$. The Cauchy determinant identity then gives
\begin{equation}
\det[\breve{\bm Z}(\mathcal A)]
=
\prod_{1\leq j<i\leq M}
\frac{(x_i-x_j)^2}{(1-x_i x_j)^2}.
\label{eqSymZCauchy}
\end{equation}
Combining Eq.~\eqref{eqSymZCauchy} with $(x_i-x_j)/(1-x_i x_j)=-\sin[\pi(n_i-n_j)/m]/\sin[\pi(n_i+n_j+1)/m]$, Eq.~\eqref{eqSymZdet} follows. The numerator is nonzero because the integers in $\mathcal A$ are distinct. The denominator is nonzero because $n_i,n_j<\bar{n}$ implies $0<n_i+n_j+1<m$.
\end{proof}

We now apply Lemma~\ref{lemmaSymDetZ} to evaluate the Jost function for systems with single-parity coupling exponents.

\begin{lemma}[Single-parity Jost function]
\label{lemmaSymSingleParityJost}
Suppose that $D=N$ and that the exponent set $\mathcal{N}=\mathcal{N}_p$ consists of a single parity $p\in \{\rm e,o\}$. Let
\begin{equation}
    A_p \equiv \breve{Z}(\mathcal{N}_p^{<}) \prod_{i\in\mathcal{I}_p}|c_i|^2.
    \label{eqSymApDef}
\end{equation}
By Lemma~\ref{lemmaSymDetZ}, this coefficient is strictly positive ($A_p>0$), and the Jost function satisfies the asymptotic relation
\begin{equation}
    J(\omega) \sim A_p\breve{L}_{\mathcal{N}_p}(\omega)
    \label{eqSymSingleParityJost}
\end{equation}
 for both interior and borderline interactions.
\end{lemma}

\begin{proof}
For an interior interaction, the Leibniz expansion reads
\begin{equation}
J(\omega)
=
\sum_{\pi\in\mathfrak{S}(\mathcal I_p)}
\operatorname{sgn}(\pi)
\prod_{i\in\mathcal I_p}
J_{i,\pi(i)}(\omega),
\label{eqSymSingleLeibniz}
\end{equation}
where $\mathfrak{S}(\mathcal I_p)$ denotes the set of all permutations of $\mathcal I_p$.
 By Eq.~\eqref{eqSymJsameParity}, every product in Eq.~\eqref{eqSymSingleLeibniz} satisfies
\begin{equation}
\prod_{i\in\mathcal I_p}J_{i,\pi(i)}(\omega)
\sim
\left(\prod_{i\in\mathcal I_p}|c_i|^2\right)
\breve L_{\mathcal{N}_p}(\omega)
\prod_{i\in\mathcal I_p}\breve Z_{i,\pi(i)}(\mathcal N_p).
\label{eqSymSinglePermutationProduct}
\end{equation}
The total power of $q_\omega$ in Eq.~\eqref{eqSymSinglePermutationProduct} is independent of $\pi$. 

Summing the coefficient products and invoking Lemma~\ref{lemmaSymDetZ} establishes Eq.~\eqref{eqSymSingleParityJost} for an interior interaction, where $\mathcal{N}_p=\mathcal{N}_p^<$ by definition.

For a borderline interaction, let $\iota$ be the index with $n_{{\iota}}=\bar{n}$, and let $\bm J_{\cancel{{\iota} {\iota}}}(\omega)$ be the  submatrix obtained by deleting row and column $\iota$ of $\bm{J}(\omega)$. The terms in the Leibniz expansion with $\pi(\iota)={{\iota}}$ sum to $J_{{\iota} {\iota}}(\omega)\det[\bm J_{\cancel{{\iota} {\iota}}}(\omega)]$. Every remaining term has the same power-law order as $\breve L_{\mathcal N_p^{<}}(\omega)$ but contains no logarithmic factor that $J_{{\iota} {\iota}}(\omega)$ carries [see Eq.~\eqref{eqSymJborderLog}]. Hence
\begin{equation}
J(\omega)
=
J_{{\iota} {\iota}}(\omega)\det[\bm  J_{\cancel{{\iota} {\iota}}}(\omega)]
+
O\!\left(\breve L_{\mathcal N_p^{<}}(\omega)\right).
\label{eqSymBorderDetSplit}
\end{equation}
Because the submatrix $\bm{J}_{\cancel{{\iota} {\iota}}}(\omega)$ contains only interior indices, we evaluate its determinant using the interior-interaction case of Lemma~\ref{lemmaSymSingleParityJost} established above.
Thus Eq.~\eqref{eqSymBorderDetSplit} becomes
\begin{equation}
J(\omega)
\sim
|c_{{\iota}}|^2\breve L_{m,m-1}(\omega)
\left[
\breve Z(\mathcal N_p^{<})
\prod_{i\in\mathcal I_p\setminus\{\iota\}}|c_i|^2
\right]
\breve L_{\mathcal N_p^{<}}(\omega)
=
A_p\breve L_{\mathcal N_p}(\omega).
\label{eqSymBorderJostAsymptotic}
\end{equation}
Eq.~\eqref{eqSymBorderJostAsymptotic} proves Eq.~\eqref{eqSymSingleParityJost} for a borderline interaction.
\end{proof}

We now evaluate the zero-energy $S$-matrix determinant using Lemma~\ref{lemmaSymSingleParityJost} and Eq.~\eqref{SratioJ}. We first consider strictly interior interactions, where all coupling exponents satisfy $n < \bar{n}$. For every such exponent, the ratio of the infrared model functions evaluated on opposite sides of the real axis yields the phase factor
\begin{equation}
    \lim_{E\to0^+} \frac{\breve{L}_{m,2n}(E-\mathrm{i}0)}{\breve{L}_{m,2n}(E+\mathrm{i}0)} = \exp\!\left[\frac{2\pi \mathrm{i}}{m}(2n+1)\right].
    \label{eqSymThresholdPhase}
\end{equation}
Taking the product of these ratios over all $n \in \mathcal{N}_p$ directly yields 
\begin{equation}
    \det[\bm{S}(0^+)] =  \exp\left[2\pi \mathrm{i} \frac{\eta(\mathcal{N}_p)}{m}\right].
    \label{eqSymSingleDetS}
\end{equation}

We now extend this result to borderline interactions, which include the threshold exponent $\bar{n} \in \mathcal{N}_p$. For this exponent, the boundary values diverge logarithmically:
\begin{equation}
    \breve{L}_{m,m-1}(E\pm \mathrm{i}0) = \frac{2}{m}\left(\log E\mp \mathrm{i}\pi\right).
    \label{eqSymBorderBoundaryValues}
\end{equation}
Taking the limit of their ratio as $E \to 0^+$ gives
\begin{equation}
    \lim_{E\to0^+} \frac{\breve{L}_{m,m-1}(E-\mathrm{i}0)}{\breve{L}_{m,m-1}(E+\mathrm{i}0)} = 1 = \exp\!\left[\frac{2\pi \mathrm{i}}{m}(2\bar{n}+1)\right],
    \label{eqSymBorderThresholdPhase}
\end{equation}
where we have used the borderline identity $2\bar{n}+1=m$. Because the ratio in Eq.~\eqref{eqSymBorderThresholdPhase} evaluates to exactly $1$, including the borderline mode in the product over $\mathcal{N}_p$ does not alter the total phase. Therefore, the determinant identity in Eq.~\eqref{eqSymSingleDetS} remains valid for borderline interactions as well.

\subsubsection{Inverse J-matrix and the individual S-matrix elements}

The determinant of the S-matrix in Eq.~\eqref{eqSymSingleDetS} does not determine the four individual S-matrix elements. We therefore derive the asymptotic properties of $\bm J(\omega)^{-1}$ and evaluate the individual scattering amplitudes using Eq.~\eqref{eqgvT0}. Lemma~\ref{lemmaSymSingleParityJost} shows that $\bm J(\omega)$ is invertible for all sufficiently small nonzero $\omega$. Its inverse matrix elements are given by the adjugate formula
\begin{equation}
[\bm J(\omega)^{-1}]_{ij}
=
\frac{(-1)^{i+j}\det[\bm J_{\cancel{ji}}(\omega)]}
{J(\omega)},
\label{eqadjugateSingleparity}
\end{equation}
where $\bm J_{\cancel{ji}}(\omega)$ is obtained by deleting row $j$ and column $i$ from $\bm J(\omega)$, with its row and column indices   inherited from the remaining index sets $\{1,\ldots,N\}\setminus\{j\}$ and $\{1,\ldots,N\}\setminus\{i\}$, respectively. The denominator in Eq.~\eqref{eqadjugateSingleparity} is known from Lemma~\ref{lemmaSymSingleParityJost}, so we next determine the minor determinants.

\begin{lemma}
\label{lemmadetJijcancel}
Suppose that $D=N$ and $\mathcal N=\mathcal N_p$ has one parity $p\in\{\mathrm e,\mathrm o\}$. For every $1\le i,j\le N$, there is a constant $A_{ij}\neq0$ such that, as $\omega\to0$,
\begin{equation}
\det[\bm J_{\cancel{ij}}(\omega)]
\sim
A_{ij}\,q_\omega^{m-1-n_i-n_j}
\begin{cases}
\breve L_{\mathcal N_p}(\omega), & n_i,n_j<\bar{n},\\[2pt]
\breve L_{\mathcal N_p^{<}}(\omega), & n_i=\bar{n}\text{ or }n_j=\bar{n}.
\end{cases}
\label{eqSymSingleParityMinorAsymptotics}
\end{equation}
The first line of Eq.~\eqref{eqSymSingleParityMinorAsymptotics} applies to both interior and borderline interactions.
\end{lemma}

\begin{proof}
The proof of Lemma~\ref{lemmadetJijcancel} parallels that of Lemma~\ref{lemmaSymSingleParityJost}. In the Leibniz expansion of $\det[\bm J_{\cancel{ij}}(\omega)]$, the coefficients of the leading-order product terms again sum to a Cauchy-type determinant. Proving that this determinant is nonvanishing is essential to ensure that these leading-order terms do not cancel.

To this end, we construct the matrix $\breve{\bm Z}_{\cancel{ij}}(\mathcal N)$, which plays a role analogous to that of $\breve{\bm Z}(\mathcal N)$ in the evaluation of $\det[\bm{J}(\omega)]$. Here, $\mathcal N=\{n_1,\ldots,n_D\}$ consist of $D\ge1$ distinct nonnegative integers of one parity satisfying $n_h\le \bar{n}$.  For interior interactions, the full matrix $\breve{\bm Z}(\mathcal N)$ is well defined, and $\breve{\bm Z}_{\cancel{ij}}(\mathcal N)$ is simply the minor obtained by deleting its $i^{\text{th}}$ row and $j^{\text{th}}$ column. 
For borderline interactions, $\breve{\bm Z}(\mathcal N)$ is not defined because the denominator $\breve\kappa_{m,2n_\iota}=\breve\kappa_{m,m-1}$ appearing in its $\iota^{\mathrm{th}}$ row is undefined. Nevertheless, the minor $\breve{\bm Z}_{\cancel{\iota j}}(\mathcal N)$ remains well defined because this ill-defined row is removed.

In what follows, we will demonstrate that
\begin{equation}
\det[\breve{\bm Z}_{\cancel{ij}}(\mathcal N)]\neq0.
\label{eqSymCauchyMinorNonzero}
\end{equation}

With $x_h$ as in the proof of Lemma~\ref{lemmaSymDetZ}, each retained $(h,\ell)$ entry of $\breve{\bm Z}_{\cancel{ij}}$  equals $(1-x_h^2)/(1-x_hx_\ell)$. Dividing each retained row by $1-x_h^2$ and applying the Cauchy determinant identity gives
\begin{equation}
\det[\breve{\bm Z}_{\cancel{ij}}(\mathcal N)]
=
\frac{
\displaystyle
\prod_{h\neq i}(1-x_h^2)
\prod_{\substack{h<h'\\h,h'\neq i}}(x_{h'}-x_h)
\prod_{\substack{\ell<\ell'\\\ell,\ell'\neq j}}(x_{\ell'}-x_\ell)
}{
\displaystyle
\prod_{h\neq i}\prod_{\ell\neq j}(1-x_hx_\ell)
}.
\label{eqdetbreveZcancel}
\end{equation}
Distinct exponents satisfy $0<|n_h-n_{h'}|<m$, so the difference factors in Eq.~\eqref{eqdetbreveZcancel} are nonzero. For every retained pair, $n_h<\bar{n}$ and $n_\ell\le \bar{n}$ give $0<2n_h+1<m$ and $0<n_h+n_\ell+1<m$. Thus the row factors and denominators are also nonzero, proving Eq.~\eqref{eqSymCauchyMinorNonzero}, including $D=1$ by the empty-determinant convention.

We first consider interior interactions, for which $\mathcal N_p=\mathcal N_p^{<}$. The definition of the model functions  in Eq.~\eqref{eqKsum2} and of $\breve{Z}_{h\ell}(\mathcal N)$ in Eq.~\eqref{eqSymZdef} give $\breve L_{m,n_h+n_\ell}(\omega)=\breve L_{m,2n_h}(\omega)\breve Z_{h\ell}(\mathcal N)q_\omega^{n_\ell-n_h}$. The leading entrywise asymptotics in Eq.~\eqref{eqSymJsameParity} therefore take the form
\begin{equation}
J_{h\ell}(\omega)
=
\bigl[c_h\breve L_{m,2n_h}(\omega)q_\omega^{-n_h}\bigr]
\bigl[\breve Z_{h\ell}(\mathcal N)+o(1)\bigr]
\bigl[c_\ell^*q_\omega^{n_\ell}\bigr].
\label{eqSymSingleParityMinorEntry}
\end{equation}
The first and last factors in Eq.~\eqref{eqSymSingleParityMinorEntry} depend only on the row index $h$ and column index $\ell$, respectively. Dividing each retained row and column by its corresponding factor gives a matrix that converges entrywise to $\breve{\bm Z}_{\cancel{ij}}(\mathcal N)$, whose determinant is nonzero by Eq.~\eqref{eqSymCauchyMinorNonzero}. As in the interior proof of Eq.~\eqref{eqSymSingleParityJost}, we therefore obtain
\begin{subequations}
\label{eqSymSingleParityMinorFactorization}
\begin{align}
\det[\bm J_{\cancel{ij}}(\omega)]
\sim{}&
\left(\prod_{\substack{h=1\\h\neq i}}^{D}
c_h\breve L_{m,2n_h}(\omega)q_\omega^{-n_h}\right)
\left(\prod_{\substack{\ell=1\\\ell\neq j}}^{D}
c_\ell^*q_\omega^{n_\ell}\right)
\det[\breve{\bm Z}_{\cancel{ij}}(\mathcal N_p^{<})]
\label{eqSymSingleParityMinorFactorization1}
\\
={}& a_{ij}(\mathcal N_p^{<})\,q_\omega^{m-1-n_i-n_j}
\breve L_{\mathcal N_p^{<}}(\omega),
\label{eqSymSingleParityMinorFactorization0}
\end{align}
\end{subequations}
where Eq.~\eqref{eqKsum2} gives
\begin{equation}
a_{ij}(\mathcal N_p^{<})
=
\left(\prod_{n_h\in\mathcal N_p^{<}}|c_h|^2\right)
\det[\breve{\bm Z}_{\cancel{ij}}(\mathcal N_p^{<})]
\left[-\frac{2\pi \mathrm{i}}{m}\,c_i c_j^*\breve\kappa_{m,2n_i}\right]^{-1}
\neq0.
\label{eqSymSingleParityMinorCoefficient}
\end{equation}
Thus Eqs.~\eqref{eqSymSingleParityMinorFactorization} and \eqref{eqSymSingleParityMinorCoefficient} prove the first line of Eq.~\eqref{eqSymSingleParityMinorAsymptotics} for interior interactions, with $A_{ij}=a_{ij}(\mathcal N_p^{<})$.

We next consider borderline interactions, for which the prescribed ordering gives $\iota=N$ and $\mathcal N_p^{<}\neq\mathcal N_p$. First consider $i,j<{\iota}$. The minor $\bm J_{\cancel{ij}}(\omega)$ contains $J_{{\iota} {\iota}}(\omega)$, which diverges logarithmically by Eq.~\eqref{eqSymJborderLog}, while every other entry diverges algebraically by Eqs.~\eqref{eqSymJsameParity} and \eqref{eqKsum2}. As in the proof of Lemma~\ref{lemmaSymSingleParityJost}, all products in the Leibniz expansion contain the same algebraic power in $q_{\omega}$, while those containing $J_{{\iota} {\iota}}(\omega)$ dominate those without it by an additional logarithmic factor. Their signed sum is $J_{{\iota} {\iota}}(\omega)\det[\bm J^{(D-1)}_{\cancel{ij}}(\omega)]$, where $\bm J^{(D-1)}(\omega)\equiv\bm J_{\cancel{{\iota} {\iota}}}(\omega)$ and $\bm J^{(D-1)}_{\cancel{ij}}(\omega)$ is its submatrix obtained by deleting row $i$ and column $j$.

The block $\bm J^{(D-1)}(\omega)$ contains only entries with $n_h,n_\ell<\bar{n}$, so the determinant calculation in Eq.~\eqref{eqSymSingleParityMinorFactorization1} applies to $\det[\bm J^{(D-1)}_{\cancel{ij}}(\omega)]$ with both product bounds replaced by $D-1$. Consequently, $\det[\bm J^{(D-1)}_{\cancel{ij}}(\omega)]$  has the asymptotic form in Eq.~\eqref{eqSymSingleParityMinorFactorization0}, with $a_{ij}(\mathcal N_p^{<})$ given by Eq.~\eqref{eqSymSingleParityMinorCoefficient}. Since $a_{ij}(\mathcal N_p^{<})\neq0$, the signed sum $J_{{\iota} {\iota}}(\omega)\det[\bm J^{(D-1)}_{\cancel{ij}}(\omega)]$ has the same order as each product containing $J_{{\iota} {\iota}}(\omega)$. It therefore dominates the finite sum of the remaining products, and $\det[\bm J_{\cancel{ij}}(\omega)]$ is asymptotically equal to this signed sum.

Using $J_{{\iota} {\iota}}(\omega)\sim |c_{{\iota}}|^2\breve L_{m,m-1}(\omega)$ from Eq.~\eqref{eqSymJborderLog} and $\breve L_{\mathcal N_p}(\omega)=\breve L_{\mathcal N_p^{<}}(\omega)\breve L_{m,m-1}(\omega)$ from Eq.~\eqref{eqSymLset}, the interior-minor asymptotics in Eq.~\eqref{eqSymSingleParityMinorFactorization0} then give the first line of Eq.~\eqref{eqSymSingleParityMinorAsymptotics} for borderline interactions, with $A_{ij}=|c_{{\iota}}|^2a_{ij}(\mathcal N_p^{<})\neq0$.

For $i=j={{\iota}}$, we apply Lemma~\ref{lemmaSymSingleParityJost} to $\bm J^{(D-1)}(\omega)$. Since $m-1-2n_{{\iota}}=0$, this gives Eq.~\eqref{eqSymSingleParityMinorAsymptotics} with $A_{{\iota} {\iota}}=(\prod_{h<{\iota}}|c_h|^2)\det[\breve{\bm Z}(\mathcal N_p^{<})]\neq0$, including $N=1$ by the empty-determinant convention.

When $i={\iota}$ and $j<{\iota}$, the extended definition of $\breve{\bm Z}_{\cancel{{\iota}j}}(\mathcal N)$   allows the row and column factorization in Eq.~\eqref{eqSymSingleParityMinorFactorization1} to be applied similarly to $\det[\bm J_{\cancel{ij}}(\omega)]$. Hence Eq.~\eqref{eqSymSingleParityMinorFactorization0} applies similarly to this case with $\det[\breve{\bm Z}_{\cancel{{\iota}j}}(\mathcal N)]\neq0$ by Eq.~\eqref{eqSymCauchyMinorNonzero}. This proves the second line of Eq.~\eqref{eqSymSingleParityMinorAsymptotics} for $i={\iota}$.

For $i<{\iota}$ and $j={\iota}$, it is not hard to show that

\begin{equation}
\det[\bm J_{\cancel{i{\iota}}}(\omega)]
\sim
\frac{c_{{\iota}} c_i^*}{c_i c_{{\iota}}^*}
\det[\bm J_{\cancel{{\iota}i}}(\omega)].
\label{eqSymSingleParityMinorTranspose}
\end{equation}
Eq.~\eqref{eqSymSingleParityMinorTranspose} gives $A_{i{\iota}}=c_{{\iota}} c_i^*[c_i c_{{\iota}}^*]^{-1}A_{{\iota}i}\neq0$, proving the remaining case of Eq.~\eqref{eqSymSingleParityMinorAsymptotics}.

\end{proof}

Combining Lemmas~\ref{lemmaSymSingleParityJost} and \ref{lemmadetJijcancel} with the adjugate formula in Eq.~\eqref{eqadjugateSingleparity} gives the following inverse estimates.

\begin{lemma}[Single-parity inverse]
\label{lemmaSymSingleParityInv}
Suppose that $D=N$ and $\mathcal N=\mathcal N_p$ has one parity. Define, for $i,j\in\mathcal I_p$,
\begin{equation}
g_{ij}(\omega)
\equiv
\begin{cases}
q_\omega^{m-1-n_i-n_j},
& n_i,n_j<\bar{n},\\[0.4em]
q_\omega^{\bar{n}-n_i}|\log|\omega||^{-1},
& n_i<\bar{n},\ n_j=\bar{n},\\[0.4em]
q_\omega^{\bar{n}-n_j}|\log|\omega||^{-1},
& n_i=\bar{n},\ n_j<\bar{n},\\[0.4em]
|\log|\omega||^{-1},
& n_i=n_j=\bar{n}.
\end{cases}
\label{eqSymgijDef}
\end{equation}
The last three cases in Eq.~\eqref{eqSymgijDef} occur only for a borderline interaction. The inverse satisfies
\begin{equation}
[\bm J(\omega)^{-1}]_{ij}
=
O\!\left(g_{ij}(\omega)\right),
\qquad i,j\in\mathcal I_p.
\label{eqSymSingleInverseBound}
\end{equation}
Consequently,
\begin{equation}
\|\bm J(\omega)^{-1}\|\longrightarrow0.
\label{eqSymSingleInverseNorm}
\end{equation}
\end{lemma}

\begin{proof}
Substitute the determinant asymptotics of Lemma~\ref{lemmaSymSingleParityJost} and the minor asymptotics in Eq.~\eqref{eqSymSingleParityMinorAsymptotics}, with $i$ and $j$ interchanged, into Eq.~\eqref{eqadjugateSingleparity}. For $n_i,n_j<\bar{n}$, the factors $\breve L_{\mathcal N_p}(\omega)$ cancel. If either index is borderline, the numerator instead contains $\breve L_{\mathcal N_p^{<}}(\omega)$, leaving the factor $[\breve L_{m,m-1}(\omega)]^{-1}=O(|\log|\omega||^{-1})$. Since $m-1=2\bar{n}$, these cancellations give Eq.~\eqref{eqSymSingleInverseBound} with the scales in Eq.~\eqref{eqSymgijDef}. Every scale tends to zero, so Eq.~\eqref{eqSymSingleInverseBound} and finite dimensionality imply Eq.~\eqref{eqSymSingleInverseNorm}.
\end{proof}

We now combine Lemma~\ref{lemmaSymSingleParityInv} with the parity-coupling hierarchy in Eq.~\eqref{eqSymVparityHierarchy} to obtain the three S-matrix limits in Eq.~\eqref{eqSymSingleParityObjective2}. Since every leading coupling exponent $n_i$ has parity $p$,
\begin{equation}
V_{p,i}(k)=\sqrt{2}c_i k^{n_i}+o(k^{n_i}),
\qquad
V_{\bar p,i}(k)=o(k^{n_i}),
\qquad 1\le i\le N.
\label{eqSymSingleVectorHierarchy}
\end{equation}
For $(\lambda,\lambda')\neq(p,p)$, Eq.~\eqref{eqgvT0} gives
\begin{equation}
S_{\lambda\lambda'}(E)-\delta_{\lambda\lambda'}
=
2\pi \mathrm{i}\rho(E)
\sum_{i,j=1}^{N}
V_{\lambda',i}(k_E)^*
[\bm J(E+\mathrm{i}0)^{-1}]_{ij}
V_{\lambda,j}(k_E).
\label{eqSymSingleSsum}
\end{equation}
By Eq.~\eqref{eqSymSingleVectorHierarchy}, at least one of the two parity couplings in each summand of Eq.~\eqref{eqSymSingleSsum} is of strictly lower order than its reference scale. Thus $V_{\lambda',i}(k_E)^*V_{\lambda,j}(k_E)=o(k_E^{n_i+n_j})$. Since $\rho(E)=m^{-1}k_E^{1-m}$, Eq.~\eqref{eqSymgijDef} gives $\rho(E)k_E^{n_i+n_j}g_{ij}(E+\mathrm{i}0)=O(1)$ for both interior and borderline interactions. Together with Eq.~\eqref{eqSymSingleInverseBound}, this makes each summand in Eq.~\eqref{eqSymSingleSsum} $o(1)$ after multiplication by $\rho(E)$. The sum is finite, which proves Eq.~\eqref{eqSymSingleParityObjective2}.

The determinant identity
\begin{equation}
\det[\bm S(E)]
=
S_{pp}(E)S_{\bar p\bar p}(E)
-
S_{p\bar p}(E)S_{\bar p p}(E)
\label{eqSymTwoByTwoDet}
\end{equation}
then determines the remaining diagonal element. Since $S_{\bar p\bar p}(E)\to1$, Eqs.~\eqref{eqSymSingleDetS}, \eqref{eqSymSingleParityObjective2}, and \eqref{eqSymTwoByTwoDet} give Eq.~\eqref{eqSymSingleParityObjective1}. Finally, invertibility near zero and Eq.~\eqref{eqSymSingleInverseNorm} imply noncriticality by Lemma~\ref{definitionCritical0}. This completes the proof of Theorem~2 for single-parity systems with $D=N$.

\subsection{Mixed-parity systems with \texorpdfstring{$D=N$}{D=N}}
\label{SecSymMixedParityDN}

We now consider $D=N$ when both $\mathcal N_{\mathrm e}$ and $\mathcal N_{\mathrm o}$ are nonempty. In the parity-ordered basis,
\begin{equation}
\bm J(\omega)
=
\begin{pmatrix}
\bm J^{\mathrm e}(\omega) & \bm J^{\mathrm{eo}}(\omega)\\
\bm J^{\mathrm{oe}}(\omega) & \bm J^{\mathrm o}(\omega)
\end{pmatrix}.
\label{eqSymMixedDNJ}
\end{equation}
Our objective is to prove
\begin{equation}
\bm S(0^+)
=
\begin{pmatrix}
 \exp\left[2\pi \mathrm{i} \frac{\eta(\mathcal{N}_{\rm{e}})}{m}\right] & 0\\
0 &  \exp\left[2\pi \mathrm{i} \frac{\eta(\mathcal{N}_{\rm{o}})}{m}\right]
\end{pmatrix},
\label{eqSymMixedObjective}
\end{equation}
and to show that every mixed-parity system with $D=N$ is noncritical.

The block form in Eq.~\eqref{eqSymMixedDNJ} isolates the two same-parity blocks and their mixed-parity couplings. The fixed-parity results of Sec.~\ref{SecSymSingleParityDN} give
\begin{equation}
\det[\bm J^p(\omega)]
\sim
A_p\breve L_{\mathcal N_p}(\omega),
\qquad
[\bm J^p(\omega)^{-1}]_{ij}
=
O\!\left(g_{ij}(\omega)\right),
\qquad
\|\bm J^p(\omega)^{-1}\|\longrightarrow0,
\label{eqsym-mixed-single-parity-input}
\end{equation}
for $p\in\{\mathrm e,\mathrm o\}$ and $i,j\in\mathcal I_p$.

 The main idea of the proof is to map the S-matrix limits for the single-parity system to a mixed-parity system. To do so, we construct, for each parity $p$, an auxiliary single-parity system and compare its diagonal S-matrix element with $S_{pp}(E)$ of the full system. The off-diagonal elements of the S-matrix are shown to be zero separately.

\subsubsection{Inverse J-matrix and noncriticality}

To compute the zero-energy limit of the S-matrix elements and establish noncriticality, we first determine the low-energy asymptotics of $\bm J(\omega)^{-1}$. 
As for the J-matrix, we divide the inverse J-matrix into four blocks
\begin{equation}
\bm J(\omega)^{-1}
=
\begin{pmatrix}
[\bm J(\omega) ^{-1}]^{\mathrm e}&[\bm J(\omega) ^{-1}]^{\mathrm{eo}}\\
[\bm J(\omega) ^{-1}]^{\mathrm{oe}} & [\bm J(\omega) ^{-1}]^{\mathrm o}
\end{pmatrix}.
\label{eqSymMixedDNJinv}
\end{equation}

We show that each same-parity block $[\bm J(\omega) ^{-1}]^p$ agrees at leading order with $\bm J^p(\omega)^{-1}$, while every opposite-parity element, i.e. elements in $[\bm J(\omega) ^{-1}]^{\mathrm{eo}}$, $[\bm J(\omega) ^{-1}]^{\mathrm{oe}}$ is smaller than its inverse reference order $q^{m-1-n_i-n_j}_{\omega}$.

\begin{lemma}[Mixed-parity inverse]
\label{lemmaSymMixedInv}
For all sufficiently small nonzero $\omega$, $\bm J(\omega)$ is invertible. Within $[\bm J(\omega) ^{-1}]^{p}$, $p\in\{\mathrm e,\mathrm o\}$,
\begin{equation}
[\bm J(\omega)^{-1}]_{ij}
=
[\bm J^p(\omega)^{-1}]_{ij}
[1+o(1)]
\qquad
i,j\in\mathcal I_p.
\label{eqSymMixedInvSame}
\end{equation}
Within $[\bm J(\omega) ^{-1}]^{\mathrm{eo}}$, $[\bm J(\omega) ^{-1}]^{\mathrm{oe}}$, 
\begin{equation}
[\bm J(\omega)^{-1}]_{ij}
=
o\!\left(q_\omega^{m-1-n_i-n_j}\right),
\qquad
i\in\mathcal I_{\mathrm e},\ j\in\mathcal I_{\mathrm o}
\ \text{or vice versa}.
\label{eqSymMixedInvCross}
\end{equation}

Consequently,
\begin{equation}
\|\bm J(\omega)^{-1}\|\longrightarrow0.
\label{eqSymMixedInvNorm}
\end{equation}

\end{lemma}

\begin{proof}
For every $\omega$ such that $\det[\bm J(\omega)]\neq0$, the inverse matrix elements are given by the adjugate formula
\begin{equation}
[\bm J(\omega)^{-1}]_{ij}
=
\frac{(-1)^{i+j}\det[\bm J_{\cancel{ji}}(\omega)]}
{\det[\bm J(\omega)]},
\label{eqSymMixedAdjugate}
\end{equation}
where $\bm J_{\cancel{ji}}(\omega)$ is obtained from $\bm J(\omega)$ by deleting row $j$ and column $i$.

We first estimate the denominator in Eq.~\eqref{eqSymMixedAdjugate}.
Consider the Leibniz expansion of $\det[\bm{J}(\omega)]$. The signed sum of all Leibniz products containing only entries in single-parity blocks $\bm J^{\mathrm e}(\omega)$ and $\bm J^{\mathrm o}(\omega)$ is $\det[\bm J^{\mathrm e}(\omega)]\det[\bm J^{\mathrm o}(\omega)]$. By Eq.~\eqref{eqsym-mixed-single-parity-input}, this sum is $\Theta[q_\omega^{\eta(\mathcal{N})-mN}]$ for interior interactions and $\Theta(|\log|\omega||\,q_\omega^{\eta(\mathcal{N})-mN})$ for borderline interactions.
In comparison, by Eqs.~\eqref{eqSymJsameParity}, \eqref{eqKsum2}, \eqref{eqSymJborderLog}, and \eqref{eqSymJeoScaling}, all other Leibniz products containing a mixed-parity entries are $o(q_\omega^{\eta(\mathcal{N})-mN})$ for interior interactions and $o(|\log|\omega||\,q_\omega^{\eta(\mathcal{N})-mN})$ for borderline interactions.
Therefore, for both interior and borderline interactions, 
\begin{equation}
\det[\bm J(\omega)]
=
\det[\bm J^{\mathrm e}(\omega)]
\det[\bm J^{\mathrm o}(\omega)]
[1+o(1)].
\label{eqSymMixedDetFactorization}
\end{equation}
Hence $\bm J(\omega)$ is invertible for all sufficiently small nonzero $\omega$.

We next estimate the numerator in Eq.~\eqref{eqSymMixedAdjugate} for $i,j\in\mathcal I_p$, and let $\bar p$ denote the opposite parity. Each Leibniz product of $\det[\bm J_{\cancel{ji}}(\omega)]$ now contains one entry from every row except $j$ and every column except $i$. The signed sum of products containing only same-parity entries is $\det[\bm J^p_{\cancel{ji}}(\omega)]\det[\bm J^{\bar p}(\omega)]$, where $\bm J^p_{\cancel{ji}}(\omega)$ is obtained from $\bm J^p(\omega)$ by deleting row $j$ and column $i$, with indices inherited from $\bm J(\omega)$. Lemmas~\ref{lemmaSymSingleParityJost} and \ref{lemmadetJijcancel} give the asymptotic behavior of the determinants of $\bm J^{\bar{p}}(\omega)$ and $\bm J^p_{\cancel{ji}}(\omega)$. Every Leibniz product of $\det[\bm J_{\cancel{ji}}(\omega)]$ containing a mixed-parity entry is strictly lower order by the same arguments used for the denominator. Thus
\begin{equation}
\det[\bm J_{\cancel{ji}}(\omega)]
=
\det[\bm J^p_{\cancel{ji}}(\omega)]
\det[\bm J^{\bar p}(\omega)]
[1+o(1)].
\label{eqSymMixedCofactorPower1}
\end{equation}
Taking the ratio of Eqs.~\eqref{eqSymMixedCofactorPower1} and \eqref{eqSymMixedDetFactorization}, and applying the adjugate formula in Eq.~\eqref{eqSymMixedAdjugate} to both $\bm J(\omega)$ and $\bm J^p(\omega)$, proves Eq.~\eqref{eqSymMixedInvSame}.

If $i\in\mathcal I_{\mathrm e}$ and $j\in\mathcal I_{\mathrm o}$, deleting row $j$ and column $i$ leaves one more even row than even columns. Every Leibniz product of  $\det[\bm J_{\cancel{ji}}(\omega)]$ therefore contains an entry of $\bm J^{\mathrm{eo}}(\omega)$, which is strictly below its reference order by Eq.~\eqref{eqSymJeoScaling}.  Combining Eqs.~\eqref{eqKsum2}, \eqref{eqSymJsameParity}, \eqref{eqSymJborderLog}, and \eqref{eqSymJeoScaling} with the reference power $q_\omega^{\eta(\mathcal{N})-mN+m-1-n_i-n_j}$ established above, and summing the finite Leibniz expansion, gives Eq.~\eqref{eqdetJBorderMixedParity}.

\begin{equation}
\det[\bm J_{\cancel{ji}}(\omega)]
=
\begin{cases}
o\!\left(q_\omega^{\eta(\mathcal{N})-mN+m-1-n_i-n_j}\right),
& \bar{n}\notin\mathcal N,
\\[2pt]
o\!\left(|\log|\omega||\,q_\omega^{\eta(\mathcal{N})-mN+m-1-n_i-n_j}\right),
& \bar{n}\in\mathcal N,\ i,j\neq {\iota},
\\[2pt]
o\!\left(q_\omega^{\eta(\mathcal{N})-mN+\bar{n}-n_i}\right),
& \bar{n}\in\mathcal N,\ i\neq {\iota},\ j={\iota},
\\[2pt]
o\!\left(q_\omega^{\eta(\mathcal{N})-mN+\bar{n}-n_j}\right),
& \bar{n}\in\mathcal N,\ i={{\iota}},\ j\neq {\iota}.
\end{cases}
\label{eqdetJBorderMixedParity}
\end{equation}
The last two lines use $m-1=2\bar{n}$. Substituting Eqs.~\eqref{eqSymMixedDetFactorization} and \eqref{eqdetJBorderMixedParity} into Eq.~\eqref{eqSymMixedAdjugate}, with the determinant orders established above, proves Eq.~\eqref{eqSymMixedInvCross} for $[\bm J(\omega)^{-1}]^{\mathrm{eo}}$. Interchanging the parities gives the same result for $[\bm J(\omega)^{-1}]^{\mathrm{oe}}$.
\end{proof}

Lemma~\ref{lemmaSymMixedInv} shows that the same-parity blocks of the full inverse have the same leading asymptotic behavior as $\bm J^p(\omega)^{-1}$, while every opposite-parity element is below the inverse reference order $q_\omega^{m-1-n_i-n_j}$. These estimates are used in the S-matrix calculation below. Since $\bm J(\omega)$ is invertible near zero and Eq.~\eqref{eqSymMixedInvNorm} holds, Lemma~\ref{definitionCritical0} also shows that every mixed-parity system with $D=N$ is noncritical.

\subsubsection{Individual S-matrix elements}

With Lemma~\ref{lemmaSymMixedInv} and the single-parity results of Sec.~\ref{SecSymSingleParityDN} established, we now determine the four matrix elements of $\bm S(0^+)$. We first show that the parity-mixing entries $S_{\mathrm{eo}}(0^+)$ and $S_{\mathrm{oe}}(0^+)$ vanish. We then construct an auxiliary single-parity system for each $p\in\{\mathrm e,\mathrm o\}$ and compare its diagonal $p$-channel S-matrix element with the full system.

For $p\in\{\mathrm e,\mathrm o\}$, let $\bar p$ denote the opposite parity. Eq.~\eqref{eqgvT0} gives
\begin{equation}
S_{p\bar p}(E)
=
2\pi \mathrm{i}\rho(E)
\sum_{i,j=1}^{N}
V_{\bar p,i}(k_E)^*
[\bm J(E+\mathrm{i}0)^{-1}]_{ij}
V_{p,j}(k_E).
\label{eqsym-mixed-offdiag-S-expression}
\end{equation}
First consider a term in Eq.~\eqref{eqsym-mixed-offdiag-S-expression} with $n_i,n_j<\bar{n}$. By Lemma~\ref{lemmaSymMixedInv} and Eq.~\eqref{eqsym-mixed-single-parity-input}, the inverse element $[\bm J(E+\mathrm{i}0)^{-1}]_{ij}$ is $O(k_E^{m-1-n_i-n_j})$. Since $\rho(E)=m^{-1}k_E^{1-m}$, a nonvanishing contribution requires the parity couplings $V_{\bar p,i}(k_E)^*$ and $V_{p,j}(k_E)$ to attain their respective reference orders $k_E^{n_i}$ and $k_E^{n_j}$, and the inverse element to attain $k_E^{m-1-n_i-n_j}$. By Eq.~\eqref{eqSymVparityHierarchy}, the two couplings attain these orders only when $i\in\mathcal I_{\bar p}$ and $j\in\mathcal I_p$. The indices then have opposite parities, so Eq.~\eqref{eqSymMixedInvCross} makes the inverse element strictly lower order than $k_E^{m-1-n_i-n_j}$. Thus every such term is $o(1)$ after multiplication by $\rho(E)$.
For interior interactions, this estimate covers every summand in Eq.~\eqref{eqsym-mixed-offdiag-S-expression}. Since the summation is finite, 
\begin{equation}
S_{\mathrm{eo}}(0^+)=S_{\mathrm{oe}}(0^+)=0.
\label{eqsym-mixed-offdiag-S}
\end{equation}

For borderline interactions, we also need to  consider the terms with $i={{\iota}}$ or $j={{\iota}}$. If $n_i$ and $n_j$ have the same parity, Eqs.~\eqref{eqSymMixedInvSame}, \eqref{eqsym-mixed-single-parity-input}, and \eqref{eqSymgijDef} give $[\bm J(E+\mathrm{i}0)^{-1}]_{ij}=O(k_E^{m-1-n_i-n_j}|\log E|^{-1})$. Together with Eq.~\eqref{eqSymVparityHierarchy}, this makes the term $O(|\log E|^{-1})=o(1)$ after multiplication by $\rho(E)$. If $n_i$ and $n_j$ have opposite parities, Eqs.~\eqref{eqSymMixedInvCross} and \eqref{eqSymVparityHierarchy} give $o(1)$ directly. 
  Thus, for borderline interactions, every summand is also $o(1)$ after multiplication by $\rho(E)$.  Hence Eq.~\eqref{eqsym-mixed-offdiag-S} holds for both interior and borderline interactions.

We next determine the diagonal entries by comparison with an auxiliary single-parity system. Fix $p\in\{\mathrm e,\mathrm o\}$ and define the orthogonal projector $\bm\Pi_p=\sum_{i\in\mathcal I_p}|u_i\rangle\langle u_i|$. We construct the auxiliary system on the range of $\bm\Pi_p$ by replacing $|v(k)\rangle$ and $\bm W$ with $\bm\Pi_p|v(k)\rangle$ and $\bm\Pi_p\bm W\bm\Pi_p$, respectively.
Its exponent set is $\mathcal N_p$, its J-matrix is $\bm J^p(\omega)$, and its diagonal $p$-channel S-matrix element is
\begin{equation}
S_{pp}^{\mathrm{aux}}(E)
=
1+2\pi \mathrm{i}\rho(E)
\sum_{i,j\in\mathcal I_p}
V_{p,i}(k_E)^*
[(\bm J^p(E+\mathrm{i}0))^{-1}]_{ij}
V_{p,j}(k_E).
\label{eqsym-mixed-auxiliary-S-expression}
\end{equation}
The auxiliary system is a single-parity $D=N$ system, so Eq.~\eqref{eqSymSingleParityObjective1} gives $S_{pp}^{\mathrm{aux}}(0^+)= \exp\left[2\pi \mathrm{i} \frac{\eta(\mathcal{N}_p)}{m}\right]$. For separable potentials, the  matrix \(\bm\Pi_p\bm W\bm\Pi_p\) need not be invertible and therefore need not correspond to an interaction matrix \(-(\bm U^{\rm aux})^{-1}\) within the original model class. This does not affect the comparison, since the argument uses only the corresponding \(\bm J\)-matrix and the algebraic \(S\)-matrix relation.
 It remains to show that $S_{pp}(E)$ of the full system differs from $S_{pp}^{\mathrm{aux}}(E)$ by $o(1)$.

Subtracting Eq.~\eqref{eqsym-mixed-auxiliary-S-expression} from Eq.~\eqref{eqgvT0} with $\lambda=\lambda'=p$ gives
\begin{multline}
S_{pp}(E)-S_{pp}^{\mathrm{aux}}(E)
=
2\pi \mathrm{i}\rho(E)
\sum_{i,j\in\mathcal I_p}
V_{p,i}(k_E)^*
\left(
[\bm J(E+\mathrm{i}0)^{-1}]_{ij}
-[(\bm J^p(E+\mathrm{i}0))^{-1}]_{ij}
\right)
V_{p,j}(k_E)
\\
+2\pi \mathrm{i}\rho(E)
\sum_{\substack{1\leq i,j\leq N\\
(i,j)\notin\mathcal I_p\times\mathcal I_p}}
V_{p,i}(k_E)^*
[\bm J(E+\mathrm{i}0)^{-1}]_{ij}
V_{p,j}(k_E).
\label{eqScompareAux}
\end{multline}
In the first sum in Eq.~\eqref{eqScompareAux}, Eqs.~\eqref{eqSymMixedInvSame} and \eqref{eqsym-mixed-single-parity-input} make the inverse difference $o(g_{ij}(E+\mathrm{i}0))$. Eq.~\eqref{eqSymVparityHierarchy} bounds the two parity couplings by $O(k_E^{n_i})$ and $O(k_E^{n_j})$, while Eq.~\eqref{eqSymgijDef} gives $\rho(E)k_E^{n_i+n_j}g_{ij}(E+\mathrm{i}0)=O(1)$. Thus every summand is $o(1)$ after multiplication by $\rho(E)$.

In the second sum in Eq.~\eqref{eqScompareAux}, terms with indices of opposite parity are $o(1)$ after multiplication by $\rho(E)$ by Eqs.~\eqref{eqSymMixedInvCross} and \eqref{eqSymVparityHierarchy}. If both indices lie in $\mathcal I_{\bar p}$, the two parity couplings are $o(k_E^{n_i})$ and $o(k_E^{n_j})$ by Eq.~\eqref{eqSymVparityHierarchy}, while Eqs.~\eqref{eqSymMixedInvSame}, \eqref{eqsym-mixed-single-parity-input}, and \eqref{eqSymgijDef} bound the inverse element by $O(k_E^{m-1-n_i-n_j})$. These terms are therefore also $o(1)$ after multiplication by $\rho(E)$, including when one or both indices equal $\iota$. Both sums are finite, so $S_{pp}(E)-S_{pp}^{\mathrm{aux}}(E)\to0$. Consequently,
\begin{equation}
S_{\mathrm{ee}}(0^+)= \exp\left[2\pi \mathrm{i} \frac{\eta(\mathcal{N}_{\rm{e}})}{m}\right],
\qquad
S_{\mathrm{oo}}(0^+)= \exp\left[2\pi \mathrm{i} \frac{\eta(\mathcal{N}_{\rm{o}})}{m}\right].
\label{eqsym-mixed-diagonal-S}
\end{equation}
Eqs.~\eqref{eqsym-mixed-offdiag-S} and \eqref{eqsym-mixed-diagonal-S} prove Eq.~\eqref{eqSymMixedObjective}, completing the proof of Theorem~2 for all $D=N$ systems.

For later use, Eq.~\eqref{eqsym-mixed-single-parity-input} and the determinant factorization in Eq.~\eqref{eqSymMixedDetFactorization} give the Jost-function asymptotics for both interior and borderline interactions with $D=N$,
\begin{equation}
J(\omega)
\sim
A_{\mathrm e}A_{\mathrm o}
\breve L_{\mathcal N_{\mathrm e}}(\omega)
\breve L_{\mathcal N_{\mathrm o}}(\omega)
=
A_{\mathrm e}A_{\mathrm o}\breve L_{\mathcal N}(\omega).
\label{eqsym-mixed-Jost-final}
\end{equation}
Eq.~\eqref{eqsym-mixed-Jost-final} will be used in the $D<N$ analysis below.

\subsection{Systems with \texorpdfstring{$D<N$}{D<N}}
\label{SecSymDlessN}

We now complete the proof of Theorem~2 for $D<N$, allowing $\mathcal N$ to be empty or to contain one or both parities. The remaining $N-D$ basis directions span the infrared-finite subspace, and the corresponding block $\bm J^{\perp}(\omega)$ has the finite zero-energy limit $\bm J^{\perp}(0)$. Extending the $D=N$ results requires addressing two questions.
First, can the presence of the infrared-finite subspace make the system critical? Second, for a noncritical system, how can we show that coupling to this subspace preserves the zero-energy S-matrix obtained from the infrared-divergent subspace?
As for interior interactions with antisymmetric dispersion in Lemma~\ref{lemma_inter_crit}, we will show that criticality is determined by the invertibility of the limiting infrared-finite block $\bm J^{\perp}(0)\equiv \bm{J}^{\perp}(\omega=0) $.
To address the second question, we compare the full system with an auxiliary $D=N$ system restricted to the infrared-divergent subspace.

The proof therefore reduces to the following two lemmas.
\begin{lemma}[Criticality condition for $D<N$]
\label{lemmaCriticalSym}
For $D<N$, the system is critical if and only if
\begin{equation}
\det[\bm J^{\perp}(0)]=0.
\label{eqSymCriticalCondition}
\end{equation}
\end{lemma}

\begin{lemma}[S-matrix elements for $D<N$]
\label{lemmaSymDlessSMatrix}
Suppose that $D<N$ and $\det[\bm J^{\perp}(0)]\neq0$. Then
\begin{equation}
S_{\mathrm{ee}}(0^+)= \exp\left[2\pi \mathrm{i} \frac{\eta(\mathcal{N}_{\rm{e}})}{m}\right],
\qquad
S_{\mathrm{oo}}(0^+)= \exp\left[2\pi \mathrm{i} \frac{\eta(\mathcal{N}_{\rm{o}})}{m}\right],
\qquad
S_{\mathrm{eo}}(0^+)=S_{\mathrm{oe}}(0^+)=0.
\label{eqSymDlessSentries}
\end{equation}
\end{lemma}

Both proofs use properties of the block $\bm J^{(D)}(\omega)$ defined in Eq.~\eqref{eqSymJDdef}. For $0<D<N$, this block has the same asymptotic structure as the full J-matrices treated in Secs.~\ref{SecSymSingleParityDN} and \ref{SecSymMixedParityDN}. Hence
\begin{equation}
\det[\bm J^{(D)}(\omega)]
\sim
 A_{\mathrm e}A_{\mathrm o} \breve L_{\mathcal{N}}(\omega),
\label{eqSymJDdetReuse}
\end{equation}
where $A_p=1$ when $\mathcal I_p=\varnothing$, and
\begin{equation}
\|[\bm J^{(D)}(\omega)]^{-1}\|\longrightarrow0.
\label{eqSymJDinvReuse}
\end{equation}

\subsubsection{Criticality condition: proof of Lemma~\ref{lemmaCriticalSym}}

Criticality is characterized by Lemma~\ref{definitionCritical0} through singularity of $\bm J(\omega)$ or unboundedness of $\bm J(\omega)^{-1}$ {as $\omega\rightarrow 0$}. For $0<D<N$, Eq.~\eqref{eqSymJDdetReuse} shows that the infrared-divergent block $\bm J^{(D)}(\omega)$ is invertible for all sufficiently small nonzero $\omega$. Therefore, we can use the block determinant and block inverse techniques to obtain  the Jost function and $\bm J(\omega)^{-1}$.

The determinant $\det[\bm{J}(\omega)]$ can be written as the product of determinant of two smaller matrices
\begin{subequations}
\begin{equation}
J(\omega)
=
\det[\bm J^{(D)}(\omega)]
\det[\bm J_{\mathrm{red}}(\omega)],
\label{eqSymDlessSchurDet}
\end{equation}
where the reduced matrix $\bm J_{\mathrm{red}}(\omega)$ is given by 
\begin{equation}
\bm J_{\mathrm{red}}(\omega)
\equiv
\bm J^{\perp}(\omega)
-
\bm C(\omega)
[\bm J^{(D)}(\omega)]^{-1}
\bm R(\omega).
\label{eqSymDlessJred}
\end{equation}
\label{eqJostProduct}
\end{subequations}
When $\bm J_{\mathrm{red}}(\omega)$ is invertible,  we can further write the inverse of the J-matrix in a block form
\begin{equation}
\bm J(\omega)^{-1}
=
\begin{pmatrix}
[\bm J^{(D)}]^{-1}
+[\bm J^{(D)}]^{-1}\bm R\bm J_{\mathrm{red}}^{-1}
\bm C[\bm J^{(D)}]^{-1}
&
-[\bm J^{(D)}]^{-1}\bm R\bm J_{\mathrm{red}}^{-1}
\\[0.5em]
-\bm J_{\mathrm{red}}^{-1}\bm C[\bm J^{(D)}]^{-1}
&
\bm J_{\mathrm{red}}^{-1}
\end{pmatrix}.
\label{eqSymDlessBlockInverse}
\end{equation}
The argument $\omega$ is omitted from the blocks on the right-hand side of Eq.~\eqref{eqSymDlessBlockInverse}. Every block is evaluated at the same $\omega$.

To employ the inverse identity in Eq.~\eqref{eqSymDlessBlockInverse}, we first examine the zero-energy limit of $\bm J_{\mathrm{red}}(\omega)$ in Eq.~\eqref{eqSymDlessJred}. Its second term is written as 
\eq{
\left[
\bm C(\omega)[\bm J^{(D)}(\omega)]^{-1}\bm R(\omega)
\right]_{rs}
=
\sum_{i,j=1}^{D}
C_{ri}(\omega)
[\bm J^{(D)}(\omega)^{-1}]_{ij}
R_{js}(\omega).
\label{eqSymDlessCorrectionEntry}
}
The coupling estimates in Eq.~\eqref{eqSymRCPiecewise}, together with the inverse bounds in Eqs.~\eqref{eqSymgijDef}, \eqref{eqSymSingleInverseBound}, \eqref{eqSymMixedInvSame}, and \eqref{eqSymMixedInvCross}, give, for every summand in Eq.~\eqref{eqSymDlessCorrectionEntry},
\begin{equation}
C_{ri}(\omega)
\bigl[[\bm J^{(D)}(\omega)]^{-1}\bigr]_{ij}
R_{js}(\omega)
=
\begin{cases}
o(1),
& n_i<\bar{n}\text{ or }n_j<\bar{n},\\[0.4em]
O\!\left(\lvert\log|\omega|\rvert^{-1}\right),
& i=j={\iota}.
\end{cases}
\label{eqSymDlessCorrectionInterior}
\end{equation}

 Since the sum in Eq.~\eqref{eqSymDlessCorrectionEntry} is finite,
\begin{equation}
\bm C(\omega)
[\bm J^{(D)}(\omega)]^{-1}
\bm R(\omega)
\longrightarrow0.
\label{eqSymDlessCorrectionZero}
\end{equation}
Combining Eqs.~\eqref{eqSymJperpLimit}, \eqref{eqSymDlessJred}, and \eqref{eqSymDlessCorrectionZero} gives
\begin{equation}
\bm J_{\mathrm{red}}(\omega)
=
\bm J^{\perp}(0)+o(1).
\label{eqSymDlessJredLimit}
\end{equation}
By continuity of the determinant, Eq.~\eqref{eqSymDlessJredLimit} implies
\begin{equation}
\det[\bm J_{\mathrm{red}}(\omega)]
=
\det[\bm J^{\perp}(0)]+o(1).
\label{eqSymDlessJredDetLimit}
\end{equation}

Combining Eqs.~\eqref{eqSymJDdetReuse}, \eqref{eqSymDlessSchurDet}, and \eqref{eqSymDlessJredDetLimit} also gives the Jost-function asymptotics
\begin{equation}
J(\omega)
=
A_{\rm{e}}A_{\rm{o}}[\det[\bm J^{\perp}(0)]+o(1)]
\breve L_{\mathcal{N}}(\omega).
\label{eqSymDlessJostGeneral}
\end{equation}

We continue with the inverse of the $J$-matrix. 
We estimate the matrix elements of $[\bm J^{(D)}(\omega)]^{-1}\bm R(\omega)$ and $\bm C(\omega)[\bm J^{(D)}(\omega)]^{-1}$, which enter the block inverse in Eq.~\eqref{eqSymDlessBlockInverse} and have dimensions $D\times(N-D)$ and $(N-D)\times D$, respectively. For $1\le i\le D$ and $r\in\mathcal I_{\perp}$, their $(i,r)$ and $(r,i)$ entries are the sums of $\bigl[\bm J^{(D)}(\omega)^{-1}\bigr]_{ih}R_{hr}(\omega)$ and $C_{rh}(\omega)\bigl[\bm J^{(D)}(\omega)^{-1}\bigr]_{hi}$, respectively, over $1\le h\le D$. Using the coupling bounds in Eq.~\eqref{eqSymRCPiecewise}, the inverse estimates of Lemmas~\ref{lemmaSymSingleParityInv} and \ref{lemmaSymMixedInv}, and the borderline cofactor estimates in Eq.~\eqref{eqdetJBorderMixedParity} from the proof of the latter, we obtain

\begin{equation}
\left[\bm J^{(D)}(\omega)^{-1}\bm R(\omega)\right]_{ir},\quad
\left[\bm C(\omega)\bm J^{(D)}(\omega)^{-1}\right]_{ri}
=
\begin{cases}
o\!\left(q_{\omega}^{\bar{n}-n_i}\right), & n_i<\bar{n},\\[0.4em]
O\!\left(|\log|\omega||^{-1}\right), & n_i=\bar{n}.
\end{cases}
\label{eqSymDlessMixedProducts}
\end{equation}
Therefore, every matrix element in Eq.~\eqref{eqSymDlessMixedProducts} tends to zero.

First, suppose that $\det[\bm J^{\perp}(0)]\neq0$. By Eq.~\eqref{eqSymDlessJredLimit} and the continuity of matrix inversion, we have
\begin{equation}
\bm J_{\mathrm{red}}(\omega)^{-1}
\longrightarrow
[\bm J^{\perp}(0)]^{-1}.
\label{eqSymDlessJredInvLimit}
\end{equation}
Substituting Eqs.~\eqref{eqSymJDinvReuse}, \eqref{eqSymDlessJredInvLimit}, and \eqref{eqSymDlessMixedProducts} into the block inversion formula of Eq.~\eqref{eqSymDlessBlockInverse}, we obtain an asymptotic relation completely analogous to the antisymmetric case with interior interactions [cf.\ Eq.~\eqref{eqInteriorFullInverseLimitAnti}]:
\begin{equation}
\bm J(\omega)^{-1}
\longrightarrow
\begin{pmatrix}
\bm 0 & \bm 0\\
\bm 0 & [\bm J^{\perp}(0)]^{-1}
\end{pmatrix}.
\label{eqSymDlessFullInvLimit}
\end{equation}
Equation~\eqref{eqSymDlessFullInvLimit} demonstrates that $\|\bm J(\omega)^{-1}\|$ remains bounded as $\omega\to 0$. By Lemma~\ref{definitionCritical0}, the system is therefore noncritical.

Conversely, suppose that $\det[\bm{J}^\perp(0)]=0$. By applying the exact same logic used in Sec.~\ref{SecDlessNinterior} (immediately following Eq.~\eqref{eqInteriorFullInverseLimitAnti}), we can deduce that the system must be critical. Thus, both directions of Lemma~\ref{lemmaCriticalSym} are established.

\subsubsection{S-matrix elements: proof of Lemma~\ref{lemmaSymDlessSMatrix}}
We now prove Lemma~\ref{lemmaSymDlessSMatrix}, which establishes the zero-energy $S$-matrix limits for $D<N$. 
We begin with the trivial case of $D=0$. Here, $\bm{J}(\omega)=\bm{J}^{\perp}(\omega)$, and the condition $\det[\bm{J}^{\perp}(0)]\neq 0$ assumed in the lemma ensures that $\bm{J}(\omega)^{-1}\to[\bm{J}^{\perp}(0)]^{-1}$ as $\omega \to 0$. Furthermore, Eq.~\eqref{eqSymVperpHierarchy} dictates that $V_{\lambda,r}(k_E)=o(k_E^{\bar{n}})$ for all $r\in\mathcal{I}_{\perp}$. Substituting these limits into Eq.~\eqref{eqgvT0} yields
\begin{equation}
    \bm{S}(0^+)=\mathbb{1}_2.
    \label{eqSymDzeroS}
\end{equation}
Because $\mathcal{N}_{\mathrm{e}}=\mathcal{N}_{\mathrm{o}}=\varnothing$ and hence $\exp\left[2\pi \mathrm{i} \frac{\eta(\mathcal{N}_{\rm{e}})}{m}\right]=\exp\left[2\pi \mathrm{i} \frac{\eta(\mathcal{N}_{\rm{o}})}{m}\right]=1$ when $D=0$, Eq.~\eqref{eqSymDzeroS} trivially satisfies Eq.~\eqref{eqSymDlessSentries}, concluding the proof for this case.

We now assume $0<D<N$ and $\det[\bm{J}^{\perp}(0)]\neq 0$. Analogous to our extension from the single-parity to the mixed-parity case for $D=N$, our strategy for the $D<N$ regime is to compare the full system against an auxiliary system to which the previously established $D=N$ results apply. To this end, we first summarize the asymptotic properties of $\bm{J}(\omega)^{-1}$ obtained in the proof of Lemma~\ref{lemmaCriticalSym}.

\begin{lemma}[Inverse $\bm{J}$-matrix for $D<N$]
\label{lemmaSymDlessInv}
For indices $i, j \leq D$, the matrix elements in the principal block of the inverse satisfy
\begin{equation}
    [\bm{J}(\omega)^{-1}]_{ij} - [\bm{J}^{(D)}(\omega)^{-1}]_{ij} = o\!\left(q_\omega^{m-1-n_i-n_j}\right).
    \label{eqSymDlessInvIRcomparison}
\end{equation}
For $i\leq D$ and $r\in\mathcal{I}_{\perp}$, the elements in the off-diagonal blocks are bounded by
\begin{equation}
    [\bm{J}(\omega)^{-1}]_{ir}, \ [\bm{J}(\omega)^{-1}]_{ri} = o\!\left(q_\omega^{\bar{n}-n_i}\right),
    \label{eqSymDlessInvIRperp}
\end{equation}
while, for $r,s\in\mathcal{I}_{\perp}$, the lower diagonal block approaches a finite limit:
\begin{equation}
    [\bm{J}(\omega)^{-1}]_{rs} \longrightarrow [\bm{J}^{\perp}(0)^{-1}]_{rs}.
    \label{eqSymDlessInvPerp}
\end{equation}
\end{lemma}

 We now construct an auxiliary $D=N$ system, directly analogous to the single-parity construction in Sec.~\ref{SecSymMixedParityDN}. We define the projector onto the infrared-divergent subspace as
\begin{equation}
    \bm{\Pi}_D \equiv \sum_{i=1}^{D}|u_i\rangle\langle u_i|.
    \label{eqSymIRProjector}
\end{equation}
Applying this projector restricts the coupling (or form-factor) vector and the interaction matrix to the infrared-divergent subspace, defining the auxiliary system:
\begin{equation}
    |v^{\mathrm{aux}}(k)\rangle \equiv \bm{\Pi}_D|v(k)\rangle, \qquad
    [\bm{W}]^{\mathrm{aux}} \equiv \bm{\Pi}_D\bm{W}\bm{\Pi}_D.
    \label{eqSymDlessAuxiliarySystem}
\end{equation}
The parity-resolved coupling (or form-factor) vectors are projected similarly, such that $|v_\lambda^{\mathrm{aux}}(k)\rangle \equiv \bm{\Pi}_D|v_\lambda(k)\rangle$ for $\lambda\in\{\mathrm{e},\mathrm{o}\}$. Consequently, the $\bm{J}$-matrix of this auxiliary system identically matches the leading principal block of the full system,
\begin{equation}
    \bm{J}^{\mathrm{aux}}(\omega) = \bm{J}^{(D)}(\omega).
    \label{eqSymDlessAuxiliaryJ}
\end{equation}
Because its dimension equals the number of infrared-divergent directions $D$, this auxiliary system is governed by the $D=N$ framework previously established in Secs.~\ref{SecSymSingleParityDN} and \ref{SecSymMixedParityDN}.

Its S-matrix elements are
\begin{equation}
S_{\lambda\lambda'}^{\mathrm{aux}}(E)
=
\delta_{\lambda\lambda'}
+2\pi \mathrm{i}\rho(E)
\langle v_{\lambda'}^{\mathrm{aux}}(k_E)|
[\bm J^{(D)}(E+\mathrm{i}0)]^{-1}
|v_\lambda^{\mathrm{aux}}(k_E)\rangle,
\qquad
\lambda,\lambda'\in\{\mathrm e,\mathrm o\},
\label{eqSymDlessAuxiliaryS}
\end{equation}

and the $D=N$ results give
\begin{equation}
S_{\mathrm{ee}}^{\mathrm{aux}}(0^+)= \exp\left[2\pi \mathrm{i} \frac{\eta(\mathcal{N}_{\rm{e}})}{m}\right],
\qquad
S_{\mathrm{oo}}^{\mathrm{aux}}(0^+)= \exp\left[2\pi \mathrm{i} \frac{\eta(\mathcal{N}_{\rm{o}})}{m}\right],
\qquad
S_{\mathrm{eo}}^{\mathrm{aux}}(0^+)
=S_{\mathrm{oe}}^{\mathrm{aux}}(0^+)=0.
\label{eqSymDlessAuxiliarySLimit}
\end{equation}

To apply the limit of the auxiliary $S$-matrix in Eq.~\eqref{eqSymDlessAuxiliarySLimit} to the full system, we must show that the $S$-matrix elements of both systems converge to the identical zero-energy limit. We decompose the parity-resolved coupling vector into its infrared-divergent and infrared-finite components, $|v_\lambda(k)\rangle=(|v_\lambda^{\mathrm{aux}}(k)\rangle, |v_\lambda^{\perp}(k)\rangle)^T$.
Applying the block inversion identity from Eq.~\eqref{eqSymDlessBlockInverse} yields the difference in the unweighted transition amplitudes between the full and auxiliary systems

\begin{align}
    &\langle v_{\lambda'}(k)|\bm{J}(\omega)^{-1}|v_\lambda(k)\rangle - \langle v_{\lambda'}^{\mathrm{aux}}(k)|[\bm{J}^{(D)}(\omega)]^{-1}|v_\lambda^{\mathrm{aux}}(k)\rangle
    \nonumber\\
    &\quad=
    \left(
    \langle v_{\lambda'}^{\perp}(k)| - \langle v_{\lambda'}^{\mathrm{aux}}(k)|[\bm{J}^{(D)}(\omega)]^{-1}\bm{R}(\omega)
    \right)
    \bm{J}_{\mathrm{red}}(\omega)^{-1}
    \left(
    |v_\lambda^{\perp}(k)\rangle - \bm{C}(\omega)[\bm{J}^{(D)}(\omega)]^{-1}|v_\lambda^{\mathrm{aux}}(k)\rangle
    \right).
    \label{eqSymDlessAmplitudeFactorization}
\end{align}

We evaluate this expression on-shell by setting $\omega=E+\mathrm{i}0$ and $k=k_E$. The central matrix, $\bm{J}_{\mathrm{red}}(E+\mathrm{i}0)^{-1}$, remains bounded as $E\to 0$ according to Eq.~\eqref{eqSymDlessJredInvLimit}. We now demonstrate that both outer parenthetical factors scale as $o(k_E^{\bar{n}})$. The first term in each parenthesis, representing the infrared-finite component $|v_\lambda^{\perp}(k_E)\rangle$, is bounded by $o(k_E^{\bar{n}})$ via Eq.~\eqref{eqSymVperpHierarchy}. The second term involves matrix elements of $[\bm{J}^{(D)}(\omega)]^{-1}\bm{R}(\omega)$ (or $\bm{C}(\omega)[\bm{J}^{(D)}(\omega)]^{-1}$), which we bound component by component: for an interior index $i$, multiplying the coupling $V_{\lambda,i}(k_E)=O(k_E^{n_i})$ by the first line of Eq.~\eqref{eqSymDlessMixedProducts} directly yields $o(k_E^{\bar{n}})$. For the borderline index $\iota$, the second line of Eq.~\eqref{eqSymDlessMixedProducts} provides a bound of $O(k_E^{\bar{n}}|\log E|^{-1})$, which is also $o(k_E^{\bar{n}})$. 

Because each outer factor scales as $o(k_E^{\bar{n}})$, the right-hand side of Eq.~\eqref{eqSymDlessAmplitudeFactorization} scales as $o(k_E^{2\bar{n}})$. The density of states scales inversely to this, $\rho(E)=m^{-1}k_E^{-2\bar{n}}$. Thus, multiplying the transition amplitude difference by $\rho(E)$ yields a vanishing limit:
\begin{equation}
    \rho(E)
    \left[
    \langle v_{\lambda'}(k_E)|\bm{J}(E+\mathrm{i}0)^{-1}|v_\lambda(k_E)\rangle - \langle v_{\lambda'}^{\mathrm{aux}}(k_E)|[\bm{J}^{(D)}(E+\mathrm{i}0)]^{-1}|v_\lambda^{\mathrm{aux}}(k_E)\rangle
    \right]
    \longrightarrow 0.
    \label{eqSymDlessSComparison}
\end{equation}
Substituting this asymptotic behavior into the $T$-matrix definition in Eq.~\eqref{eqgvT0} confirms that the full and auxiliary $S$-matrices share the same zero-energy limit
\begin{equation}
    S_{\lambda\lambda'}(0^+) = S_{\lambda\lambda'}^{\mathrm{aux}}(0^+),
    \qquad
    \lambda,\lambda'\in\{\mathrm{e},\mathrm{o}\}.
    \label{eqSymDlessSReduction}
\end{equation}
Combining Eq.~\eqref{eqSymDlessSReduction} with the established auxiliary limit in Eq.~\eqref{eqSymDlessAuxiliarySLimit} establishes Eq.~\eqref{eqSymDlessSentries}, thereby completing the proof of Lemma~\ref{lemmaSymDlessSMatrix} for the $0<D<N$ regime.

To summarize, we have established Theorem~2 for the $D=N$ regime in both the single-parity (Sec.~\ref{SecSymSingleParityDN}) and mixed-parity (Sec.~\ref{SecSymMixedParityDN}) cases. Furthermore, Lemmas~\ref{lemmaCriticalSym} and \ref{lemmaSymDlessSMatrix} extend this result to all noncritical systems with $D<N$, including $D=0$. Because these scenarios collectively exhaust all possible exponent sets $\mathcal{N}$, the proof of Theorem~2 is now complete.

\section{Proof of Levinson's Theorem}

In this section, we provide the proof of the generalized Levinson's theorem given in Eq.~(12) of the main text.

The theorem relates the total change of the scattering phase across the continuum spectrum to the number of bound states $N_B$ and a universal zero-energy contribution. Specifically, we will prove
\begin{equation}
\Delta\delta
=\pi\Delta N+\pi\frac{mD-\eta(\mathcal N)}{m},
\label{eqLevS1}
\end{equation}
where $\Delta N$ is $N-N_B$ for emitter scattering and $-N_B$ for separable potential scattering.
Here, $N_B$ counts bound states with energies
outside the continuum spectrum, including
their algebraic multiplicities, and $\eta(\mathcal N)$ is defined in Eq.~\eqref{eqnAdefi}. We assume that the system is noncritical, has no bound states in the continuum, and has a finite, nonzero $\det[\bm S(E)]$ at every nonzero continuum energy. We further
assume that $J(E\pm\mathrm{i}0)\neq0$ at all
nonzero energies $E$ in the continuum spectrum.

We retain the analyticity of $v_a(k)$ at $k=0$ and ultraviolet convergence assumed in the main text. We additionally assume that the coupling coefficients or form factors $v_a(k)$ are twice continuously differentiable for real $k$.  For contact interactions,
any required ultraviolet regularization is chosen
to leave the form factors unchanged near $k=0$,
preserve their smoothness, and ensure ultraviolet
convergence.

In Fig.~\ref{figureLevinson}, we illustrate Levinson's theorem by plotting the trajectory of $\det[\boldsymbol S(E)]$ in the complex plane for representative dissipative systems with (a) antisymmetric and (b) symmetric dispersion.

Our goal is to evaluate the winding phase of the
S-matrix determinant, which is defined by the
integral
\begin{equation}
2\mathrm{i}\Delta\delta
=
\int_{E_{\rm min}}^{E_{\rm max}}dE\,
\frac{d}{dE}\ln\det[\bm S(E)].
\label{eqDeltaPhi}
\end{equation}

Here, $E_{\rm max}=+\infty$, while
$E_{\rm min}=-\infty$ for antisymmetric dispersion
and $E_{\rm min}=0$ for symmetric dispersion.
For antisymmetric dispersion, the integral denotes
the sum of the improper integrals over $E<0$
and $E>0$, with the logarithm followed
continuously on each interval. The phase jump
at $E=0$ is excluded. For symmetric dispersion,
the integral starts at $E=0^+$.
The regularity assumptions on $v_a(k)$, together
with the nonvanishing conditions stated above,
ensure that the logarithmic derivative is
well defined at every nonzero continuum energy.

\begin{figure}
    \centering    \includegraphics[width=0.6\linewidth]{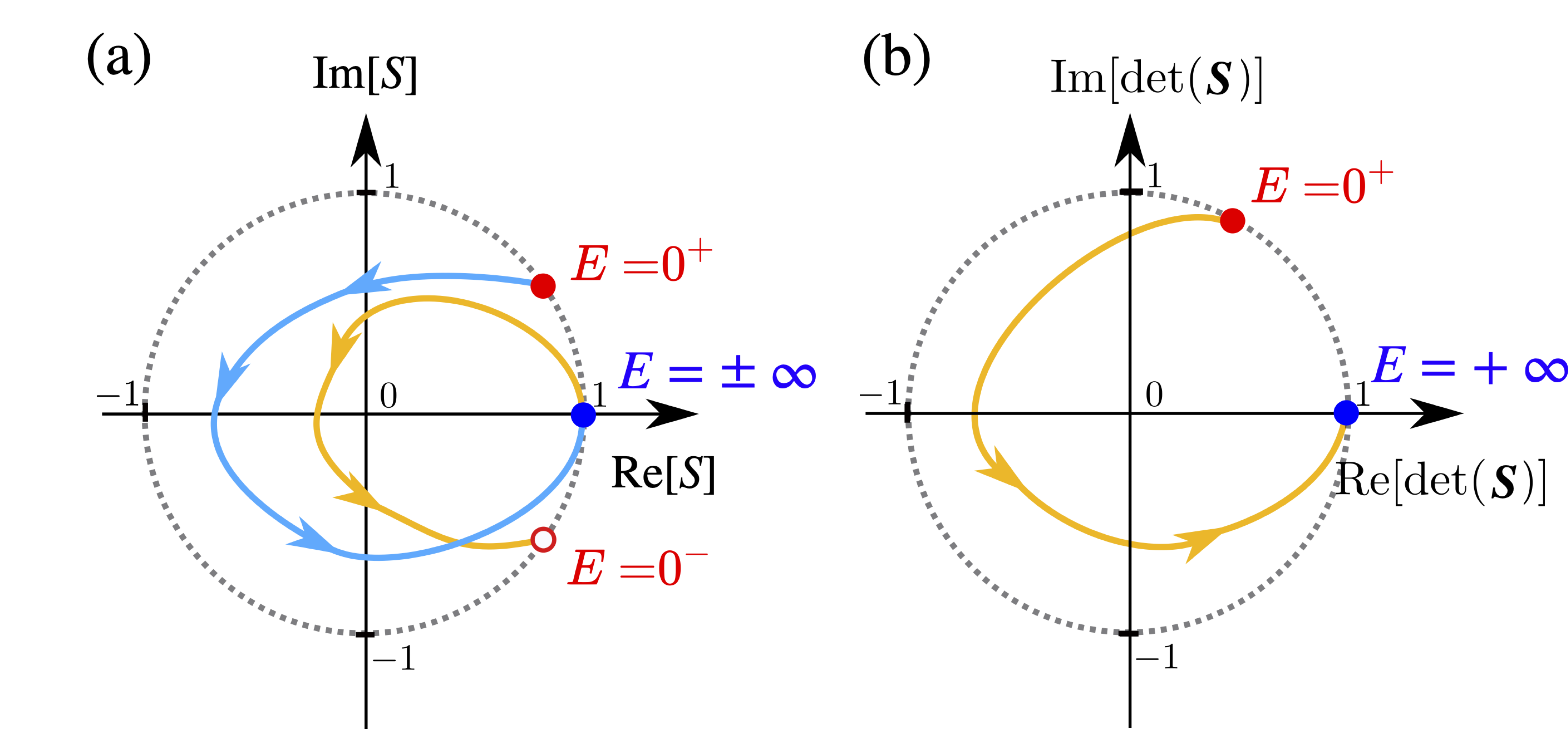}
    \caption{
Illustration of the trajectories of $\det[\boldsymbol S(E)]$ in the complex plane for a dissipative system as $E$ increases along the continuum spectrum for (a) antisymmetric dispersion and (b) symmetric dispersion. In (a), the trajectory starts and ends at $1$ with a discontinuity across $E=0$. The segments of the trajectory for $E<0$ and $E>0$ are colored yellow and blue, respectively. In (b), the trajectory begins at $\det[\bm S(0^+)]$ and ends at $\det[\bm S(+\infty)]=1$.
    }
    \label{figureLevinson}
\end{figure}

Our strategy is to evaluate this integral using complex analysis. By relating the S-matrix determinant to the Jost function via Eq.~\eqref{SratioJ}, we obtain
\begin{equation}
\frac{d}{dE}\ln\det[\bm S(E)]
=
\frac{\partial_E J(E-\mathrm{i}0)}{J(E-\mathrm{i}0)}
-
\frac{\partial_E J(E+\mathrm{i}0)}{J(E+\mathrm{i}0)}.
\label{eqLevCutDiscontinuity}
\end{equation}
Eq.~\eqref{eqLevCutDiscontinuity} transforms the real-axis integral into a contour integral of $\partial_\omega J(\omega)/J(\omega)$ in the complex $\omega$-plane. This integral can then be evaluated by closing the contour and applying the argument principle. The contributions to the integral arise from the residues at the poles of the integrand, which occur at bound-state zeros of $J(\omega)$, and from the integrals along the closing arcs at low and high energies. Since $J(\omega)$ is analytic away from the continuum spectrum [Eqs.~\eqref{eqMKrelations} and \eqref{eqKdefi}], these residues equal the multiplicities of the zeros of $J(\omega)$. The shape of the integration contour depends on the dispersion relation, as shown in Fig.~\ref{fig:levinson}.

\begin{figure}[t]
    \centering
\includegraphics[width=0.6\linewidth]{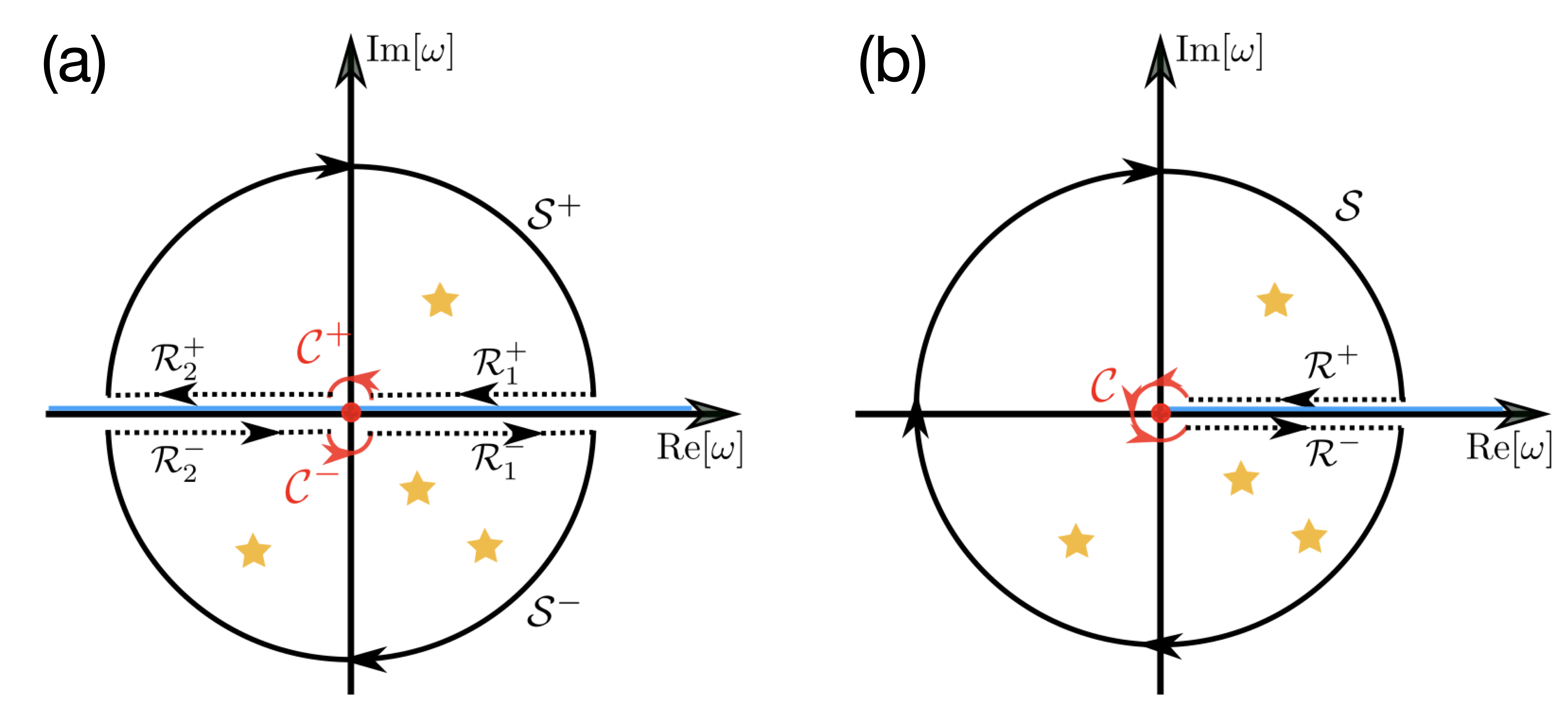}
    \caption{Integration contours for calculating the winding phase of $\det[\bm S(E)]$. The calculation is performed by expressing the integral along the continuum spectrum (blue line) in terms of paths along the two sides of the branch cut of the Jost function, $J(\omega)$ (dashed lines, $\mathcal R$). To apply the residue theorem, these open paths are then closed by adding large arcs at infinity ($\mathcal S$) and small arcs around the origin ($\mathcal C$). The total winding phase is thus determined by the sum of residues at the enclosed bound states (yellow stars) and the integral contributions from the closing arcs. 
For non-Hermitian Hamiltonians, bound-state
energies may lie in either half-plane, as
illustrated by the yellow stars.
The specific geometry of the resulting closed contour depends on the dispersion:
\textbf{(a)}~For antisymmetric dispersion, the method requires two separate closed contours, one in the upper and one in the lower half-plane, using paths $\mathcal R_1^{\pm}$, $\mathcal R_2^{\pm}$, $\mathcal S^{\pm}$, and $\mathcal C^{\pm}$.
\textbf{(b)}~For symmetric dispersion, a single closed contour suffices. This contour follows both sides of the branch cut along the positive real axis, using paths $\mathcal R^{\pm}$, $\mathcal S$, and $\mathcal C$.
    }
    \label{fig:levinson}
\end{figure}

\subsection{Antisymmetric Dispersion}

For antisymmetric dispersion, the continuum covers the entire real axis. The integral in Eq.~\eqref{eqDeltaPhi} can be expressed as an integral along contours just above and below the real axis, as shown in Fig.~\ref{fig:levinson}(a).
The paths $\mathcal R_1^\pm$ follow the positive real axis and $\mathcal R_2^\pm$ the negative real axis, with $+$ and $-$ denoting the upper and lower sides, respectively. As shown in the figure, the upper paths are oriented toward decreasing energy and the lower paths toward increasing energy. Eq.~\eqref{eqLevCutDiscontinuity} therefore gives
\begin{equation}
2\mathrm{i}\Delta\delta
=
\left(\int_{\mathcal R_1^-}+\int_{\mathcal R_2^-}\right)
 d\omega\,\frac{\partial_\omega J(\omega)}{J(\omega)}
+
\left(\int_{\mathcal R_1^+}+\int_{\mathcal R_2^+}\right)
 d\omega\,\frac{\partial_\omega J(\omega)}{J(\omega)}.
\label{eqLevRealAxisAnti}
\end{equation}
We close these paths with large semicircles at infinity ($\mathcal S^\pm$) and small semicircles around the origin ($\mathcal C^\pm$).
With the orientations in Eq.~\eqref{eqLevRealAxisAnti}, both closed contours and the large semicircles are clockwise, while the small semicircles are counterclockwise. All contour formulas below are understood after evaluation at finite radii, followed by the limits in which the small arcs shrink to zero and the large arcs recede to infinity.
The integral can then be evaluated using the argument principle:
\begin{equation}
2\mathrm{i}\Delta\delta
=
\left(\oint_{\mathrm{upper}}+\oint_{\mathrm{lower}}\right)
 d\omega\,\frac{\partial_\omega J(\omega)}{J(\omega)}
-
\left(\int_{\mathcal S^+\cup\mathcal S^-}
      +\int_{\mathcal C^+\cup\mathcal C^-}\right)
 d\omega\,\frac{\partial_\omega J(\omega)}{J(\omega)}.
\label{eqWinding}
\end{equation}

The closed contour integrals count the zeros of $J(\omega)$.
For non-Hermitian interactions, the bound-state energies may lie in the complex plane, as discussed in Sec.~\ref{SecCritical}.  Since both contours are clockwise, their combined contribution is
\begin{equation}
\left(\oint_{\mathrm{upper}}+\oint_{\mathrm{lower}}\right)
 d\omega\,\frac{\partial_\omega J(\omega)}{J(\omega)}
=-2\pi\mathrm{i}N_B.
\label{eqLevBoundStateCountAnti}
\end{equation}

The integral around the small semicircles $\mathcal C^\pm$ measures the phase winding from the low-energy behavior of $J(\omega)$. 
By Lemma~\ref{lemma1} and Eqs.~\eqref{eqLml}
and \eqref{eqLsetAnti}, for both interior and
borderline interactions, the leading term of
$J(\omega)$ is proportional in each half-plane
to $q_\omega^{\eta(\mathcal N)-mD}$.
For borderline interactions, the leading
coefficients $A^\pm$ may differ between the
two half-planes.
In particular, for $\mathcal N=\{\bar n\}$,
$J(\omega)\to A^\pm\neq0$ in the respective
half-planes. For trivial interactions,
$J(\omega)=\det[\bm{J}^\perp(\omega)]\to\det[\bm J^\perp(0)]\neq0$ based on Eq.~\eqref{eqInteriorJPerpLimitAnti} and
Lemma~\ref{definitionCritical0}.

For interior, borderline, and trivial interactions,
the leading prefactors of $J(\omega)$ are nonzero
constants within each half-plane and therefore
do not contribute to the small-arc integrals. 
Moreover, the ratio of $J(\omega)$ to its
leading term approaches $1$ uniformly on each
shrinking semicircle, so its contribution to
the contour integral of
$\partial_\omega J(\omega)/J(\omega)$
vanishes as the semicircle shrinks to zero.
 The two counterclockwise semicircles $\mathcal{C}^\pm$ together increase $\arg\omega$ by $2\pi$, giving
\begin{equation}
\int_{\mathcal C^+\cup\mathcal C^-}
 d\omega\,\frac{\partial_\omega J(\omega)}{J(\omega)}
=2\pi\mathrm{i}\left[\frac{\eta(\mathcal N)}{m}-D\right].
\label{eqLevSmallArc}
\end{equation}

The contribution from the large semicircles
$\mathcal S^\pm$ is determined by the
high-frequency asymptotics of $J(\omega)$.
We assume that $\bm K(\omega)\to\bm 0$
uniformly along these arcs as their radii
tend to infinity.

For emitter scattering, $J(\omega)\sim(-\omega)^N$, by Eq.~\eqref{eqMKrelations}. For separable potentials, $J(\omega)$ approaches the nonzero constant $\det[-\bm U^{-1}]$ [Eqs.~\eqref{eqMKrelations} and \eqref{eqWrelations0}]. Evaluating the change of the continuous logarithm along the two clockwise semicircles gives
\begin{equation}
\int_{\mathcal S^+\cup\mathcal S^-}
 d\omega\,\frac{\partial_\omega J(\omega)}{J(\omega)}
=
\begin{cases}
-2\pi\mathrm{i}N, & \text{emitter scattering},\\
0, & \text{separable potential scattering}.
\end{cases}
\label{eqLevLargeArc}
\end{equation}
Combining Eqs.~\eqref{eqLevBoundStateCountAnti}, \eqref{eqLevSmallArc}, and \eqref{eqLevLargeArc} in Eq.~\eqref{eqWinding} yields Eq.~\eqref{eqLevS1} for the antisymmetric case.

\subsection{Symmetric Dispersion}

For symmetric dispersion, the continuum lies only on the positive real axis. The winding phase is therefore evaluated along a single contour that wraps around this branch cut, as shown in Fig.~\ref{fig:levinson}(b).
The upper side is traversed toward decreasing energy and the lower side toward increasing energy. The closed contour and the large circle $\mathcal S$ are clockwise, while the small circle $\mathcal C$ is counterclockwise. As in Eq.~\eqref{eqWinding},
\begin{equation}
2\mathrm{i}\Delta\delta
=\oint d\omega\,\frac{\partial_\omega J(\omega)}{J(\omega)}
-\int_{\mathcal S}d\omega\,\frac{\partial_\omega J(\omega)}{J(\omega)}
-\int_{\mathcal C}d\omega\,\frac{\partial_\omega J(\omega)}{J(\omega)}.
\label{eqLevWindingSym}
\end{equation}
The closed contour encloses all bound states outside the continuum, giving a contribution of $-2\pi\mathrm{i}N_B$.

The integral around the small circle $\mathcal C$
is determined by the Jost-function asymptotics
established in Lemma~\ref{lemmaSymSingleParityJost}
and Eqs.~\eqref{eqsym-mixed-Jost-final} and
\eqref{eqSymDlessJostGeneral}.
For a noncritical system, these results give
$J(\omega)\sim A\breve L_{\mathcal N}(\omega)$,
where $A\neq0$ is a constant and
$\breve L_{\mathcal N}(\omega)$ is defined in
Eq.~\eqref{eqSymLset}.
As in the antisymmetric case, the ratio of
$J(\omega)$ to its leading term approaches $1$
uniformly on the shrinking circle, so the
subleading corrections contribute nothing to
the contour integral in the limit.
Thus, the small-circle contribution is
determined by $\breve L_{\mathcal N}(\omega)$.

For interior interactions, Eq.~\eqref{eqKsum2}
gives
$\breve L_{\mathcal N}(\omega)
\propto q_\omega^{\eta(\mathcal N)-mD}$.
As in the antisymmetric case, this power-law
dependence gives the small-circle contribution
in Eq.~\eqref{eqLevSmallArc}.
This also includes trivial interactions,
for which $\mathcal N=\varnothing$ and
$\breve L_{\varnothing}(\omega)=1$.

For borderline interactions,
$\breve L_{\mathcal N}(\omega)
=\breve L_{\mathcal N^<}(\omega)
\breve L_{m,m-1}(\omega)$.
Since
$\eta(\mathcal N^<)-m(D-1)
=\eta(\mathcal N)-mD$,
the power-law factor
$\breve L_{\mathcal N^<}(\omega)$
gives the same contribution as for
interior interactions.
The logarithmic factor
$\breve L_{m,m-1}(\omega)$,
defined in Eq.~\eqref{eqSymmLbrevedefiborder},
contributes no net phase winding in the
shrinking-circle limit. Indeed,
\begin{equation}
\frac{\partial_\omega\breve L_{m,m-1}(\omega)}
{\breve L_{m,m-1}(\omega)}
=
\frac{1}{
\omega\{\log|\omega|
+\mathrm{i}[\arg(\omega)-\pi]\}
}.
\label{eqLevBorderlineLogDerivative}
\end{equation}
Since $|\omega|$ is constant on $\mathcal C$,
the magnitude of this factor's contour integral
is bounded by $2\pi/|\log|\omega||$, which
vanishes as $\mathcal C$ shrinks to zero.

Thus Eq.~\eqref{eqLevSmallArc} also holds
for symmetric dispersion, with $\mathcal C$
replacing $\mathcal C^+\cup\mathcal C^-$.

The integral over the large circle $\mathcal S$ gives $-2\pi\mathrm{i}N$ for emitter scattering and $0$ for separable potentials, identical to the antisymmetric case  in Eq.~\eqref{eqLevLargeArc}. Again, combining these three contributions in Eq.~\eqref{eqLevWindingSym} recovers Eq.~\eqref{eqLevS1}, completing the proof for symmetric dispersion.

 \section{General structure of the low-energy expansion}
\label{secGeneralLowEnergyExpansion}

In this section and Sec.~\ref{secContactLowEnergyExpansion}, we prove the results presented in the End Matter on the low-energy expansion of the scattering phase for the analytic monomial dispersion $\epsilon(k)=k^m$, where $m>1$ is an integer. Even $m$ corresponds to symmetric dispersion, whereas odd $m$ corresponds to antisymmetric dispersion. Throughout these two sections, we write $k\equiv |E|^{1/m}$. Our objective in the present section is to show that the low-energy phase correction has an asymptotic expansion in positive integer powers of $k$ and that, for even $m$, the powers $k^j$ with $j\in m\mathbb Z$ are absent. In Sec.~\ref{secContactLowEnergyExpansion}, we specialize to contact interactions and determine the leading correction.

We define the determinant phase locally by $\det[\bm S(E)]=e^{2\mathrm{i}\delta(E)}$, with the logarithm branches chosen continuously near the relevant side of the threshold. For nonunitary scattering, $\delta(E)$ may be complex. The relation between the S-matrix and the Jost function for antisymmetric and symmetric dispersions, Eq.~\eqref{SratioJ}, gives
\begin{equation}
2\mathrm{i}\delta(E)
=
\log J(E-\mathrm{i}0)-\log J(E+\mathrm{i}0).
\label{eqPhaseFromJostBoundaryValues}
\end{equation}
Thus the low-energy expansion of $\delta(E)$ is determined by the asymptotic expansion of the Jost function $J(\omega)$ near $\omega=0$. Sections~\ref{secProofOdd} and \ref{SecProofSym} established only the leading term of $J(\omega)$, which is sufficient to determine the zero-energy S-matrix. Here we retain the subleading terms that determine the structure of the phase expansion. We use $q_\omega=|\omega|^{1/m}e^{\mathrm{i}\arg(\omega)/m}$ with $\arg(\omega)\in(0,2\pi)$, and $\mu=e^{\mathrm{i}\pi/m}$, consistently with the preceding sections.

\subsection{Low-energy expansion of the Jost function}
\label{subsecInfraredExpansionClass}

Each matrix element of the projected-resolvent matrix $\bm K(\omega)$ is a linear combination of integrals of the form
\begin{equation}
\int_{-\infty}^{+\infty}dk\,
\frac{f(k)}{\omega-k^m}.
\label{eqfpropintegral}
\end{equation}
To determine the subleading terms in Eq.~\eqref{eqfpropintegral}, we extend the infrared model functions $L_{m,l}(\omega)$ and $\breve L_{m,l}(\omega)$, introduced in Eqs.~\eqref{eqSLmndefinition0} and \eqref{eqSLmndefinition1}, to all nonnegative integers $l$. We retain the coefficients $\kappa_{m,l}(\pm)$ and $\breve\kappa_{m,l}$ defined in Eqs.~\eqref{eqkappa} and \eqref{eqkappasym}, respectively.

\begin{definition}[Infrared model functions for analytic monomial dispersions] \label{defimodelfunctions}
For odd $m$, $l\in\mathbb Z_{\geq0}$, and $\omega\in\mathbb C\setminus\mathbb R$, define
\begin{equation}
L_{m,l}(\omega)
\equiv
-\frac{2\pi \mathrm{i}}{m}\,
\kappa_{m,l}\!\left(\operatorname{sgn}[\operatorname{Im}\omega]\right)
q_\omega^{l+1-m}.
\label{eqExtendedAntisymmetricModelFunction}
\end{equation}
For even $m$, $l\in\mathbb Z_{\geq0}$, and $\omega\in\mathbb C\setminus[0,\infty)$, define
\begin{equation}
\breve L_{m,l}(\omega)
\equiv
\begin{cases}
0,
& l\ \text{odd},\\[0.4em]
-\dfrac{2\pi \mathrm{i}}{m}\,
\breve\kappa_{m,l}\,
q_\omega^{l+1-m},
& l\ \text{even}.
\end{cases}
\label{eqExtendedSymmetricModelFunction}
\end{equation}
\end{definition}

For $0\leq l<m-1$, Eqs.~\eqref{eqExtendedAntisymmetricModelFunction} and \eqref{eqExtendedSymmetricModelFunction} agree with the convergent integral definitions in Eqs.~\eqref{eqSLmndefinition0} and \eqref{eqSLmndefinition1}. For $l\geq m-1$, these integrals require ultraviolet regularization, while Eqs.~\eqref{eqExtendedAntisymmetricModelFunction} and \eqref{eqExtendedSymmetricModelFunction} define their regularization-independent nonanalytic parts. Specifically, let $g(k)$ be an ultraviolet regularization function that equals one near $k=0$ and decays sufficiently rapidly as $|k|\to\infty$. Then
\begin{equation}
\int_{-\infty}^{+\infty}dk\,
\frac{k^l g(k)}{\omega-k^m}
=
F_g(\omega)
+
\begin{cases}
L_{m,l}(\omega), & m\ \text{odd},\\
\breve L_{m,l}(\omega), & m\ \text{even},
\end{cases}
\label{eqRegularizedInfraredDecomposition}
\end{equation}
where $F_g(\omega)$ is analytic at $\omega=0$ and contains the dependence on the regularization function $g(k)$.

The following lemma gives the asymptotic expansion of the integral in Eq.~\eqref{eqfpropintegral} in terms of the infrared model functions.
\begin{lemma}[Low-energy expansion of the projected-resolvent integrals]
\label{lemmaLowEnergyIntegralExpansion}
Let $m>1$ be an integer, and let $f$ be analytic near $k=0$, with Taylor expansion $f(k)=\sum_{j=0}^{\infty}f_jk^j$. Suppose that
$\int_{|k|>1}dk\,|f(k)|/|k|^m$ is finite. For every integer $M\geq0$, there exists a function $F_M(\omega)$ analytic at $\omega=0$ with the following properties as $\omega\to0$.

For odd $m$ and $\omega\in\mathbb C\setminus\mathbb R$,
\begin{equation}
\int_{-\infty}^{+\infty}dk\,
\frac{f(k)}{\omega-k^m}
=
\sum_{j=0}^{M}f_jL_{m,j}(\omega)
+F_M(\omega)
+o\!\left(L_{m,M}(\omega)\right).
\label{eqLowEnergyIntegralExpansionOdd}
\end{equation}
For even $m$ and $\omega\in\mathbb C\setminus[0,\infty)$,
\begin{equation}
\int_{-\infty}^{+\infty}dk\,
\frac{f(k)}{\omega-k^m}
=
\sum_{j=0}^{M}f_{2j}\breve L_{m,2j}(\omega)
+F_M(\omega)
+o\!\left(\breve L_{m,2M}(\omega)\right).
\label{eqLowEnergyIntegralExpansionEven}
\end{equation}
\end{lemma}

\begin{proof}
Choose a smooth even cutoff that equals one near $k=0$ and has compact support. Subtract the Taylor polynomial of $f(k)$ through order $M$ for odd $m$, or the even Taylor polynomial through order $2M$ for even $m$, inside this cutoff. Eq.~\eqref{eqRegularizedInfraredDecomposition} gives the model-function contribution of each subtracted monomial. The contribution supported away from $k=0$ is analytic at $\omega=0$, because its denominator can be expanded uniformly in powers of $\omega/k^m$.

The remaining local Taylor series begins at order $M+1$ for odd $m$. Its first possible model term is therefore $L_{m,M+1}(\omega)$, which is $o(L_{m,M}(\omega))$ because their ratio is $O(q_\omega)$. For even $m$, the odd part of the numerator integrates to zero, and the remaining even Taylor series begins at order $2M+2$. Its first possible model term is $\breve L_{m,2M+2}(\omega)=o(\breve L_{m,2M}(\omega))$. All subsequent model terms are of still lower order. These observations give Eqs.~\eqref{eqLowEnergyIntegralExpansionOdd} and \eqref{eqLowEnergyIntegralExpansionEven}.
\end{proof}

In particular, the integral diverges when the lowest Taylor order with a nonzero model function is less than $m-1$. For odd $m$, the order $l=m-1$ gives a finite half-plane-dependent term. For even $m$, $l=m-1$ is odd, and the corresponding symmetric model function vanishes.

Lemma~\ref{lemmaLowEnergyIntegralExpansion} implies that every matrix element of $\bm J(\omega)$ has an asymptotic Laurent expansion in $q_\omega$. Since the determinant is a finite sum of finite products of matrix elements, the Jost function has the same structure.

\begin{lemma}[Jost-function expansion for analytic dispersions]
\label{lemmaGeneralJostExpansionClass}
The products $L_{\mathcal N}(\omega)$ and $\breve L_{\mathcal N}(\omega)$ are defined in Eqs.~\eqref{eqLsetAnti} and \eqref{eqSymLset}, respectively. For a noncritical system with $\epsilon(k)=k^m$ and even $m$, there exist coefficients $\{a_j\}_{j\geq0}$, with $a_0\neq0$, such that, for every integer $M\geq0$,
\begin{equation}
J(\omega)
=
\breve L_{\mathcal N}(\omega)
\left[
\sum_{j=0}^{M}a_jq_\omega^j
+o\!\left(q_\omega^M\right)
\right],
\qquad
\omega\to0,
\quad
\omega\in\mathbb C\setminus[0,\infty).
\label{eqGeneralJostExpansionClass}
\end{equation}
For a noncritical system with $\epsilon(k)=k^m$ and odd $m$, there exist coefficients $\{a_j^{\pm}\}_{j\geq0}$, with $a_0^{\pm}\neq0$, such that, for every integer $M\geq0$,
\begin{equation}
J(\omega)
=
L_{\mathcal N}(\omega)
\left[
\sum_{j=0}^{M}a_j^{\pm}q_\omega^j
+o\!\left(q_\omega^M\right)
\right],
\qquad
\omega\to0,
\quad
\pm\operatorname{Im}\omega>0.
\label{eqAnalyticAntisymmetricJostExpansion}
\end{equation}
The coefficients $a_j^+$ and $a_j^-$ need not be equal.
\end{lemma}

\begin{proof}
Lemma~\ref{lemmaLowEnergyIntegralExpansion} gives an asymptotic Laurent expansion in $q_\omega$ for every matrix element of $\bm K(\omega)$. The remaining terms in $\bm J(\omega)$ are analytic in $\omega=q_\omega^m$, so every matrix element of $\bm J(\omega)$ has the same expansion structure. Taking the finite-dimensional determinant preserves this structure. The leading Jost-function asymptotics established in Secs.~\ref{secProofOdd} and \ref{SecProofSym} give the factors $L_{\mathcal N}(\omega)$ and $\breve L_{\mathcal N}(\omega)$, respectively. In the antisymmetric borderline case, $L_{m,m-1}(\omega)$ is a nonzero half-plane-dependent constant and can be absorbed into $a_0^{\pm}$. Noncriticality guarantees that the corresponding leading coefficients are nonzero. For odd $m$, the coefficients $\kappa_{m,l}(\pm)$ depend on the half-plane, which allows the two coefficient sequences in Eq.~\eqref{eqAnalyticAntisymmetricJostExpansion} to differ.
\end{proof}

Because the leading coefficients in Eqs.~\eqref{eqGeneralJostExpansionClass} and \eqref{eqAnalyticAntisymmetricJostExpansion} are nonzero, taking the logarithm preserves the expansion in integer powers of $q_\omega$. Equation~\eqref{eqPhaseFromJostBoundaryValues} then gives an expansion of $\delta(E)-\delta(0^+)$ in positive integer powers of $k$ for even $m$, and separate expansions of $\delta(E)-\delta(0^\pm)$ in positive integer powers of $k$ for odd $m$.

\subsection{Missing powers in the scattering-phase expansion for $\epsilon(k)=k^m$ with even $m$}

For even $m$, analytic continuation around the branch point imposes a further selection rule on the powers of $k$.

\begin{lemma}[Missing powers for symmetric analytic dispersion]
\label{lemmaMissingMultiplesEvenM}
Consider $\epsilon(k)=k^m$ with even $m$. Away from a critical point, there exist coefficients $\{d_j\}_{j\geq1}$ such that, for every integer $M\geq1$,
\begin{equation}
\delta(E)-\delta(0^+)
=
\sum_{\substack{1\leq j\leq M\\ j\notin m\mathbb Z}}
d_jk^j
+o\!\left(k^M\right),
\qquad
E\to0^+.
\label{eqSymmetricAnalyticPhaseExpansion}
\end{equation}
In particular, every term $k^{jm}=E^j$ with $j\geq1$ is absent.
\end{lemma}

\begin{proof}
Fix an integer $M\geq1$. By Eq.~\eqref{eqGeneralJostExpansionClass}, $\breve L_{\mathcal N}(\omega)$ is a nonzero constant times a power of $q_\omega$, and the logarithm of the remaining factor has the expansion $\log a_0+\sum_{j=1}^{M}b_jq_\omega^j+o(q_\omega^M)$, where $b_j$ are constant coefficients. For $E=k^m>0$, the boundary values of the complex momentum are $q_{E+\mathrm{i}0}=k$ and $q_{E-\mathrm{i}0}=\mu^2k$. The difference of the logarithmic boundary values of $\breve L_{\mathcal N}(\omega)$ across the real line is independent of $E$ and determines $\delta(0^+)$. Substituting the remaining terms into Eq.~\eqref{eqPhaseFromJostBoundaryValues} gives
\begin{equation}
2\mathrm{i}\bigl[\delta(E)-\delta(0^+)\bigr]
=
\sum_{j=1}^{M}
b_j\bigl(\mu^{2j}-1\bigr)k^j
+o\!\left(k^M\right).
\label{eqPhaseMonodromySelectionRule}
\end{equation}
Since $\mu^{2j}=e^{2\pi \mathrm{i} j/m}$, the coefficient in Eq.~\eqref{eqPhaseMonodromySelectionRule} vanishes whenever $j\in m\mathbb Z$. Since $M$ is arbitrary, Eq.~\eqref{eqSymmetricAnalyticPhaseExpansion} follows.
\end{proof}

For $\epsilon(k)=k^m$ with odd $m$, the two half-plane coefficient sequences in Eq.~\eqref{eqAnalyticAntisymmetricJostExpansion} need not coincide. Their boundary values are therefore not obtained from a common power series by a fixed rotation of $q_\omega$, and no analogous selection rule excludes positive integer powers divisible by $m$. Particular powers may still be absent because of the interaction structure or fine-tuning.

\section{Leading low-energy expansion for contact interactions}
\label{secContactLowEnergyExpansion}

In this section, we formally prove the leading low-energy expansions stated in the End Matter for contact interactions with analytic monomial dispersions $\epsilon(k)=k^m$ (integer $m>1$). Even and odd $m$ correspond to symmetric and antisymmetric dispersions, respectively. We retain the momentum convention $k\equiv |E|^{1/m}$ introduced in Sec.~\ref{secGeneralLowEnergyExpansion}. The contact interactions are described by the form factors
\begin{equation}
|v(k)\rangle
=
\bigl(k^{n_1},k^{n_2},\ldots,k^{n_N}\bigr)^T,
\qquad
0\leq n_1<n_2<\cdots<n_N,
\label{eqContactFormFactors}
\end{equation}
paired with an invertible coupling matrix $\bm U$.

We restrict our discussion to interior interactions, which requires $n_1<\bar{n}$ and  that no exponent $n_i$ equals $\bar{n}=(m-1)/2$. For odd $m$, this explicitly excludes $n_i = \bar{n}$. For even $m$, $\bar{n}$ is a half-integer, meaning this condition is trivially satisfied and all interactions are inherently interior.

We inherit all results on interior interactions from earlier sections. For contact interactions, the exponents $n_i<\bar n$ and the corresponding vectors $|u_i\rangle$ defined in Lemma~1 are precisely the exponents and associated canonical basis vectors in Eq.~\eqref{eqContactFormFactors}. As before, $D$ denotes the number of exponents satisfying $n_i<\bar n$.

For $i,j\leq D$, the projected-resolvent integral $K_{ij}(\omega)$ is ultraviolet (UV) convergent and equals the corresponding infrared model function without requiring UV regularization.  Outside this principal block, $K_{ij}(\omega)$ with $n_i+n_j\geq m-1$ require UV regularization, which we fix to the form introduced in Eq.~\eqref{eqRegularizedInfraredDecomposition}. Lemma~\ref{lemmaLowEnergyIntegralExpansion} then provides the decomposition
\begin{equation}
\bm J(\omega)
=
\bm J_{\mathrm{IR}}(\omega)
+
\bm J_{\mathrm{an}}(\omega).
\label{eqContactJEntryAnalyticPlusIR}
\end{equation}
Here, $\bm J_{\mathrm{IR}}(\omega)$ is the infrared component, defined piecewise by the model functions (see Def.~\ref{defimodelfunctions}):
\begin{equation}
[\bm J_{\mathrm{IR}}(\omega)]_{ij}
=
\begin{cases}
\breve L_{m,n_i+n_j}(\omega), & m\text{ even},\\
L_{m,n_i+n_j}(\omega), & m\text{ odd}.
\end{cases}
\label{eqContactUnifiedModelFunction}
\end{equation}
The residual matrix $\bm J_{\mathrm{an}}(\omega)$ absorbs all regularization-dependent terms and remains  analytic at $\omega=0$.

The remainder of this section is organized as follows. In Sec.~\ref{subsubNewNotationModelFunctions}, we introduce a unified notation for the infrared model functions. In Sec.~\ref{subsecContactExpansionStatement}, we present the asymptotic expansion of the Jost function and immediately extract the leading-order scattering phase. Finally, the rigorous proof of the Jost-function expansion is provided in Sec.~\ref{subsecProofContactFirstJostCorrection}.

\subsection{New notations for the infrared model functions \label{subsubNewNotationModelFunctions}}
Before proceeding, we introduce an alternative notation for the infrared model functions to render our analysis more transparent.

To streamline the analysis, we introduce an alternative representation for the infrared model functions $L_{m,l}(\omega)$ for odd integers $m$. In the original formulation [Eq.~\eqref{eqkappa} and Def.~\ref{defimodelfunctions}], the jump across the real axis is governed by the ratio $\kappa_{m,l}(-)/\kappa_{m,l}(+)=(-1)^l\mu^{-(l+1)}$, which depends explicitly on $l$. Because the full low-energy expansion involves sums over products of multiple model functions with different values of $l$, tracking these $l$-dependent phase factors across the upper and lower half-planes is cumbersome.

We can absorb this $l$-dependence by introducing the modified complex momentum
\begin{equation}
\widetilde q_\omega
\equiv
\begin{cases}
q_\omega, & \operatorname{Im}\omega>0,\\
-\mu^{-1}q_\omega, & \operatorname{Im}\omega<0,
\end{cases}
\qquad
\mu=e^{\mathrm{i}\pi/m}.
\label{eqModifiedQOddM}
\end{equation}
This allows us to recast the model function as
\begin{equation}
L_{m,l}(\omega)=-\operatorname{sgn}(\operatorname{Im}\omega)  \dfrac{2\pi \mathrm{i}}{m}\kappa_{m,l}(+)
\widetilde{q}_\omega^{\,l+1-m}.
\label{eqNewLmlRepresentation}
\end{equation}
In this representation, the difference between the two half-planes is entirely captured by the sign factor $\operatorname{sgn}(\operatorname{Im}\omega)$ and the $l$-independent variable $\widetilde q_\omega$, eliminating any explicit $l$-dependent phases and significantly simplifying the expansion. We emphasize that this reformulation applies specifically to odd $m$; for more general antisymmetric dispersions, this exact scheme may not directly apply, and the model functions may require a different parametrization.

To further unify the low-energy expansions for both even and odd $m$, we introduce a universal complex momentum $x_\omega$, a sign function $\chi(\omega)$, and frequency-independent coefficients $\gamma_{m,l}$:
\begin{equation}
x_\omega
\equiv
\begin{cases}
q_\omega, & m\text{ even},\\
\widetilde q_\omega, & m\text{ odd},
\end{cases}
\qquad
\chi(\omega)
\equiv
\begin{cases}
1, & m\text{ even},\\
-\operatorname{sgn}(\operatorname{Im}\omega), & m\text{ odd},
\end{cases}
\label{eqContactUnifiedXChi}
\end{equation}
and
\begin{equation}
\gamma_{m,l}
\equiv
\begin{cases}
-\dfrac{2\pi \mathrm{i}}{m}\breve\kappa_{m,l},
& m\text{ even, } l\text{ even},\\
0,
& m\text{ even, } l\text{ odd},\\
\dfrac{2\pi \mathrm{i}}{m}\kappa_{m,l}(+),
& m\text{ odd},
\end{cases}
\label{eqContactUnifiedIRCoefficient}
\end{equation}
where $\kappa_{m,l}(+)$ and $\breve\kappa_{m,l}$ are defined in Eqs.~\eqref{eqkappa} and \eqref{eqkappasym}. 

With these definitions, the infrared model functions for analytic dispersions $\epsilon(k)=k^m$ assume a single, universal form for all integer $m$. The piecewise definition of $[\bm J_{\mathrm{IR}}(\omega)]_{ij}$ in Eq.~\eqref{eqContactUnifiedModelFunction} collapses to the compact expression
\begin{equation}
[\bm J_{\mathrm{IR}}(\omega)]_{ij}
=
\chi(\omega)\,
\gamma_{m,n_i+n_j}\,
x_\omega^{n_i+n_j+1-m}.
\label{eqContactUnifiedIRTerm}
\end{equation}

\subsection{Reduction to the Jost-function expansion}
\label{subsecContactExpansionStatement}

The leading-order expansion of the low-energy scattering phase is entirely determined by the first nonanalytic correction to the Jost function. Below, we introduce the exponents governing this correction, state the asymptotic expansion of the Jost function (Lemma~\ref{lemmaContactFirstJostCorrection}), and derive the resulting phase shift as a corollary. The proof of the lemma itself is deferred to Sec.~\ref{subsecProofContactFirstJostCorrection}.

We define the exponent
\begin{equation}
\ell_<
\equiv
m-1-2n_D,
\label{eqContactLless}
\end{equation}
and, for $D<N$, we define
\begin{equation}
\ell_>
\equiv
2n_{D+1}-m+1.
\label{eqContactLgreater}
\end{equation}
Both exponents are strictly positive integers when defined. Crucially, these exponents track the leading nonanalytic momentum dependencies in different sectors of the $J$-matrix: $\ell_<$ corresponds to the weakest infrared divergence within the leading principal block $\bm{J}^{(D)}(\omega)$ (specifically originating from $J_{DD}(\omega)$), while $\ell_>$ represents the lowest-order nonanalytic correction in the infrared-finite sector (specifically from $J_{D+1,D+1}(\omega)$). 

The actual power of the first nonanalytic correction to both the Jost function and the scattering phase is governed by the dominant of these scales, denoted $\ell_{\mathrm{cut}}$:
\begin{equation}
\ell_{\mathrm{cut}}
\equiv
\begin{cases}
\min(\ell_<,\ell_>), & 0<D<N,\\
\ell_<, & D=N.
\end{cases}
\label{eqContactLcut}
\end{equation}
Note that each candidate exponent in Eq.~\eqref{eqContactLcut} has a parity opposite to $m$; consequently, $\ell_{\mathrm{cut}}$ always possesses a parity opposite to $m$.

\begin{lemma}[First nonanalytic Jost correction for contact interactions]
\label{lemmaContactFirstJostCorrection}
Assume that the contact interaction is interior $(n_1<\bar{n}, \bar{n}\notin \mathcal{N})$ and noncritical. 

For even $m$, there exist constants $A\neq 0$ and $B$ such that
\begin{equation}
J(\omega)
=
A\,\breve L_{\mathcal N}(\omega) 
\left[
1+Bq_\omega^{\ell_{\mathrm{cut}}}
+o\!\left(q_\omega^{\ell_{\mathrm{cut}}}\right)
\right].
\label{eqContactJostExpansionEvenM}
\end{equation}

For odd $m$, there exist constants $A\neq 0$ and $B^\pm$ such that
\begin{equation}
J(\omega)
=
A\,L_{\mathcal N}(\omega) 
\left[
1+B^\pm\widetilde q_\omega^{\ell_{\mathrm{cut}}}
+o\!\left(\widetilde q_\omega^{\ell_{\mathrm{cut}}}\right)
\right],
\qquad
\pm\operatorname{Im}\omega>0.
\label{eqContactJostExpansionOddM}
\end{equation}
Here, $\breve L_{\mathcal N}(\omega)$ and $L_{\mathcal N}(\omega)$ are defined in Eqs.~\eqref{eqSymLset} and \eqref{eqLsetAnti}, respectively.

In general, the odd-$m$ coefficients satisfy  
\begin{equation}
B^-=-B^+.
\label{eqOddMBRelationNonresonant}
\end{equation}
However, if the system satisfies the resonance condition
\begin{equation}
0<D<N,
\qquad
n_D+n_{D+1}=m-1,
\label{eqOddMResonanceCondition}
\end{equation}
then the sum $B^++B^-$ is generically nonzero.
\end{lemma}

While the coefficient at order $\ell_{\mathrm{cut}}$ is generically nonzero, it may vanish under additional fine-tuning. The formal proof of Lemma~\ref{lemmaContactFirstJostCorrection} is deferred to Sec.~\ref{subsecProofContactFirstJostCorrection}. From this lemma, we immediately obtain the corresponding leading-order expansion for the scattering phase.

\begin{corollary}[Leading low-energy phase expansion]
\label{corollaryContactPhaseCorrection}
Assume that the contact interaction is interior and noncritical. For even $m$, the low-energy phase shift satisfies
\begin{equation}
\delta(E)-\delta(0^+)
=
-\alpha k^{\ell_{\mathrm{cut}}}
+o\!\left(k^{\ell_{\mathrm{cut}}}\right).
\label{eqContactPhaseEvenM}
\end{equation}
For odd $m$, the corresponding expansions as $E\to0^\pm$ take the form
\begin{equation}
\delta(E)-\delta(0^\pm)
=
-\alpha^\pm k^{\ell_{\mathrm{cut}}}
+o\!\left(k^{\ell_{\mathrm{cut}}}\right).
\label{eqContactPhaseOddM}
\end{equation}
In general, the odd-$m$ coefficients coincide ($\alpha^+=\alpha^-$). However, if the resonance condition~\eqref{eqOddMResonanceCondition} is satisfied, $\alpha^+$ and $\alpha^-$ are generically distinct. Furthermore, for a Hermitian interaction, all coefficients ($\alpha$ and $\alpha^\pm$) are real.
\end{corollary}

\begin{proof}

We first consider even $m$. For an energy $E=k^m>0$, the complex momentum at the branch cut satisfies
\begin{equation}
q_{E+\mathrm{i}0}=k,
\qquad
q_{E-\mathrm{i}0}=\mu^2k.
\label{eqContactQEvenMApproach}
\end{equation}
By substituting the expansion from Lemma~\ref{lemmaContactFirstJostCorrection} into Eq.~\eqref{eqPhaseFromJostBoundaryValues}, and identifying the zero-energy limit as $2\mathrm{i}\delta(0^+)= \log(\breve L_{\mathcal N}(E-\mathrm{i}0))-\log(\breve L_{\mathcal N}(E+\mathrm{i}0))$, we can expand the logarithm via $\log(1+z)=z+o(z)$ to obtain
\begin{equation}
2\mathrm{i}\bigl[\delta(E)-\delta(0^+)\bigr]
=
B\left(\mu^{2\ell_{\mathrm{cut}}}-1\right)
k^{\ell_{\mathrm{cut}}}
+o\!\left(k^{\ell_{\mathrm{cut}}}\right).
\label{eqContactPhaseDerivationEvenM}
\end{equation}
This establishes Eq.~\eqref{eqContactPhaseEvenM} with 
$
\alpha
=
\frac{B\left(1-\mu^{2\ell_{\mathrm{cut}}}\right)}{2\mathrm{i}}
$. 
Furthermore, a Hermitian system  satisfies $J(E-\mathrm{i}0)=J(E+\mathrm{i}0)^*$. Through Eq.~\eqref{eqPhaseFromJostBoundaryValues}, this symmetry ensures a  real scattering phase $\delta(E)$, which in turn dictates that $\alpha\in\mathbb R$.

For odd $m$ and a given energy $E$, the modified complex momentum $\widetilde q_\omega$ [Eq.~\eqref{eqModifiedQOddM}] yields the upper-half-plane limits 
\begin{equation}
\widetilde q_{E+\mathrm{i}0}
=
\begin{cases}
k, & E>0,\\
\mu k, & E<0,
\end{cases}
\label{eqModifiedQUpperBoundary}
\end{equation}
and the corresponding lower-half-plane limits
\begin{equation}
\widetilde q_{E-\mathrm{i}0}
=
\begin{cases}
-\mu k, & E>0,\\
-k, & E<0.
\end{cases}
\label{eqModifiedQLowerBoundary}
\end{equation}
For either sign of $E$, the ratio of the $L_{\mathcal N}$ factors is independent of $|E|$ and determines the zero-energy phase $2\mathrm{i}\delta(0^\pm) = \log[L_{\mathcal N}(E-\mathrm{i}0)] - \log[L_{\mathcal N}(E+\mathrm{i}0)]$. Noting that $\ell_{\mathrm{cut}}$ is always even for odd $m$, we substitute Eq.~\eqref{eqContactJostExpansionOddM} into Eq.~\eqref{eqPhaseFromJostBoundaryValues} and expand to leading order via $\log(1+z)=z+o(z)$. This directly recovers Eq.~\eqref{eqContactPhaseOddM} with the coefficients $\alpha^+ = (B^+ - \mu^{\ell_{\mathrm{cut}}} B^-)/(2\mathrm{i})$ for $E \to 0^+$ and $\alpha^- = (\mu^{\ell_{\mathrm{cut}}} B^+ - B^-)/(2\mathrm{i})$ for $E \to 0^-$.

When the resonance condition~\eqref{eqOddMResonanceCondition} does not hold, substituting the symmetry $B^- = -B^+$ [Eq.~\eqref{eqOddMBRelationNonresonant}] into these expressions yields equal coefficients: $\alpha^+ = \alpha^- = B^+(1 + \mu^{\ell_{\mathrm{cut}}})/(2\mathrm{i})$. Conversely, under the resonance condition~\eqref{eqOddMResonanceCondition}, the difference between the coefficients then evaluates to $\alpha^+ - \alpha^- = (1-\mu^{\ell_{\mathrm{cut}}})(B^+ + B^-)/(2\mathrm{i})$. The first factor is   nonzero  because $0<\ell_{\mathrm{cut}}<m$, while the second is generically nonzero by Lemma~\ref{lemmaContactFirstJostCorrection}. Thus, $\alpha^+$ and $\alpha^-$ are generically distinct.
Finally, just as for even $m$, a Hermitian system necessarily satisfies $J(E-\mathrm{i}0)=J(E+\mathrm{i}0)^*$, thereby dictating that $\alpha^\pm \in \mathbb{R}$.
\end{proof}

\subsection{Proof of Lemma~\ref{lemmaContactFirstJostCorrection}}
\label{subsecProofContactFirstJostCorrection}

We determine the first nonanalytic correction to $J(\omega)=\det[\bm J(\omega)]$ from the corresponding expansion of the J-matrix.

As established in Secs.~\ref{SecDlessNinterior} and \ref{SecSymDlessN} for interior interactions with antisymmetric and symmetric dispersions, respectively, the block $\bm J^{(D)}(\omega)$ is invertible for sufficiently small nonzero $\omega$. Consequently, for $D<N$, we can invoke the block-determinant factorizations from Eqs.~\eqref{eqJproductAnti} and \eqref{eqJostProduct} to write 
\eq{
J(\omega)
&=
\det[\bm J^{(D)}(\omega)]
\det[\bm J_{\mathrm{red}}(\omega)], 
\label{eqContactSchurDetFactorization} \\
\bm J_{\mathrm{red}}(\omega)
&\equiv
\bm J^\perp(\omega)-
\bm C(\omega)
[\bm J^{(D)}(\omega)]^{-1}
\bm R(\omega). \label{eqContactReducedJDefinition}
}
where $\bm{C}(\omega)$ and $\bm{R}(\omega)$ are defined in the block form of the $J$-matrix
\begin{equation}
\bm J(\omega)
=
\begin{pmatrix}
\bm J^{(D)}(\omega) & \bm R(\omega)\\
\bm C(\omega) & \bm J^\perp(\omega)
\end{pmatrix}.
\label{eqContactBlockDecomposition}
\end{equation}

For $D=N$, the infrared-finite subspace is absent,   $J(\omega)=\det[\bm J^{(D)}(\omega)]$.

In what follows, we first derive the leading low-energy expansion of the leading principal block $\bm J^{(D)}(\omega)$. For $D<N$, we subsequently analyze the leading expansion of the reduced matrix $\bm J_{\mathrm{red}}(\omega)$, combining these two results to obtain the  leading expansion of the Jost function.

\subsubsection{Expansion of the infrared-divergent block}
\label{subsubsecContactDivergentBlock}
Assuming $D>0$, our earlier derivations leading to Eqs.~\eqref{eqdet_J_D_asymptotic} and \eqref{eqsym-mixed-Jost-final} establish that the determinant of the infrared component $\bm{J}_{\mathrm{IR}}^{(D)}(\omega)$ is given  by
\begin{equation}
\det[\bm J_{\mathrm{IR}}^{(D)}(\omega)]
=
\begin{cases}
\breve Z(\mathcal N_{\mathrm e})
\breve Z(\mathcal N_{\mathrm o})
\breve L_{\mathcal N}(\omega),
& m\ \text{even},\\
Z(\mathcal N)L_{\mathcal N}(\omega),
& m\ \text{odd}.
\end{cases}
\label{eqContactIRDetCoefficient}
\end{equation}
Compared to Eqs.~\eqref{eqdet_J_D_asymptotic} and \eqref{eqsym-mixed-Jost-final}, two  distinctions arise for the contact interactions considered here. First, the asymptotic similarity ($\sim$) is upgraded to an equality ($=$). This occurs because the matrix elements of $\bm J_{\mathrm{IR}}^{(D)}(\omega)$ for our contact interactions considered are equal to the infrared model functions, rather than  merely approaching them asymptotically. Second, the coefficients $c_i$ describing the projections of the form factors onto the basis states $|u_i\rangle$ satisfy $c_i=1$ for all $i$.

Our objective is to expand the full determinant $\det[\bm J^{(D)}(\omega)]=\det[\bm J_{\mathrm{IR}}^{(D)}(\omega)+\bm J_{\mathrm{an}}^{(D)}(\omega)]$ about its leading infrared contribution $\det[\bm J_{\mathrm{IR}}^{(D)}(\omega)]$.
To facilitate this expansion, we isolate the infrared divergences by introducing the diagonal scaling matrix
\begin{equation}
\bm P(\omega)
\equiv
\operatorname{diag}_{1\leq i\leq D}
\left(x_\omega^{n_i-\bar{n}}\right).
\label{eqContactTMatrices}
\end{equation}
This allows us to factorize both the infrared and analytic components as
\begin{align}
\bm J_{\mathrm{IR}}^{(D)}(\omega)
&=
\chi(\omega)
\bm P(\omega)
\bar{\bm J}^{(D)}
\bm P(\omega),
\label{eqdetJIRD} \\
\bm J_{\mathrm{an}}^{(D)}(\omega)
&\equiv
\chi(\omega) 
\bm P(\omega)
\bar{\bm E}^{(D)}(\omega) 
\bm P(\omega),
\label{eqContactScaledPerturbationLeading}
\end{align}
where the constant coefficient matrix is defined by $[\bar{\bm J}^{(D)}]_{ij} \equiv \gamma_{m,n_i+n_j}$. Because $n_D$ is the largest exponent in the infrared-divergent block, the lowest power of $x_\omega$ in the rescaled analytic correction $\bar{\bm E}^{(D)}(\omega)$ is localized entirely within its $(D,D)$ entry. Its leading-order behavior is therefore
\begin{equation}
\bar{\bm E}^{(D)}(\omega)
=
\chi(\omega)^{-1}
[\bm J_{\mathrm{an}}^{(D)}(0)]_{DD}\,
x_\omega^{\ell_<}
|u_D\rangle\langle u_D|
+o\!\left(x_\omega^{\ell_<}\right).
\label{eqContactScaledPerturbationLeading0}
\end{equation}

Crucially, the constant matrix $\bar{\bm J}^{(D)}$ is   invertible. For odd $m$, this follows directly from Lemma~\ref{eqlemmadetZ}. For even $m$, grouping the indices by parity renders $\bar{\bm J}^{(D)}$ block-diagonal, and Lemma~\ref{lemmaSymDetZ} guarantees the invertibility of each nonempty block. We can thus rewrite the full matrix as
\begin{equation}
\bm J^{(D)}(\omega)
=
\chi(\omega)\bm P(\omega)\bar{\bm J}^{(D)}
\left[
\mathbb{1}
+
[\bar{\bm J}^{(D)}]^{-1}\bar{\bm E}^{(D)}(\omega)
\right]
\bm P(\omega),
\label{eqContactFullMatrixFactorized}
\end{equation}
where $\mathbb{1}$ is the identity matrix. Taking the determinant of Eq.~\eqref{eqContactFullMatrixFactorized} and expanding $\det[\mathbb{1} + \bm M] = 1 + \operatorname{Tr}(\bm M) + O(\|\bm M\|^2)$ using the localized leading term from Eq.~\eqref{eqContactScaledPerturbationLeading0}, we obtain
\begin{equation}
\det[\bm J^{(D)}(\omega)]
=
\det[\bm J_{\mathrm{IR}}^{(D)}(\omega)]
\left[
1+
\chi(\omega)b_D x_\omega^{\ell_<}
+o\!\left(x_\omega^{\ell_<}\right)
\right],
\qquad
b_D
\equiv
[(\bar{\bm J}^{(D)})^{-1}]_{DD}
[\bm J_{\mathrm{an}}^{(D)}(0)]_{DD}.
\label{eqContactIRDetFirstCorrection}
\end{equation}

To determine whether $x_\omega^{\ell_<}$ is generically the lowest-order term, we apply Lemmas~\ref{eqlemmadetZ} and \ref{lemmaSymDetZ} to the minor matrix $\bar{\bm J}^{(D)}_{\cancel{DD}}$ obtained by deleting the $D^{\text{th}}$ row and column. These lemmas dictate that $\det[\bar{\bm J}^{(D)}_{\cancel{DD}}]\neq0$. By Cramer's rule, the corresponding inverse element is nonzero
\begin{equation}
[(\bar{\bm J}^{(D)})^{-1}]_{DD} 
=
\frac{\det[\bar{\bm J}^{(D)}_{\cancel{DD}}]}{\det[\bar{\bm J}^{(D)}]}\neq 0.
\label{eqContactInverseElementNonzero}
\end{equation}
Consequently, the coefficient $b_D=0$ if and only if $[\bm J_{\mathrm{an}}^{(D)}(0)]_{DD}= 0$. Therefore, $b_D$ is generically nonzero and vanishes only under fine-tuned choices of the coupling matrix $\bm U$.

The upcoming analysis of $\det[\bm{J}_{\mathrm{red}}(\omega)]$ requires the inverse matrix $[\bm J^{(D)}(\omega)]^{-1}$. To isolate its asymptotic behavior, we define the auxiliary matrix
\begin{equation}
\bar{\bm G}^{(D)}(\omega)
\equiv
\left[
\bar{\bm J}^{(D)}
+
\bar{\bm E}^{(D)}(\omega)
\right]^{-1}
=
(\bar{\bm J}^{(D)})^{-1}
+O\!\left(x_\omega^{\ell_<}\right).
\label{eqContactGLess}
\end{equation}
Inverting the factorized form of $\bm J^{(D)}(\omega)$ then yields the relation
\begin{equation}
[\bm J^{(D)}(\omega)]^{-1}
=
\chi(\omega)^{-1}
\bm P(\omega)^{-1}
\bar{\bm G}^{(D)}(\omega)
\bm P(\omega)^{-1}.
\label{eqContactIRInverseExpansion}
\end{equation}

\subsubsection{Expansion of the reduced matrix}
\label{subsubsecContactReducedMatrix}

Now suppose $0<D<N$. As established previously, the reduced matrix approaches the constant zero-energy limit $\bm{J}_{\mathrm{red}}(\omega)\to \bm{J}^\perp(0)$ as $\omega\to 0$. To systematically expand $\det[\bm{J}_{\mathrm{red}}(\omega)]$ around this limit, we introduce the deviation matrix $\bm\Delta_{\mathrm{red}}(\omega) \equiv \bm J_{\mathrm{red}}(\omega) - \bm J^{\perp}(0)$. Using Eq.~\eqref{eqContactReducedJDefinition}, this takes the explicit form
\begin{equation}
\bm\Delta_{\mathrm{red}}(\omega)
=
\bm{J}^\perp(\omega)-\bm{J}^\perp(0) 
-
\bm C(\omega)[\bm J^{(D)}(\omega)]^{-1}\bm R(\omega).
\label{eqContactDeltaRedDef}
\end{equation}

By the assumption of noncriticality, the constant matrix $\bm J^{\perp}(0)$ is strictly invertible (see Lemma~\ref{lemma_inter_crit} for odd $m$ and Lemma~\ref{lemmaCriticalSym} for even $m$). We can therefore expand the determinant directly via the trace:
\begin{equation}
\det[\bm J_{\mathrm{red}}(\omega)]
=
\det[\bm J^{\perp}(0)]
\left[
1+
\operatorname{Tr}\!\left[
[\bm J^{\perp}(0)]^{-1}
\bm\Delta_{\mathrm{red}}(\omega)
\right]
+O\!\left(\|\bm\Delta_{\mathrm{red}}(\omega)\|^2\right)
\right].
\label{eqContactReducedDetExpansion}
\end{equation}
Eq.~\eqref{eqContactReducedDetExpansion} dictates that the leading correction to $\det[\bm J_{\mathrm{red}}(\omega)]$ is governed by the lowest-order terms in the matrix elements of $\bm\Delta_{\mathrm{red}}(\omega)$. In what follows, we derive the explicit low-energy expansion of this deviation matrix.

In what follows, we adopt the convention of using the indices $i,j$ ($1\leq i,j\leq D$) for the infrared-divergent leading principal block and the indices $r,s$ ($D<r,s\leq N$) for the infrared-finite sector. To isolate the infrared behavior in the remaining blocks, we define the scaling matrix
\begin{equation}
\bm Q(\omega)
\equiv
\operatorname{diag}_{D<r\leq N}
\left(x_\omega^{n_r-\bar{n}}\right).
\label{eqContactPGreater}
\end{equation}
alongside the constant infrared coefficient matrices
\begin{equation}
\begin{aligned}
[\bar{\bm J}^{\perp}]_{rs}
&\equiv
\gamma_{m,n_r+n_s},\\
[\bar{\bm R}]_{is}
&\equiv
\gamma_{m,n_i+n_s},\\
[\bar{\bm C}]_{r j}
&\equiv
\gamma_{m,n_r+n_j}.
\end{aligned}
\label{eqContactTMatricesGreater}
\end{equation}

Using the form of $\bm{J}_{\rm IR}(\omega)$   in Eq.~\eqref{eqContactUnifiedIRTerm} alongside Eqs.~\eqref{eqContactPGreater} and \eqref{eqContactTMatricesGreater},  we can decompose the $J$-matrix blocks outside   $\bm{J}^{(D)}$ 
into an analytic part and the corresponding infrared components:

\begin{subequations}
\label{eqJperpRCexp}
\begin{align}
\bm J^{\perp}(\omega)
&=
\bm J_{\mathrm{an}}^{\perp}(\omega)
+
\chi(\omega)\bm Q(\omega)
\bar{\bm J}^{\perp}
\bm Q(\omega),
\label{eqContactFiniteBlockScaling}\\
\bm R(\omega)
&=
\bm R_{\mathrm{an}}(\omega)
+
\chi(\omega)\bm P(\omega)
\bar{\bm R}
\bm Q(\omega),
\label{eqContactRBlockScaling}\\
\bm C(\omega)
&=
\bm C_{\mathrm{an}}(\omega)
+
\chi(\omega)\bm Q(\omega)
\bar{\bm C}
\bm P(\omega).
\label{eqContactCBlockScaling}
\end{align}
\label{eqContactBlockScaling}
\end{subequations}
In these expressions, the first terms [$\bm J_{\mathrm{an}}^{\perp}(\omega)$, $\bm R_{\mathrm{an}}(\omega)$, and $\bm C_{\mathrm{an}}(\omega)$] correspond to the respective blocks of the analytic matrix $\bm J_{\mathrm{an}}(\omega)$. 
The second terms coincide elementwise with the associated infrared model functions.
Substituting Eqs.~\eqref{eqContactIRInverseExpansion} and \eqref{eqContactBlockScaling} into Eq.~\eqref{eqContactDeltaRedDef} gives

\begin{equation}
\begin{alignedat}{2}
\bm\Delta_{\mathrm{red}}(\omega)&=
\bm J_{\mathrm{an}}^{\perp}(\omega)-\bm J^\perp(0)
&&\qquad \text{(1)}
\\
&-
\chi(\omega)^{-1}
\bm C_{\mathrm{an}}(\omega)
\bm P(\omega)^{-1}
\bar{\bm G}^{(D)}(\omega)
\bm P(\omega)^{-1}
\bm R_{\mathrm{an}}(\omega)
&&\qquad \text{(2)}
\\
&+
\chi(\omega)\bm Q(\omega)
\left[
\bar{\bm J}^{\perp}
-
\bar{\bm C}\bar{\bm G}^{(D)}(\omega)\bar{\bm R}
\right]
\bm Q(\omega)
&&\qquad \text{(3)}
\\
&-
\bm Q(\omega)\bar{\bm C}
\bar{\bm G}^{(D)}(\omega)
\bm P(\omega)^{-1}
\bm R_{\mathrm{an}}(\omega)
&&\qquad \text{(4)}
\\
&-
\bm C_{\mathrm{an}}(\omega)
\bm P(\omega)^{-1}
\bar{\bm G}^{(D)}(\omega)
\bar{\bm R}\bm Q(\omega)
&&\qquad \text{(5)}
\end{alignedat}
\label{eqContactReducedPowerDecomposition}
\end{equation}

Because the matrix $\bm J_{\mathrm{an}}^{\perp}(\omega)$ is an analytic function of $\omega=x^m_{\omega}$, its Taylor expansion dictates that line (1) of Eq.~\eqref{eqContactReducedPowerDecomposition} scales as
\begin{equation}
\bm J_{\mathrm{an}}^{\perp}(\omega)-\bm J^\perp(0)
=
\bm A_1 x_\omega^m+O\!\left(x_\omega^{2m}\right)
\label{eqContactJPerpAnalyticExpansion}
\end{equation}
for some constant matrix $\bm A_1$. 

Meanwhile, lines (2)–(5) of Eq.~\eqref{eqContactReducedPowerDecomposition} consist of matrix products involving the diagonal scaling matrices $\bm P(\omega)^{-1}$ and $\bm Q(\omega)$. Based on their definitions in Eqs.~\eqref{eqContactTMatrices} and \eqref{eqContactPGreater}, the lowest-order nonvanishing terms of these matrices are isolated entirely in their $D^{\text{th}}$ and $(D+1)^{\text{th}}$ diagonal entries, respectively. Their leading-order asymptotic behaviors are therefore given by
\eq{
\bm P(\omega)^{-1} &=  x_\omega^{\ell_</2} 
|u_D\rangle \langle u_D|+o(|x_\omega|^{\ell_</2} ) \\
\bm Q(\omega) &= x_\omega^{\ell_>/2} 
|u_{D+1}\rangle \langle u_{D+1}|+o(|x_\omega|^{\ell_>/2} )
}

All remaining matrix factors in lines (2)--(5) are elementwise \(O(1)\): the constant matrices \(\bar{\bm J}^{\perp}\), \(\bar{\bm R}\), and \(\bar{\bm C}\) defined in Eq.~\eqref{eqContactTMatricesGreater}, and the \(\omega\)-dependent matrices \(\bm R_{\mathrm{an}}(\omega)\), \(\bm C_{\mathrm{an}}(\omega)\), and \(\bar{\bm G}^{(D)}(\omega)\) defined in Eqs.~\eqref{eqContactRBlockScaling}, \eqref{eqContactCBlockScaling}, and \eqref{eqContactGLess}.

Table~\ref{tabContactReducedPowerScaling} summarizes the leading-order asymptotic terms originating from each line of Eq.~\eqref{eqContactReducedPowerDecomposition}. Lines (2) and (3) generate terms of order $\chi(\omega)x_\omega^{\ell_<}$ and $\chi(\omega)x_\omega^{\ell_>}$, respectively. Lines (4) and (5) produce cross terms of order $x_\omega^{(\ell_<+\ell_>)/2}$, which notably carry no phase factor $\chi(\omega)$. Since $\ell_{\rm cut}\leq \ell_< <m$, the analytic contribution from line (1) is always subleading compared to the infrared corrections in lines (2)--(5).

Furthermore, the leading coefficient of line~(3) is
\(a_3=[\bar{\bm J}^{\perp}
-\bar{\bm C}(\bar{\bm J}^{(D)})^{-1}\bar{\bm R}]_{D+1,D+1}\).
If we introduce a coefficient matrix $\bar{\bm{J}}^{(D+1)}$ with elements $\gamma_{m,n_i+n_j}$ (where $1\leq i,j\leq D+1$),
$a_3=\det[\bar{\bm{J}}^{(D+1)}]\det[\bar{\bm{J}}^{(D)}]^{-1}$.
The Cauchy-determinant identities give
\(a_3=0\) only when  \(n_{D+1}\equiv n_i \pmod m\) for some \(i\le D\).
Hence \(a_3\neq0\) whenever \(\ell_>\le\ell_<\).
In contrast, the remaining coefficients \(a_1,a_2,a_4,a_5\) depend on
\(\bm U\) or on the ultraviolet regularization. For even \(m\),  coefficients \(a_4\) and \(a_5\) both vanish when \(n_D\) and \(n_{D+1}\) have opposite parity. In particular, this always occurs under the resonance condition \(n_D+n_{D+1}=m-1\).

 \begin{table}[t]
\centering
\caption{Leading contributions from lines (1)--(5) of Eq.~\eqref{eqContactReducedPowerDecomposition} to the deviation matrix $\bm{\Delta}_{\mathrm{red}}(\omega)\equiv \bm J_{\mathrm{red}}(\omega)-\bm J^{\perp}(0)$. The coefficient $a_3$ is nonzero when $\ell_{>}\leq\ell_{<}$.  The remaining coefficients ($a_1, a_2, a_4, a_5$) are generically nonzero, but can vanish under specific fine-tuning of the UV regularization scheme or the coupling matrix $\bm{U}$. When $m$ is even and $n_D$ and $n_{D+1}$ have opposite parities, $a_4=a_5=0$.}

\label{tabContactReducedPowerScaling}
\setlength{\tabcolsep}{4pt}
\begin{tabular}{@{}c@{\hspace{1.2em}}c@{\hspace{1.2em}}c@{}}
\hline\hline
Line
& Factors of $\bm P^{-1}(\omega)$ and $\bm Q(\omega)$
& Lowest-order term \\
\hline
(1)
& ---
& $a_1\,x_\omega^{m}$ \\
(2)
& $\bm P^{-1}(\omega)\,\cdots\,\bm P^{-1}(\omega)$
& $a_2\,\chi(\omega)x_\omega^{\ell_<}$ \\
(3)
& $\bm Q(\omega)\,\cdots\,\bm Q(\omega)$
& $a_3\,\chi(\omega)x_\omega^{\ell_>}$ \\
(4)
& $\bm Q(\omega)\,\cdots\,\bm P^{-1}(\omega)$
& $a_4\,x_\omega^{(\ell_<+\ell_>)/2}$ \\
(5)
& $\bm P^{-1}(\omega)\,\cdots\,\bm Q(\omega)$
& $a_5\,x_\omega^{(\ell_<+\ell_>)/2}$ \\
\hline\hline
\end{tabular}
\label{TableDeltaLines}
\end{table}

Consequently, evaluating the trace in Eq.~\eqref{eqContactReducedDetExpansion} using the leading-order terms from Table~\ref{tabContactReducedPowerScaling} yields
\begin{equation}
\det[\bm J_{\mathrm{red}}(\omega)]=
\det[\bm J^{\perp}(0)]
\begin{cases}
1+\chi(\omega)a_<x_\omega^{\ell_<}
+o\!\left(x_\omega^{\ell_<}\right),
& \ell_<<\ell_>,\\
1+\chi(\omega) a_> x_\omega^{\ell_>}
+o\!\left(x_\omega^{\ell_>}\right),
& \ell_><\ell_<,\\
1+\left[a_0+\chi(\omega)a'\right]
x_\omega^{\ell_{\mathrm{cut}}}
+o\!\left(x_\omega^{\ell_{\mathrm{cut}}}\right),
& \ell_<=\ell_>
\end{cases}.
\label{eqdetJredexpand}
\end{equation}
where the finite constants $a_<, a_>, a_0,$ and $a'$ depend on the coupling matrix $\bm{U}$ and the UV regularization scheme. Note that the equality $\ell_<=\ell_>$  coincides with the resonance condition $n_D+n_{D+1}=m-1$ introduced in Lemma~\ref{lemmaContactFirstJostCorrection}.

 We now construct the full Jost function by combining the block determinants from Eqs.~\eqref{eqContactIRDetCoefficient}, \eqref{eqContactIRDetFirstCorrection}, and \eqref{eqdetJredexpand}. Writing the result separately for even and odd $m$, we obtain

\begin{equation}
J(\omega)
=
\begin{cases}
A\,\breve L_{\mathcal N}(\omega)
\left[
1+Bq_\omega^{\ell_{\mathrm{cut}}}
+o\!\left(q_\omega^{\ell_{\mathrm{cut}}}\right)
\right],
& m\ \text{even},\\[0.6em]
A\,L_{\mathcal N}(\omega)
\left[
1+\chi(\omega)b\widetilde q_\omega^{\ell_{\mathrm{cut}}}
+o\!\left(\widetilde q_\omega^{\ell_{\mathrm{cut}}}\right)
\right],
& m\ \text{odd (nonresonant)},\\[0.6em]
A\,L_{\mathcal N}(\omega)
\left[
1+\left(b_0+\chi(\omega)b_1\right)
\widetilde q_\omega^{\ell_{\mathrm{cut}}}
+o\!\left(\widetilde q_\omega^{\ell_{\mathrm{cut}}}\right)
\right],
& m\ \text{odd (resonant)},
\end{cases}
\label{eqdetJexpand}
\end{equation}
where $A\neq 0$ and  coefficients $B, b, b_0,$ and $b_1$ are constants dictated by $\bm{U}$ and the UV regularization. Here, ``nonresonant'' denotes cases where either $D=N$ or $n_D+n_{D+1}\neq m-1$, while ``resonant'' refers   to the condition $n_D+n_{D+1}=m-1$ (implying $D<N$).

We can now map these results directly to the claims of Lemma~\ref{lemmaContactFirstJostCorrection}. For even $m$, the first line of Eq.~\eqref{eqdetJexpand} immediately establishes Eq.~\eqref{eqContactJostExpansionEvenM}. For odd $m$, recall from Eq.~\eqref{eqContactUnifiedXChi} that the sign function evaluates to $\chi(\omega)=-1$ in the upper half-plane ($\operatorname{Im}\omega>0$) and $\chi(\omega)=1$ in the lower half-plane ($\operatorname{Im}\omega<0$). In the nonresonant case, the second line of Eq.~\eqref{eqdetJexpand} therefore yields $B^+=-b$ and $B^-={b}$, which  reproduces the nonresonant relation  $B^+ + B^- = 0$ [Eq.~\eqref{eqOddMBRelationNonresonant}]. Conversely, under the resonance condition, the final line of Eq.~\eqref{eqdetJexpand} yields $B^+={b}_0-b_1$ and $B^-={b}_0+b_1$, which shows that $B^++B^-$ is generically nonzero. This completes the proof of Lemma~\ref{lemmaContactFirstJostCorrection}.

\section{Complete Classification of Zero-Energy S-matrix
\label{AppendixTable}}

For integer dispersion exponents $m>1$, we enumerate
the allowed values determining the universal zero-energy
$S$-matrix limits established in Theorems 1 and 2
of the main text.
For antisymmetric dispersion,
Table~\ref{tableAnti} lists all allowed values of
$\eta(\mathcal N)\bmod 2m$ for nonempty interior
exponent sets $\mathcal N$.
For symmetric dispersion,
Table~\ref{tableAppen} lists all allowed values of
$\eta(\mathcal N_{\rm e})\bmod m$ and
$\eta(\mathcal N_{\rm o})\bmod m$,
with each enumeration restricted to a nonempty
parity subset.

\begin{table}
\caption{Complete set of allowed values for $\eta(\mathcal{N}) \bmod 2m$ governing the universal zero-energy S-matrix for antisymmetric dispersions  with integer $m>1$. The table enumerates all values corresponding to interior interactions $\mathcal{N}\neq \varnothing$, $\bar{n}\notin\mathcal{N}$. Here, $\mathbb{Z}_{2m}$ denotes the set of integers $\{0, 1, \dots, 2m-1\}$.}

\label{tableInteger}
\begin{center}
\begin{tabular}{@{}r !{\color{gray!50}\vrule} l @{}}
\toprule
$m$ & \( \eta(\mathcal{N}) \bmod 2m \) \\ 
\midrule
2 & $ \{1\}$ \\
3 & $ \{1\}$ \\
4 & $ \{1, 3, 4\}$ \\
5 & $\{1, 3, 4\}$ \\
6 & $\{1, 3, 4, 5, 6, 8, 9\}$ \\
7 & $ \{1, 3, 4, 5, 6, 8, 9\}$ \\
8 & $\mathbb{Z}_{2m}\setminus\{2, 14\}$ \\
9 & $\mathbb{Z}_{2m}\setminus\{0, 2, 14, 17\}$ \\
$m\geq 10$ & $\mathbb{Z}_{2m}$ \\ 
\bottomrule
\end{tabular}
\end{center}
\label{tableAnti}
\end{table}

\begin{table*}[ht]
\centering
\caption{Complete set of values for $\eta(\mathcal{N}_{\rm e}) \bmod m$ and $\eta(\mathcal{N}_{\rm o}) \bmod m$ governing the universal zero-energy scattering eigenphases for symmetric dispersions. Values are tabulated across integer exponents $m$  for nonempty subsets $\mathcal{N}_{\mathrm{e}}$ (even channel) and $\mathcal{N}_{\mathrm{o}}$ (odd channel), with $\mathbb{Z}_m \equiv \{0, 1, \dots, m-1\}$. For omitted integer $m$ in the table, both sets are $\mathbb{Z}_m$. N/A indicates that no nonempty $\mathcal N_{\rm o}$ is possible for $m=2$.}

\label{tableAppen}

\begin{tabular}{r@{\hspace{1em}}l l !{\color{gray}\vrule} r@{\hspace{1em}}l l}
\toprule
$m$ & \( \eta(\mathcal{N}_{\rm e}) \bmod m \) & \( \eta(\mathcal{N}_{\rm o}) \bmod m \) & $m$ & \( \eta(\mathcal{N}_{\rm e}) \bmod m \) & \( \eta(\mathcal{N}_{\rm o}) \bmod m \) \\ 
\midrule
2   & $\{1\}$                         & N/A                                 & 16         & $\mathbb{Z}_{16}\setminus\{0, 4, 8\}$            & $\mathbb{Z}_{16}\setminus\{0, 8, 12\}$ \\
3   & $\{1\}$                         & $\{0\}$                             & 17         & $\mathbb{Z}_{17}\setminus\{3, 4, 7, 8, 12, 16\}$ & $\mathbb{Z}_{17}\setminus\{0, 6, 13\}$ \\
4   & $\{1\}$                         & $\{3\}$                             & 18         & $\mathbb{Z}_{18}\setminus\{2, 7, 11, 16\}$       & $\{0, 3, 4, 7, 8, 10, 11, 14, 15\}$ \\
5   & $\{0, 1\}$                      & $\{3\}$                             & 19         & $\mathbb{Z}_{19}$                        & $\mathbb{Z}_{19}\setminus\{1, 4, 5, 8, 9, 12, 13, 16\}$ \\
6   & $\{0, 1, 5\}$                   & $\{3\}$                             & 21         & $\mathbb{Z}_{21}\setminus\{4, 8, 12, 16, 20\}$   & $\mathbb{Z}_{21}\setminus\{17\}$ \\
7   & $\{1, 5, 6\}$                   & $\{0, 3\}$                          & 22         & $\mathbb{Z}_{22}\setminus\{2, 11, 20\}$          & $\mathbb{Z}_{22}\setminus\{2, 5, 6, 9, 13, 16, 17, 20\}$ \\
8   & $\{1, 5, 6\}$                   & $\{2, 3, 7\}$                       & 23         & $\mathbb{Z}_{23}$                        & $\mathbb{Z}_{23}\setminus\{1, 4, 5, 8, 12, 16, 20\}$ \\
9   & $\{0, 1, 5, 6\}$                & $\{1, 3, 7\}$                       & 25         & $\mathbb{Z}_{25}\setminus\{4, 8, 12\}$           & $\mathbb{Z}_{25}$ \\
10  & $\mathbb{Z}_{10}\setminus\{2, 3, 7, 8\}$         & $\{0, 3, 7\}$                       & 26         & $\mathbb{Z}_{26}$                        & $\mathbb{Z}_{26}\setminus\{2, 6, 9, 13, 17, 20, 24\}$ \\
11  & $\mathbb{Z}_{11}\setminus\{0, 2, 7, 8\}$         & $\{0, 3, 7, 10\}$                   & 27         & $\mathbb{Z}_{27}$                        & $\mathbb{Z}_{27}\setminus\{4, 8, 12, 16, 20\}$ \\
12  & $\mathbb{Z}_{12}\setminus\{0, 4, 7, 8, 11\}$     & $\mathbb{Z}_{12}\setminus\{0, 1, 4, 5, 8\}$          & 30         & $\mathbb{Z}_{30}$                        & $\mathbb{Z}_{30}\setminus\{2, 13, 17, 28\}$ \\
13  & $\mathbb{Z}_{13}\setminus\{3, 4, 7, 8, 11, 12\}$ & $\mathbb{Z}_{13}\setminus\{0, 2, 4, 6, 9, 12\}$      & 31         & $\mathbb{Z}_{31}$                        & $\mathbb{Z}_{31}\setminus\{4, 8\}$ \\
14  & $\mathbb{Z}_{14}\setminus\{2, 3, 7, 11, 12\}$    & $\{0, 3, 4, 7, 10, 11\}$            & 34         & $\mathbb{Z}_{34}$                        & $\mathbb{Z}_{34}\setminus\{17\}$ \\
15  & $\mathbb{Z}_{15}\setminus\{2, 11\}$              & $\{0, 3, 6, 7, 10, 11, 14\}$        &  $m\geq 35$          &  $\mathbb{Z}_{m}$                               & $\mathbb{Z}_{m}$ \\

\bottomrule
\end{tabular}
\end{table*}

\section{Notation\label{secNotation}}

We summarize the notation used throughout the main text and this Supplemental Material.

\subsection*{Asymptotic notation}

All limits below are taken as $\omega\to0$ outside the continuum spectrum, or within a specified complex half-plane. Let $h_1(\omega)$ and $h_2(\omega)$ be complex-valued functions, with $h_2(\omega)\neq0$ for all sufficiently small nonzero $\omega$ in the domain considered.
\begin{itemize}
    \item $h_1(\omega)\sim h_2(\omega)$ means $h_1(\omega)/h_2(\omega)\to1$.
    \item $h_1(\omega)=\Theta(h_2(\omega))$ means that $c\leq|h_1(\omega)/h_2(\omega)|\leq C$ for some constants $0<c<C<\infty$ and all sufficiently small admissible $\omega$.
    \item $h_1(\omega)=O(h_2(\omega))$ means that $|h_1(\omega)/h_2(\omega)|$ remains bounded.
    \item $h_1(\omega)=o(h_2(\omega))$ means $h_1(\omega)/h_2(\omega)\to0$.
\end{itemize}

\subsection*{Matrices and subspaces}

\begin{itemize}
    \item $\bm K^R$: The frequency-independent emitter-interaction matrix, incorporating coherent and incoherent emitter interactions [Eq.~\eqref{eqWrelations0}].
    \item $\bm U$: The invertible form-factor coupling matrix of the separable potential [Eq.~\eqref{eqWrelations0}].
    \item $\bm W$: The frequency-independent interaction matrix used in the unified $J$-matrix formalism. It equals $\bm K^R$ for emitter scattering and $-\bm U^{-1}$ for separable potentials [Eq.~\eqref{eqWrelations0}].
    \item $\bm K(\omega)$: The $N\times N$ projected resolvent, describing continuum propagation within the emitter or form-factor space [Eq.~\eqref{eqKdefi}].
    \item $\bm J(\omega)$ and $J(\omega)$: The $N\times N$ $J$-matrix $\bm J(\omega)$ is defined in Eq.~\eqref{eqMKrelations}. Its determinant $J(\omega)\equiv\det[\bm J(\omega)]$ is the Jost function, whose boundary values determine $\det[\bm S(E)]$ [Eq.~\eqref{SratioJ}].
    \item $\bm J^{(D)}(\omega)$ and $\bm J^{(D-1)}(\omega)$: The $D\times D$ leading principal block $\bm J^{(D)}(\omega)$ acts on the span of $|u_1\rangle,\ldots,|u_D\rangle$ [Eqs.~\eqref{eqInteriorBlockDecompositionAnti} and \eqref{eqSymJDdef}]. The $(D-1)\times(D-1)$ leading principal block $\bm J^{(D-1)}(\omega)$ excludes the borderline direction $|u_D\rangle$ in the antisymmetric borderline analysis [Eq.~\eqref{eqBorderlineFiniteBlockDecomp}].
    \item $\bm J^\perp(\omega)$: The $(N-D)\times(N-D)$ infrared-finite block of $\bm J(\omega)$, restricted to the orthogonal complement of the span of $|u_1\rangle,\ldots,|u_D\rangle$ [Eqs.~\eqref{eqInteriorBlockDecompositionAnti} and \eqref{eqSymThreeBlockJ}].
    \item $\bm J^{(D-1)\perp}(\omega)$: For antisymmetric borderline interactions, the $(N-D+1)\times(N-D+1)$ complementary block containing the borderline direction $|u_D\rangle$ and the infrared-finite subspace [Eq.~\eqref{eqBorderlineFiniteBlockDecomp}].
    \item $\bm J_{\mathrm{red}}(\omega)$: The Schur complement of $\bm J^{(D)}(\omega)$ for interior interactions and symmetric borderline interactions, or of $\bm J^{(D-1)}(\omega)$ for antisymmetric borderline interactions [Eqs.~\eqref{eqInteriorReducedJAnti}, \eqref{eqSymDlessJred}, and \eqref{eqFiniteSchurBorderline}]. When $D=0$, $\bm J_{\mathrm{red}}(\omega)=\bm J(\omega)$. For antisymmetric borderline interactions, $\bm J_{\mathrm{red}}^+$ and $\bm J_{\mathrm{red}}^-$ denote its upper- and lower-half-plane limits, respectively [Eq.~\eqref{Jredlimitborderline}].
    \item $\bm J^{\mathrm e}(\omega)$ and $\bm J^{\mathrm o}(\omega)$: For symmetric dispersion, $\bm J^{\mathrm e}(\omega)$ is the block associated with even coupling exponents $n_i$, and $\bm J^{\mathrm o}(\omega)$ is the block associated with odd coupling exponents [Eq.~\eqref{eqSymThreeBlockJ}].
    \item $\bm J_{\cancel{ji}}(\omega)$: The submatrix obtained by deleting row $j$ and column $i$ of $\bm J(\omega)$, used in the adjugate formula [Eq.~\eqref{eqSymMixedAdjugate}].
    \item $[\bm J(\omega)^{-1}]_{ij}$: The $(i,j)$ element of the inverse $J$-matrix, whenever the inverse exists [Eq.~\eqref{eqSymMixedAdjugate}].
    \item $\mathbb 1_N$: The identity matrix of dimension $N$ [Eq.~\eqref{eqMKrelations}].
    \item $\sigma_{\min}(\bm A)$: The smallest singular value of a matrix $\bm A$, used in the criticality criterion [Eq.~\eqref{eqCriticalSmin}].
\end{itemize}

\subsection*{Other symbols and parameters}

\begin{itemize}
    \item $E$ and $\omega$: $E$ is the real on-shell scattering energy. The complex spectral parameter $\omega$ is used to analytically continue the resolvent and $J$-matrix.

    \item $m$: The dispersion exponent $m>1$, real in the main scattering theorems and restricted to integers for analytic monomial dispersions $\epsilon(k)=k^m$.
    \item $\bar n$: The borderline coupling exponent $\bar n\equiv(m-1)/2$, which separates algebraically infrared-divergent contributions ($n_i<\bar n$) from higher-order couplings. Equality $n_i=\bar n$ defines the borderline case.
    \item $N$: The number of interaction components, namely emitters or form factors.
    \item $D$: The number of low-energy coupling exponents in $\mathcal N$, so $D=|\mathcal N|\leq N$ (Lemma~1 of the main text).
    \item $\bm S(E)$ and $S(E)$: For symmetric dispersion, $\bm S(E)$ is the $2\times2$ scattering matrix in the even- and odd-parity channel basis [Eq.~\eqref{eqgvT0}]. For antisymmetric dispersion, scattering has a single channel with scalar amplitude $S(E)$ [Eq.~\eqref{SratioJ}].
    \item $|v(k)\rangle$: The $N$-component coupling vector, whose entries $v_a(k)$ are emitter coupling coefficients or separable-potential form factors [Eq.~\eqref{eqKdefi}].
    \item $|u_i\rangle$: The orthonormal coupling directions associated with the distinct low-energy exponents $n_i$ (Lemma~1 of the main text and Sec.~\ref{AppenIntegerSet}).
    \item $V_i(k)$: The projected coupling component $V_i(k)\equiv\langle u_i|v(k)\rangle$. For $i\leq D$, it satisfies $V_i(k)=c_i k^{n_i}+o(k^{n_i})$, with $c_i\neq0$ [Eqs.~\eqref{eqAntiViLeading} and \eqref{eqSymViLeading}].
    \item $|v_\lambda(k)\rangle$: The parity-resolved coupling vector for symmetric dispersion, with $|v_{\mathrm e}(k)\rangle=[|v(k)\rangle+|v(-k)\rangle]/\sqrt2$ and $|v_{\mathrm o}(k)\rangle=[|v(k)\rangle-|v(-k)\rangle]/\sqrt2$ [defined immediately after Eq.~\eqref{eqgvT0}].
    \item $V_{\lambda,i}(k)$: The component $\langle u_i|v_\lambda(k)\rangle$ of the parity-resolved coupling vector, where $\lambda\in\{\mathrm e,\mathrm o\}$ [Eq.~\eqref{eqSymVcomponentDef}].
    \item $\mathcal N$: The set of distinct nonnegative low-energy coupling exponents $\{n_1,\ldots,n_D\}$ with $n_i\leq\bar n$, defined in Lemma~1 of the main text and constructed in Sec.~\ref{AppenIntegerSet}.
    \item $\mathcal N_p$: The subset of $\mathcal N$ containing exponents of parity $p\in\{\mathrm e,\mathrm o\}$ [Eq.~\eqref{eqSymNpDef}].
    \item $\mathcal N^{<}$ and $\mathcal N_p^{<}$: The sets obtained by excluding the borderline exponent: $\mathcal N^{<}\equiv\mathcal N\setminus\{\bar n\}$ and $\mathcal N_p^{<}\equiv\mathcal N_p\setminus\{\bar n\}$ [Eqs.~\eqref{eqNlessBorderline} and \eqref{eqSymNpDef}].
    \item $N_{\mathrm e}$ and $N_{\mathrm o}$: The numbers of even and odd coupling exponents in $\mathcal N$, respectively, with $D=N_{\mathrm e}+N_{\mathrm o}$ (Sec.~\ref{subsecJsym}).
    \item $\eta(\mathcal A)$: The integer sum $\eta(\mathcal A)\equiv\sum_{n\in\mathcal A}(2n+1)$, with $\eta(\varnothing)=0$, determining the universal zero-energy scattering phases [Eq.~\eqref{eqnAdefi} and Theorems 1 and 2 of the main text].
    \item $k_E$: The positive on-shell momentum magnitude $k_E\equiv|E|^{1/m}$. In the low-energy expansion, $k\equiv k_E$ (Secs.~\ref{secGeneralLowEnergyExpansion} and \ref{secContactLowEnergyExpansion}).
    \item $\rho(E)$: The on-shell density of states per momentum branch, $\rho(E)=m^{-1}|E|^{1/m-1}$ [defined immediately after Eq.~\eqref{eqgvT0}].
    \item $q_\omega$: The complex momentum $q_\omega\equiv|\omega|^{1/m}e^{\mathrm i\arg(\omega)/m}$, with $\arg(\omega)\in(0,2\pi)$ as used in Eq.~\eqref{eqLml}.
    \item $\mu$: The phase factor $\mu\equiv e^{\mathrm i\pi/m}$ used in the infrared model functions [Eq.~\eqref{eqkappa}].
    \item $L_{m,l}(\omega)$: The antisymmetric-dispersion infrared model function, initially defined by a convergent integral for $0\leq l<m-1$ and extended to higher orders for analytic odd $m$ [Eqs.~\eqref{eqSLmndefinition0} and \eqref{eqExtendedAntisymmetricModelFunction}].
    \item $\breve L_{m,l}(\omega)$: The symmetric-dispersion infrared model function. Its convergent integral, borderline form, and extension for analytic even $m$ are given in Eqs.~\eqref{eqSLmndefinition1}, \eqref{eqSymmLbrevedefiborder}, and \eqref{eqExtendedSymmetricModelFunction}, respectively.
    \item $L_{\mathcal A}(\omega)$: The product $\prod_{n\in\mathcal A}L_{m,2n}(\omega)$ for an exponent set $\mathcal A$, with $L_{\varnothing}(\omega)=1$ [Eq.~\eqref{eqLsetAnti}].
    \item $\breve L_{\mathcal A}(\omega)$: The corresponding symmetric-dispersion product $\prod_{n\in\mathcal A}\breve L_{m,2n}(\omega)$, with $\breve L_{\varnothing}(\omega)=1$ [Eq.~\eqref{eqSymLset}].
    \item $Z(\mathcal A)$: The positive Cauchy-determinant factor in antisymmetric dispersion [Eq.~\eqref{eqZNoddSM}], entering the leading Jost-function coefficient [Eq.~\eqref{eqdet_J_D_asymptotic}].
    \item $\breve Z(\mathcal A)$: The positive Cauchy-determinant factor for a same-parity exponent set in symmetric dispersion [Eqs.~\eqref{eqSymZdef} and \eqref{eqSymZdet}].
    \item $\delta(E)$ and $\Delta\delta$: The generally complex determinant scattering phase $\delta(E)$ is defined by $\det[\bm S(E)]=e^{2\mathrm i\delta(E)}$. The quantity $\Delta\delta$ is its accumulated change along the continuum, with the zero-energy jump excluded for antisymmetric dispersion [Eqs.~\eqref{eqPhaseFromJostBoundaryValues} and \eqref{eqDeltaPhi}].

\end{itemize}
\bibliographystyle{apsrev4-2}
\begingroup
\makeatletter
\let\SM@savedbiblabel\label
\let\SM@savedbibref\ref
\def\SM@lastbib{LastBibItem}
\renewcommand{\label}[1]{%
  \def\SM@thislabel{#1}%
  \ifx\SM@thislabel\SM@lastbib
     \SM@savedbiblabel{SM-LastBibItem}%
  \else
     \SM@savedbiblabel{#1}%
  \fi
}
\renewcommand{\ref}[1]{%
  \def\SM@thislabel{#1}%
  \ifx\SM@thislabel\SM@lastbib
     \SM@savedbibref{SM-LastBibItem}%
  \else
     \SM@savedbibref{#1}%
  \fi
}
\makeatother
\putbib[library_Oct22]
\endgroup
\end{bibunit}
\end{document}